\documentclass[12pt]{article}

\usepackage{graphicx}
\usepackage{xcolor}
\usepackage{microtype}
\DeclareGraphicsExtensions{.pdf,.png,.jpg}
\usepackage{caption}
\usepackage{subcaption}
\usepackage{float}
\usepackage{needspace}

\usepackage{amsmath}
\usepackage{amssymb}
\usepackage{units}

\makeatletter
\newcommand{\prism@neq}[2]{\ooalign{$#1\m@th=$\cr\hidewidth$#1\m@th/$\hidewidth\cr}}
\renewcommand{\neq}{\mathrel{\mathpalette\prism@neq\relax}}

\makeatother

\newtheorem{theorem}{Theorem}
\newtheorem{proposition}{Proposition}
\newtheorem{corollary}{Corollary}

\newtheorem{remark}{Remark}

\newtheorem{observation}{Observation}

\usepackage[margin=2.7cm]{geometry}

\usepackage[english]{babel}

\newenvironment{proof}[2]{\noindent {\bf #1 #2.}\par}{\hfill\ensuremath{\square}\par\medskip\normalsize\noindent}

 \author{Constantine Sorokin\thanks{University of Glasgow.}\and Alexander Matros\thanks{University of South Carolina.}
}  
 
 \date{2026}

\title{\textbf{Suppression and Empowerment in Contests}}

\begin{document}

\maketitle

\begin{abstract}

We study a tractable two-player contest built on a truncated cubic contest success function. Its defining feature is a strategic-feedback parameter whose sign determines whether a leading player's effort lowers (suppression) or raises (empowerment) the marginal effectiveness of the trailing player's effort; standard lottery contests impose suppression by construction. The benchmark yields closed-form mixed equilibria under complete information and a unique affine Bayesian Nash equilibrium under IID private information. Expected effort is typically single-peaked in the feedback parameter. Uncertainty lowers effort under suppression but raises it under empowerment, and the same asymmetry governs information disclosure: an effort-maximizing designer withholds information under suppression and discloses fully under empowerment. Several familiar conclusions of contest theory turn out to reflect suppressive benchmarks rather than contests as such.

\smallskip
\noindent \textbf{Keywords:} Contest success function, incomplete information, mixed equilibrium, empowerment, suppression, institutions

\end{abstract}

\newpage
\section{Introduction}

Contest theory began with a universal ambition. When Gordon Tullock introduced his lottery model in 1980, he meant it for a setting that ``had much resemblance to a court trial or, indeed, to any other two-party conflict'' \cite[p.~4]{Tullock1980}. The model has carried that ambition ever since. It has become the workhorse benchmark for rent-seeking, lobbying, R\&D races, electoral competition, and many other strategic environments \cite{HillmanSamet1987,Dixit1987,Nitzan1994}; see also \cite{BayeKovenockDeVries1993,Konrad2009,CHK2008}. A large body of work has refined and extended it in many directions.

The questions contest theory is asked to answer have, however, grown. Many contests of interest --- patent races, electoral competition, lobbying, conflict --- are not designed the way, for example, auctions are. They arise within legal, political, and technological environments that evolve naturally, leaving little scope for direct mechanism design. In settings of this kind, information becomes one of the few levers available, from comparative statics under private information to deliberate disclosure policy. There is some irony in this. Auction theory has both mechanism design at its disposal \cite{Myerson1981} and a developed information theory built on top of it \cite{MilgromWeber1982,BergemannEtAl2022}. Contest theory, where direct mechanism control is rarely the relevant margin and information therefore matters most, has neither.

\begingroup
\emergencystretch=3em
\hbadness=5000
The starting point of this paper is that asymmetry. Important progress has been made on incomplete-information contests, beginning with \cite{HurleyShogren1998}. Related work includes \cite{Warneryd2003}, \cite{Wasser2013}, \cite{MaluegYates2004}, and \cite{Fey2008}. More recent contributions clarify uniqueness and monotonicity \cite{EwerhartQuartieri2020} and \cite{ProkopovychYannelis2023}, while other recent contributions develop tractable contest classes in different ways \cite{Kirkegaard2024} and \cite{Kirkegaard2025}. Mixed equilibrium, which is not a peripheral feature but a regime the model enters whenever strategic feedback is strong enough, has likewise proved difficult to characterize in closed form within the lottery family. The picture that emerges from this work is that information and mixing are tractable in particular corners of the lottery family, but a single tractable benchmark that handles both at once --- without collapsing to an extreme limit such as the all-pay case --- has not been readily available.
\par
\endgroup

This paper proposes one. We work with a cubic contest success function: the lowest-order polynomial that is symmetric across players and admits a non-trivial cross-effect between their efforts, with truncation at $0$ and $1$ to keep probabilities well-defined. The specification is deliberately minimal. It retains the standard features one expects from a contest technology --- symmetry, smoothness, continuous winning probabilities --- and the only feature it relaxes is the particular ratio-form normalization that defines the lottery family. The aim is the same modest one that motivates linear demand or quadratic cost in other areas of economics: to provide a transparent benchmark that gives access to questions that more elaborate specifications place out of reach.

The key feature of the benchmark is that strategic interaction can operate in two directions. A leading player's effort may either reduce the marginal effectiveness of the opponent's effort (\emph{suppression}) or increase it (\emph{empowerment}). Standard lottery contests, including the Tullock family, impose suppression globally, as a structural property of the ratio form rather than a consequence of any parameter choice. The cubic CSF relaxes this restriction.

The paper makes four main contributions. First, it introduces a tractable benchmark that permits both suppression and empowerment. Second, it delivers closed-form mixed equilibria under complete information and a unique affine Bayesian Nash equilibrium under IID private information. Third, it shows that uncertainty and information disclosure have opposite effects under suppression and empowerment. Fourth, it identifies an interior degree of suppression that maximizes equilibrium effort and rent dissipation.

These results make the natural questions of contest theory answerable in closed form, in much the same way they are in symmetric independent private-value auctions, but without the heavy machinery.\footnote{In the lottery family, the corresponding analysis of mixed equilibrium has been notoriously difficult; see \cite{BayeKovenockDeVries1994,Ewerhart2015}.} How does effort respond to the strength of strategic interaction? How does it respond to uncertainty? How should an effort-maximizing designer disclose information? The comparative statics flow from a single object: the cross-derivative of the winning probability changes sign with the ranking of efforts, and the parameter that controls it organizes essentially every comparative-static result the model produces. 

Expected effort is single-peaked, rather than monotone, in the strategic-feedback parameter
$a$ under complete information, and typically so under incomplete information; strong suppression with substantial type dispersion can produce a second, dropout-driven peak. Under incomplete information, private uncertainty discourages effort in the suppression range but encourages it in the empowerment range, with continuous transition through neutrality. The same pattern carries over to information design: an effort-maximizing designer prefers no disclosure under suppression and full disclosure under empowerment. These results are organized by whether strategic feedback is suppressive or empowering, a comparison that single-sign benchmarks cannot make.

Two classical applications\footnote{We discuss other potential applications in the Appendix.} make the distinction concrete. In each, the relevant approach is the same: fix a profile at which one player is investing more than the other. Ask whether a marginal increase in the higher-investment player's effort raises or lowers the marginal effectiveness of the lower-investment player's effort. By symmetry of the cross-partial, the mirror question --- whether a marginal increase in the lower-investment player's effort raises or lowers the higher player's marginal effectiveness --- gives the same answer, so the two formulations describe the same property of the contest technology. The answer is a property of the institutional environment.

\paragraph{R\&D and patent races.} Consider two firms investing in research toward a patent, and suppose that the firms are identical except that one is currently investing more than the other. The relevant institutional variable is the appropriability regime: the strength of patent protection, the patent premium \cite{AghionEtAl2005,AroraCeccagnoliCohen2008}, the ease of reverse engineering, and the extent of knowledge spillovers \cite{BloomSchankermanVanReenen2013,FershtmanMarkovich2010}. Under strong appropriability, additional R\&D by the firm that is already ahead translates into more exclusive control over the leading research line and limited spillovers, so the trailing firm's marginal R\&D becomes less effective. That is suppression. Under weak appropriability with substantial spillovers, the same additional R\&D by the leading firm generates knowledge that diffuses through imitation, reverse engineering, licensing, and labor mobility, so the trailing firm's marginal R\&D becomes more, rather than less, effective. That is empowerment. The same model, with the same firms and the same prize, fits both regimes; what changes is the institutional channel through which a lead translates into marginal returns for the other side.

\paragraph{Rent-seeking and political competition.} Now consider two politicians competing for office and spending to influence voters and the intermediaries --- media outlets, donors, brokers --- that shape political competition \cite{Stromberg2004,BesleyPrat2006}; see also \cite{PastinePastine2012,AvisEtAl2022}. Suppose one is currently outspending the other, and ask how a marginal increase in the leading politician's spending changes the marginal effectiveness of the trailing politician's effort. The answer turns on what the intermediaries themselves value. If they value future transfers from whoever ends up in office --- contracts, regulatory favors, protection against competitors --- they have a reason to coordinate around the politician who looks ahead, and additional spending by that side tightens the alignment. The trailing politician's marginal effort becomes less effective. That is suppression. If the same intermediaries instead value current standing with their audience --- viewership, advertising revenue, the credibility of a competitive story --- the logic reverses: a visible lead makes contrast and scrutiny more valuable, and the trailing politician's marginal effort becomes more effective. That is empowerment. Again, the same contest model fits both cases; what differs is the objective that drives the intermediaries' response.

Two features of these examples are worth noting. First, the empowerment and suppression labels are tied to the current ranking of investments rather than to player identities. This distinguishes them from the familiar language of complements and substitutes: in a contest, winning probabilities sum to one, so what looks like complementarity from one side is automatically substitution from the other, and the distinction does not cleanly track institutional environments. Empowerment and suppression do. Second, very different institutional environments --- appropriability regimes and political intermediaries among them --- reach the contest through the same channel: how a local lead translates into the trailing side's marginal return. The cubic benchmark is the place where that common channel becomes a single parameter, with empowerment and suppression on opposite sides of zero.

The contribution is therefore not primarily a new polynomial specification. The contribution is a tractable benchmark in which the direction of strategic interaction itself becomes an economic object.

Our work touches several strands of the existing literature. On incomplete-information contests in the lottery family, early work by Hurley and Shogren characterizes how uncertainty affects effort \cite{HurleyShogren1998}. Related work includes \cite{Warneryd2003} and \cite{Wasser2013}, with \cite{Wasser2013} in particular emphasizing that no general ranking of complete versus private information is available. Later contributions clarify existence, uniqueness, and monotone equilibria. Malueg and Yates study private values \cite{MaluegYates2004}; Fey focuses on equilibrium existence \cite{Fey2008}; Ewerhart proves uniqueness with continuous independent type distributions \cite{Ewerhart2014}; and Ewerhart and Quartieri provide a broader uniqueness analysis \cite{EwerhartQuartieri2020}. Monotone equilibria are studied in \cite{ProkopovychYannelis2023}. The cubic benchmark sits alongside this work: it yields a clean effort-discouraging comparative static under suppression and a symmetric effort-encouraging one under empowerment.

On mixed equilibrium in symmetric contests, \cite{BayeKovenockDeVries1994}, \cite{BayeKovenockDeVries1996}, and \cite{Ewerhart2015} establish the structure that mixing takes in the lottery family, including its all-pay limit; the cubic benchmark, which is not singular at zero effort, complements that picture with closed-form moments and minimal two-point supports.

\begingroup
\emergencystretch=3em
\hbadness=5000
On information design in contests, the closest methodological comparison is with the symmetric IPV persuasion framework of \cite{BergemannEtAl2022} for auctions. The closest contest-side comparisons are work on voluntary disclosure \cite{EwerhartLareida2024}, Bayesian persuasion in contests \cite{ZhangZhou2016}, optimal type disclosure \cite{Serena2022}, disclosure with stochastic prize values \cite{AntsyginaTeteryatnikova2023}, and exogenous information changes in Tullock contests \cite{AicheEtAl2019}. The contribution of the cubic benchmark in this setting is that expected effort depends on a signal only through the variance of posterior-mean costs, which makes the disclosure problem all-or-nothing --- either no disclosure or full disclosure --- and ties its sign to the suppression--empowerment distinction.
\par
\endgroup

Finally, our approach is complementary to a parallel strand that derives contest technologies from primitive axioms or studies broad CSF classes defined by restrictions such as ratio form, difference form, or log-supermodularity.
This includes \cite{Dixit1987}, \cite{Skaperdas1996}, and \cite{CheGale2000}. Recent examples include \cite{Kirkegaard2024,Kirkegaard2025}. Difference-form extensions include \cite{CubelSanchezPages2016,CubelSanchezPages2025}. Rather than characterize the class of CSFs that satisfy a list of properties, we work with one simple specification that delivers a closed-form benchmark for naturally evolving contests.

Three caveats are worth stating before the analysis begins. First, the model is static. The parameter that distinguishes suppression from empowerment is a property of the contest technology and the institutional environment, not of any temporal sequence of moves, and the comparative statics compare environments rather than phases of a dynamic race. Second, the parameter is a one-dimensional reduction of a high-dimensional institutional setting, and it inherits the standard tradeoff of any reduced-form benchmark: it makes very different applications comparable within one model, at the price of not resolving the channels through which institutions reach the contest. Third, the cubic specification leaves the unit interval at sufficiently asymmetric profiles, which is why we truncate the CSF at $0$ and $1$. This is in the same spirit as the difference-form literature \cite{Hirshleifer1989,CheGale2000} and \cite{CubelSanchezPages2016,CubelSanchezPages2025}, and we treat the truncation as a consistency requirement rather than a modeling choice. The analysis identifies the parameter region in which truncation is inactive on the equilibrium path; this is the region in which the benchmark results below apply directly.

The remainder of the paper is organized as follows. Section~2 sets out the cubic benchmark, the truncated CSF, and the hierarchy of equilibrium concepts that the truncation makes natural. Section~3.1 develops the complete-information analysis: pure and mixed equilibria, single-peaked expected effort, and explicit primitive conditions under which the benchmark survives truncation. The main characterizations throughout are stated for the unrestricted polynomial benchmark, where the raw cubic formula is taken at face value: this is the object on which the comparative statics are transparent, and the economically meaningful one wherever truncation is inactive. Survival in the literal truncated contest is then a separate question, which we address result by result rather than through a single blanket condition, because survival is regime-dependent and no single condition governs it. 

Section~3.2 turns to incomplete information: the affine Bayesian Nash equilibrium, the effect of uncertainty, the variance-driven location of the effort-maximizing parameter, and the disclosure result, together with the conditions under which each survives in the literal contest. The structure is asymmetric in an economically meaningful way: under empowerment and at neutrality the benchmark carries over subject only to a support check, while under suppression it survives moderate feedback and gives way only when suppression is strong enough to trigger dropout. Under complete information the effort-maximizing degree of suppression is interior; under incomplete information its location depends on type dispersion and, under strong suppression, on dropout. Section~4 discusses interpretation, the place of empowerment and suppression in applied contest analysis, and broader implications. Proofs and additional applications are collected in the Appendix.

\section{Model}

Our benchmark contest success function starts from the raw cubic expression
\begin{equation}\label{1}
P(x,y) = \frac{1}{2} + (x - y)\bigl(c - b(x + y) + a x y\bigr),
\end{equation}
where $x$ and $y$ are the players' costly investments, $P(x,y)$ is the probability\footnote{Any polynomial CSF must be clipped outside its admissible range. We treat this clipping as a consistency device rather than a modeling choice.} that player $X$ wins, and $1-P(x,y)$ is the probability that player $Y$ wins. We refer to this specification as the \emph{cubic} CSF.

The three terms have distinct economic roles. The constant $c>0$ captures the baseline marginal return to effort near zero. The parameter $b>0$ governs the local decline in the marginal effectiveness of own effort. The interaction parameter $a$ is the channel of interest: it generates strategic feedback between the players' investments, and its sign and magnitude determine the direction and strength of that feedback.

The key cross-derivative property is
\[
\frac{\partial^{2} P}{\partial x\,\partial y} = 2a(x - y),
\]
so the cross-derivative is linear in the effort gap. The sign of $P_{xy}$ varies with the ranking of efforts, so strategic interaction is determined locally by which player is currently ahead. If $a > 0$, extra effort by the leading\footnote{Throughout the paper we use ``leading player'' and ``trailing player'' as labels for whichever player has the higher and lower investment at a given profile. They are not names for player 1 and player 2, and they do not describe a temporal order of moves.} player lowers the marginal effectiveness of the trailing player's effort; if $a < 0$, the same marginal move raises it. We call these cases \emph{suppression} and \emph{empowerment}, respectively. Equivalently, since $1-P$ is $Y$'s winning probability, the marginal effectiveness of $Y$'s effort is $-P_y$, and a marginal increase in $x$ changes it by $-P_{xy}$; under $a>0$ this is negative whenever $x>y$, giving suppression.

The players have standard quasi-linear expected utilities
\begin{equation}
U_1(x,y) = P(x,y) - \theta_{x}\, x, \qquad U_2(x,y) = 1 - P(x,y) - \theta_{y}\, y,
\end{equation}
where the prize is normalized to $1$. Because $P$ is a polynomial, it eventually leaves $[0,1]$; we use $\bar P(x,y):=\min\{1,\max\{0,P(x,y)\}\}$ as the actual winning probability and define the \emph{admissible domain}
\[
\mathcal D:=\{(x,y)\in \mathbb{R}_{+}^{2}:0 < P(x,y) < 1\},
\]
on which $\bar P=P$.

Under complete information, costs are symmetric, $\theta_{x}=\theta_{y}=\theta$, and publicly known. Under incomplete information, we adopt the standard symmetric IID specification familiar from auction theory; see, e.g., \cite{Myerson1981,BergemannEtAl2022}. Specifically, player $i$'s type\footnote{Because payoffs are quasi-linear, modeling private information as cost uncertainty or as value uncertainty is equivalent up to a change of variables. We use cost uncertainty in line with the contest literature.} $\theta_i$ is an independent draw from a common distribution $F$ with support $[\alpha,\beta]\subset(0,c)$. The lower bound keeps costs strictly positive, and the upper bound $\beta<c$ ensures that even the highest-cost type has a positive marginal return at zero effort.

We do not characterize CSFs from a list of properties\footnote{For example, the cubic CSF need not be globally increasing in own effort: at strongly asymmetric profiles, additional effort can lower a player’s win probability. This does not affect the equilibrium analysis. Conditional on any opponent effort, an action located on a locally declining part of the raw success function cannot be optimal: reducing effort slightly both raises the player’s winning probability and lowers the effort cost. In any case, such non-monotonicity is not always implausible as a matter of substance: in some applications, for instance political advertising or sport, effort past a point can be counterproductive, as when overexposure provokes backlash or overtraining erodes performance.}; we work with a single tractable specification. The reason for using a cubic polynomial is simple: cubic is the minimal polynomial benchmark capable of generating non-trivial strategic interaction.

\begin{observation}[{Minimal strategic CSF}]\label{obs:quadratic}
If a polynomial CSF satisfies $P(x,y)=1-P(y,x)$ and has degree at most two, then $\partial^{2} P/\partial x\,\partial y \equiv 0$. Cubic is therefore the minimal polynomial degree at which symmetry and a non-trivial cross-effect are compatible.
\end{observation}

A fully general antisymmetric cubic polynomial could also include pure own-effort terms such as $x^{3}-y^{3}$. These terms do not contribute to strategic interaction in a direct way --- they affect the curvature of own-effort returns without changing the cross-derivative --- but they greatly complicate the algebra without affecting any of the substantive results below. Omitting them sharpens the role of $a$ at no substantive cost.

It is instructive to contrast the cubic CSF with the two standard alternatives in the literature. The generalized Tullock contest is suppressive by construction.

\begin{observation}[Tullock contest is suppressive]\label{obs:tullock-suppressive}
For the generalized Tullock CSF
\[
P(x,y)=\frac{x^{r}}{x^{r}+y^{r}}, \qquad r>0,
\]
we have
\[
\frac{\partial^{2} P}{\partial x\,\partial y} = \frac{r^{2}\, x^{r-1} y^{r-1}\,(x^{r} - y^{r})}{(x^{r} + y^{r})^{3}},
\]
which is positive exactly when $x > y$. Hence the generalized Tullock contest is suppressive as a functional-form property independent of any parameter choice.
\end{observation}

The all-pay auction \cite{BayeKovenockDeVries1996}, obtained in the limit as $r \to \infty$, is the limiting case of extreme suppression in this family.

The comparison with Hirshleifer's difference-form contest \cite{Hirshleifer1989} is more subtle. If $P(x,y) = G(x-y)$ for a symmetric CDF $G$ with density $g$, then $P_{xy} = -g'(x-y)$. Symmetry forces $g'(0) = 0$, so the relevant local object is $-g''(0)/2$, which (locally) plays the role of $a$. For standard smooth unimodal noise --- logistic, normal, Student-$t$ --- one has $g''(0) < 0$, so the contest is locally suppressive at the symmetric profile, and integrability forces reversion to suppression at large effort gaps; global empowerment is therefore unavailable in the difference-form family, whereas the cubic CSF is empowering throughout its admissible domain $\mathcal{D}$ whenever $a < 0$.

\begin{observation}[Geometry of the admissible domain]\label{obs:domain-geometry}
The admissible domain $\mathcal D$ is a non-convex region containing the diagonal $x=y$, symmetric across it, unbounded along it, and bounded along every non-diagonal ray in $\mathbb R_+^2$.
\end{observation}

Figure \ref{DomainFig} shows representative contour plots of the raw cubic probability on a fixed window of the effort space; the Appendix gives the analytic description of $\mathcal D$ in diagonal coordinates, the ray-boundedness argument, and the non-convexity argument. Intuitively, probabilities remain well-defined only when neither player's effort dominates too strongly. This restricts the admissible region to profiles where efforts are relatively close, generating a band around the diagonal in which the contest behaves smoothly.

\begin{figure}[t]
\centering
\begin{subfigure}[b]{0.31\textwidth}
\includegraphics[width=\textwidth]{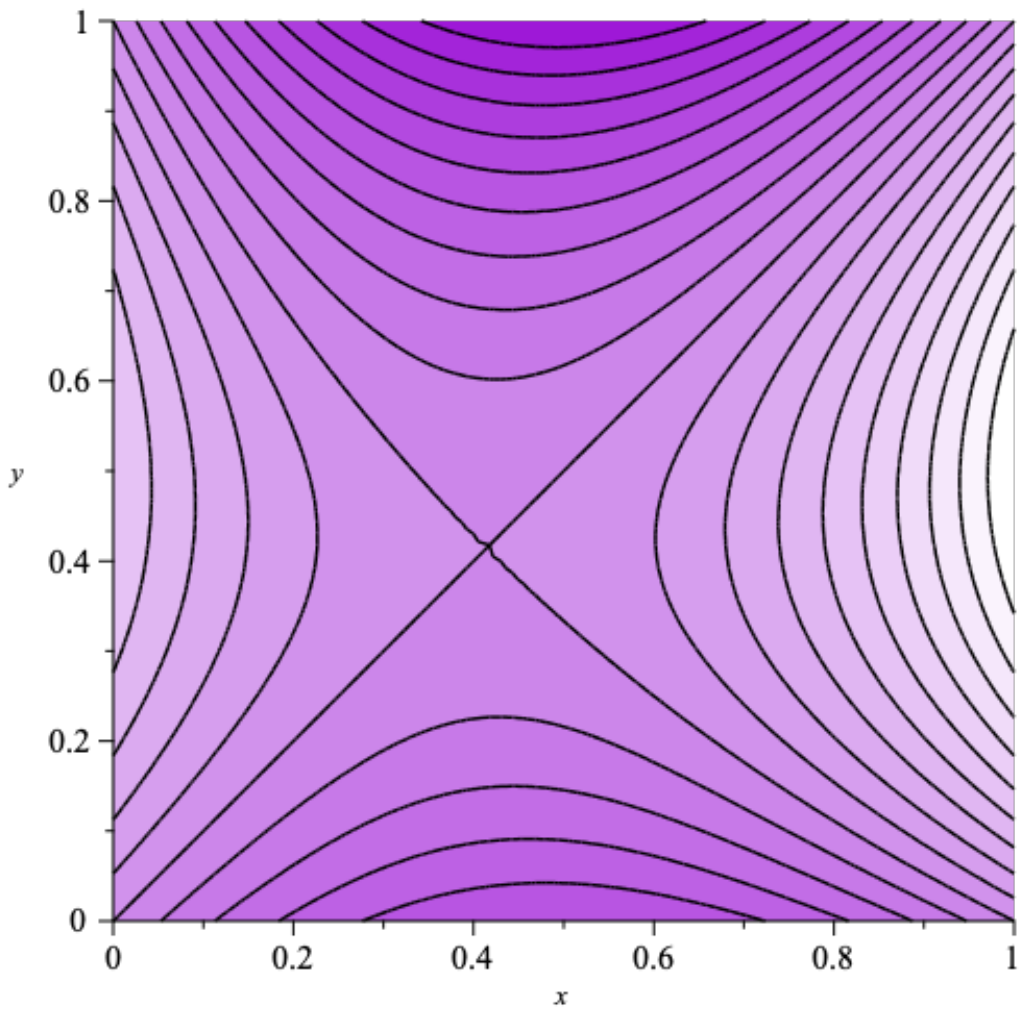}
\caption{$a=-1$}
\end{subfigure}
\hfill
\begin{subfigure}[b]{0.31\textwidth}
\includegraphics[width=\textwidth]{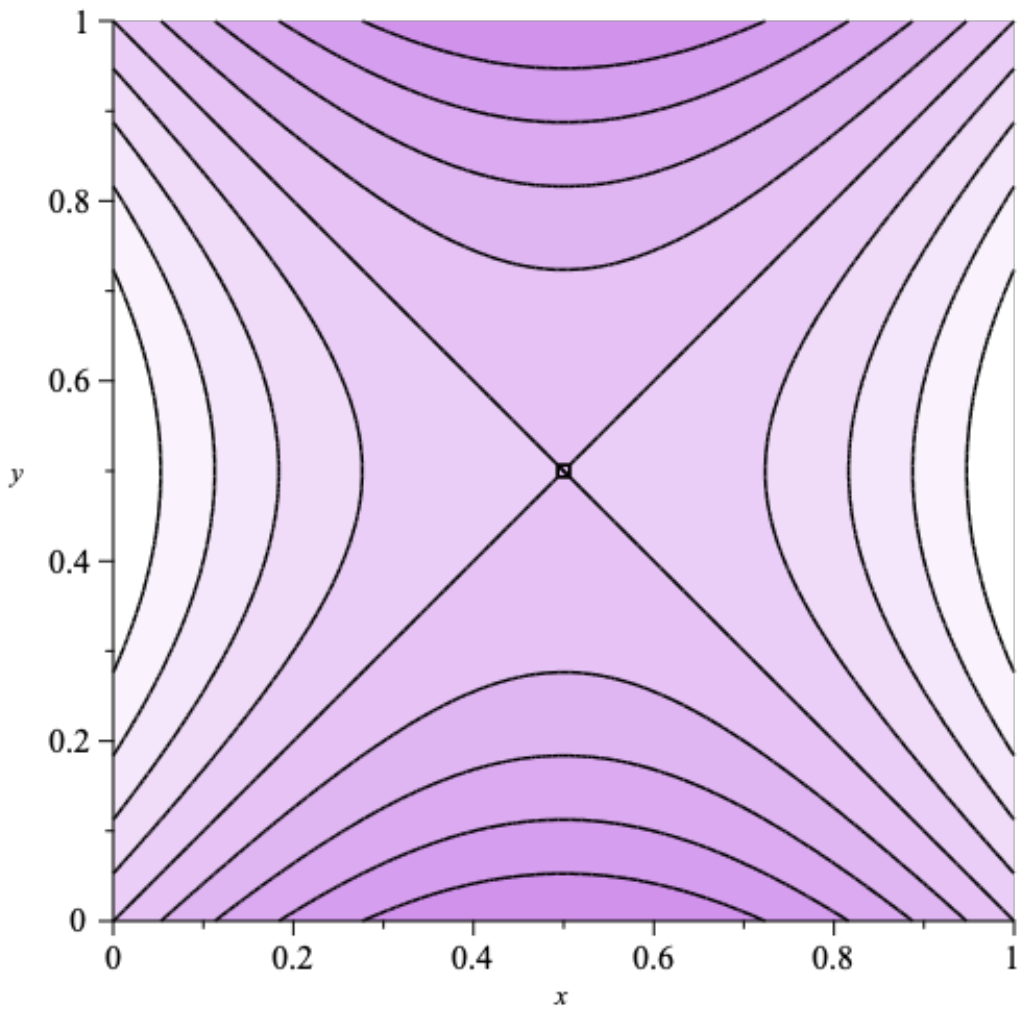}
\caption{$a=0$}
\end{subfigure}
\hfill
\begin{subfigure}[b]{0.31\textwidth}
\includegraphics[width=\textwidth]{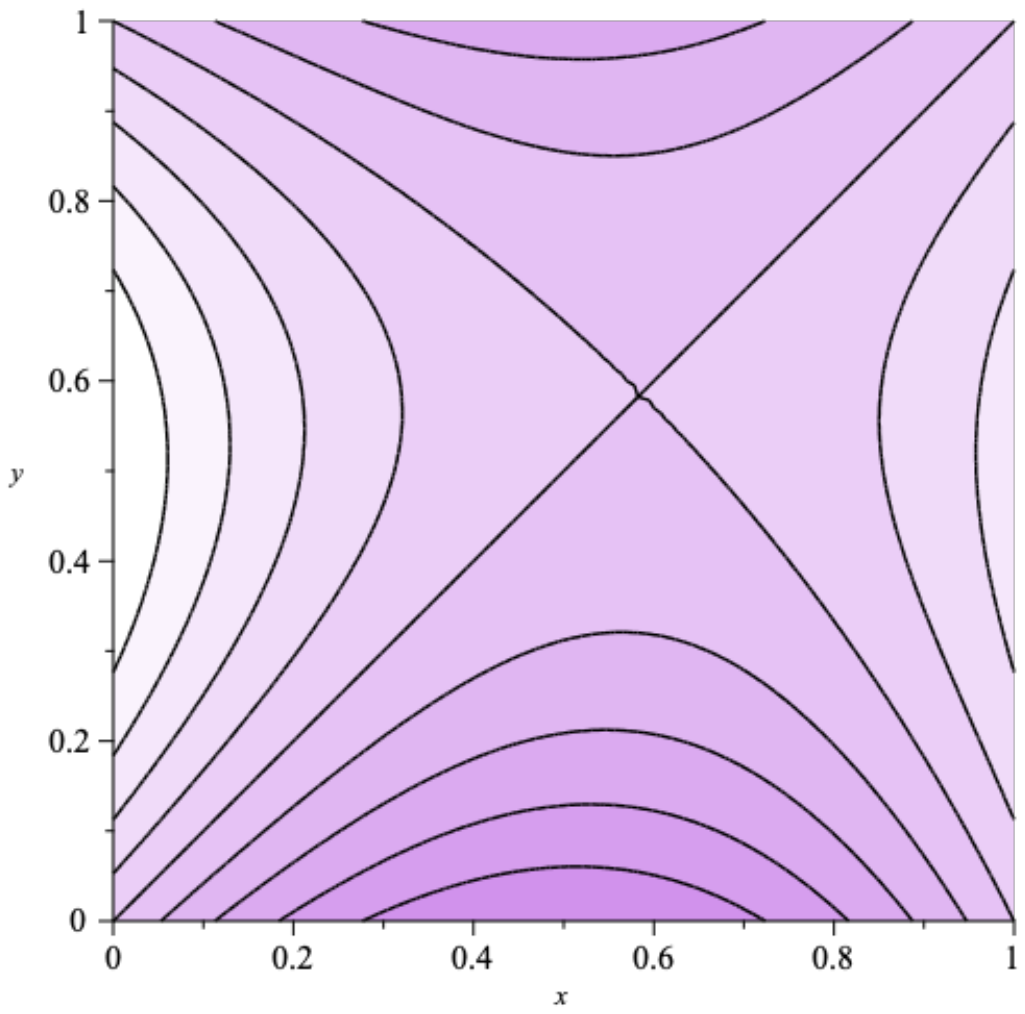}
\caption{$a=0.5$}
\end{subfigure}
\caption{Representative contour plots of the raw cubic probability on a fixed window of the effort space. The admissible domain is the region between the level curves $P=0$ and $P=1$; the diagonal $x=y$ is the $P=\frac{1}{2}$ locus and remains inside the domain for all parameter values.}
\label{DomainFig}
\end{figure}

A final modeling question concerns which equilibrium concept the truncated contest calls for. Because the cubic formula is the economically meaningful object but truncation is unavoidable, we organize the analysis into three increasingly demanding equilibrium concepts. Let the \emph{unrestricted polynomial game} denote the auxiliary game obtained by dropping only the $[0,1]$ truncation of $P$, while retaining the effort space $\mathbb R_+$. In the Bayesian analysis below, it will also be useful to consider the \emph{affine relaxation}, in which the action space is enlarged from $\mathbb R_+$ to $\mathbb R$; this auxiliary object is used only to derive the closed-form affine rule in Theorem \ref{thm:bn-unrestricted}, and is distinct from the unrestricted polynomial game.
\begin{enumerate}
\item \emph{Unrestricted benchmark.} A Nash equilibrium of the unrestricted polynomial game. This is the natural benchmark if one treats the cubic formula itself as the economically relevant object.

\item \emph{Local benchmark survival.} An unrestricted candidate survives locally if the candidate actions --- and, in mixed or Bayesian settings, the relevant support --- lie in $\mathcal D$ and the candidate is a local Nash equilibrium of the truncated contest. Equivalently, no sufficiently small feasible deviation is profitable. Under this interpretation the cubic CSF is read as a local approximation, so truncation is inactive in the payoff comparisons that matter.

\item \emph{Restricted-game survival.} An unrestricted candidate survives globally if it is a Nash equilibrium of the literal truncated contest based on $\bar P$. This is the strongest reading, because one must also rule out profitable distant deviations that may leave $\mathcal D$ and trigger clipping at $0$ or $1$.
\end{enumerate}
Each headline theorem below first characterizes the relevant algebraic benchmark; subsequent propositions and remarks identify the conditions under which that benchmark survives in the local or global sense. We do not attempt to classify equilibria whose existence is created purely by truncation; our focus is on when the interior cubic benchmark remains valid.

\section{Results}

\subsection{Complete Information}

Complete information is the natural benchmark for the model and the environment that permits the cleanest comparison with the existing contest literature. It also contains the basic geometry that drives the incomplete-information analysis below.

We start with the unrestricted polynomial game and extract the comparative statics of interest. All equilibrium interaction analyzed below takes place on the admissible domain $\mathcal D$; we return to the truncated contest only to ask when the same conclusions survive. For $a\neq 0$, define
\[
\kappa =\frac{b}{a},
\qquad
\zeta=\kappa^2-\frac{c-\theta}{a}=\frac{b^2-a(c-\theta)}{a^2}.
\]
These are the canonical variables of the complete-information problem. At a symmetric profile the first-order condition reduces to
\[
ax^2-2bx+c-\theta=0,
\]
whose discriminant is $4a^2\zeta$. The sign of $\zeta$ therefore determines whether the symmetric fixed point exists, while the sign of $a$ determines whether the contest is suppressive or empowering.

The first result gives the unrestricted benchmark and unifies the pure and mixed cases. The equilibrium structure depends on whether the symmetric first-order condition admits an interior solution: if it does, equilibrium is unique, pure, and symmetric; if it does not, players must randomize, and equilibrium is pinned down by the first two moments alone.

\begin{theorem}[Complete-information benchmark: pure and mixed equilibria]\label{thm:cinfo-unrestricted}
\leavevmode\par\noindent
Let $a\neq 0$.
\begin{enumerate}
\item \emph{Pure-equilibrium region.} If $\zeta\geq 0$, the unrestricted contest has a unique equilibrium, which is symmetric:
\begin{equation}\label{pure}
x^*=y^*=\kappa-\mathrm{sgn}(a)\sqrt{\zeta}.
\end{equation}

\item \emph{Mixed-equilibrium region.} If $\zeta<0$, no pure equilibrium exists. Every unrestricted mixed equilibrium strategy has mean $\kappa$ and variance $-\zeta$:
\begin{equation}\label{mixed}
E[x]=E[y]=\kappa,
\qquad
\mathrm{Var}(x)=\mathrm{Var}(y)=-\zeta,
\end{equation}
and, conversely, every pair of distributions on $[0,\infty)$ with these moments is an unrestricted mixed equilibrium.
\end{enumerate}
The two regions meet at $\zeta=0$, i.e.\ when $b^2=a(c-\theta)$.
\end{theorem}

Two features of this result deserve emphasis. First, the equilibrium itself is symmetric, so neither player is leading at the equilibrium profile. Empowerment and suppression are properties of the contest technology rather than of the equilibrium outcome: they describe how the cross-partial behaves at the asymmetric profiles a player contemplates when evaluating deviations, and it is this off-equilibrium geometry that decides whether equilibrium is pure or mixed. Under empowerment ($a<0$) one always has $\zeta>0$, so equilibrium is pure. Under suppression ($a>0$), increasing $a$ eventually drives $\zeta$ below zero, the symmetric fixed point disappears, and mixing becomes unavoidable.

Second, mixed equilibrium is sharp in moments but not in distributions: expected utility against a mixed opponent is quadratic in own effort, so indifference pins down only the first two moments. The mixed-equilibrium region is therefore informationally robust --- only mean and variance matter, and any pair of distributions with the required moments forms an equilibrium. This contrasts with mixed equilibrium in the lottery family, where closed-form characterization typically requires the full distribution rather than just its moments \cite{BayeKovenockDeVries1994,Ewerhart2015}. The interpretation of $\kappa$ in \eqref{pure}--\eqref{mixed} is now natural: it is the target effort level around which equilibrium organizes, played as a pure action in the pure-equilibrium region and as the common mean in the mixed-equilibrium region.

\begin{remark}[Moment uniqueness versus distributional multiplicity]\label{rem:cinfo-moment-unique}
When $\zeta<0$, every pair of distributions on $[0,\infty)$ with mean $\kappa$ and variance $-\zeta$ is an unrestricted mixed equilibrium. This multiplicity matters for the truncation analysis below, where admissibility depends on the particular support. The incomplete-information analysis resolves it: in the fully active Bayesian region, the unique affine equilibrium has moments converging to \eqref{mixed} as $\sigma_\theta^2\to 0$.
\end{remark}

\begin{remark}[Degenerate case $a=0$]\label{rem:cinfo-zero}
If $a=0$, each player has the dominant strategy $x^*=y^*=(c-\theta)/(2b)$, which is strictly positive whenever $c>\theta$. Moreover,
\[
\frac{b-\sqrt{b^2-a(c-\theta)}}{a}\longrightarrow \frac{c-\theta}{2b}\qquad\text{as }a\to 0,
\]
so the pure equilibrium extends continuously through the degenerate case, and the unrestricted equilibrium is pure and unique in a neighborhood of $a=0$.
\end{remark}

The curvature of best responses makes the role of $a$ transparent.

\begin{corollary}[Best-response curvature]\label{cor:cinfo-br-curvature}
Fix an opponent effort $y$ such that the unrestricted best response is interior, i.e.\ $ay<b$. Under empowerment ($a<0$) this condition holds for every $y\geq 0$ and best responses are convex throughout. Under suppression ($a>0$) the condition restricts $y$ to the interval $[0,\kappa)$; on this interior branch best responses are concave on the pure-equilibrium side ($\zeta>0$) and convex on the mixed-equilibrium side ($\zeta<0$). The change in curvature occurs exactly at $\zeta=0$, which is also the pure-to-mixed boundary in the unrestricted benchmark.
\end{corollary}

Figure \ref{BR} plots the unrestricted interior branches as solid curves and adds dotted segments only where the true best-response correspondences in the truncated contest differ from them. The curvature signs translate into a clear economic statement about how deviations from symmetric play are punished. Under empowerment, a player who is behind responds by investing \emph{more}, so off-equilibrium play is corrected from below; under suppression in the pure-equilibrium region, the player who is behind responds by investing \emph{less}, so off-equilibrium play is corrected from above. In both cases the correction is sufficient to return play to the symmetric fixed point. In the mixed-equilibrium region, by contrast, both directions of correction become locally amplifying rather than locally damping: the symmetric fixed point ceases to be reachable by best-response adjustments, which is why pure-strategy equilibrium fails and mixing becomes necessary.

Figure \ref{T-vs-P} shows the geometric contrast with Tullock directly. In the lottery family, best responses inherit a singularity at zero effort that has long been a source of technical difficulty in the closed-form analysis of mixed equilibrium. The cubic benchmark is smooth at zero effort, which is why both pure and mixed equilibria admit clean characterizations through Theorem \ref{thm:cinfo-unrestricted}.

\begin{figure}[ht]
\centering
\includegraphics[width=\textwidth]{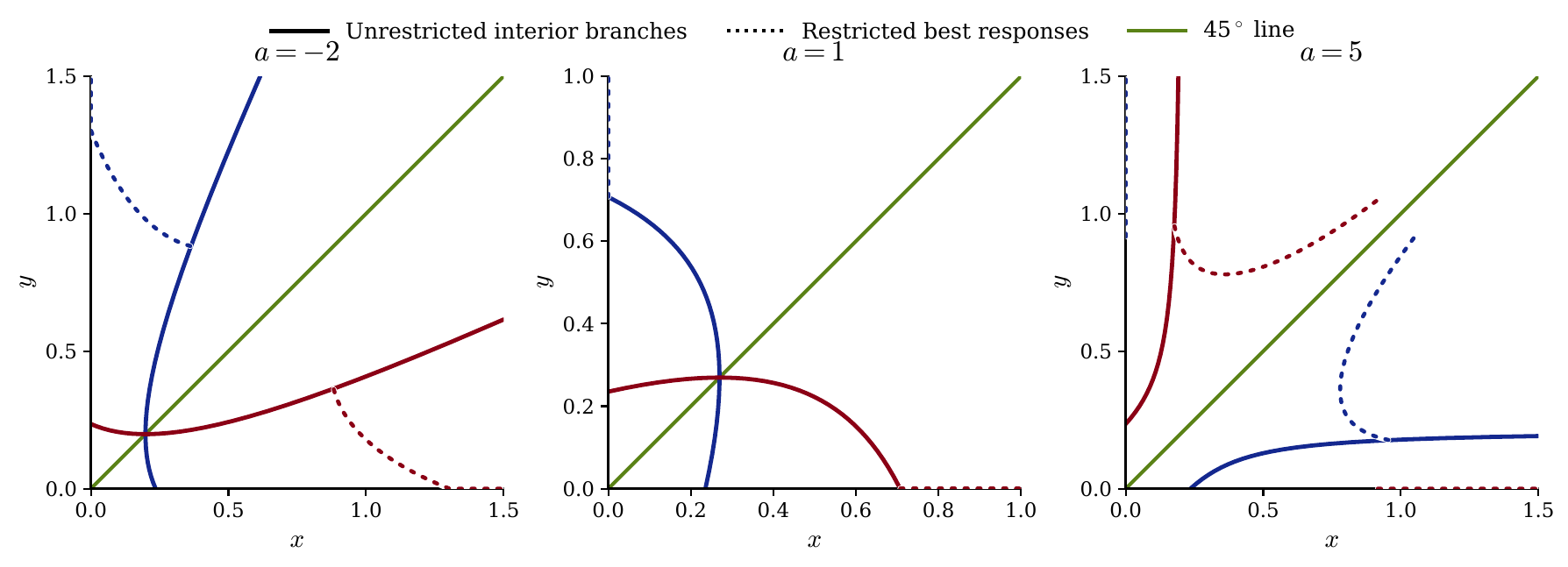}
\caption{Best-response correspondences for $a\in\{-2,1,5\}$ with $b=\nicefrac{17}{16}$, $c=1$, and $\theta=\nicefrac{1}{2}$. Solid curves show the unrestricted interior branches analyzed in Corollary \ref{cor:cinfo-br-curvature}; dotted segments show only the parts of the true best responses in the truncated contest that differ from those unrestricted branches; and the green line is the $45^{\circ}$ line.}\label{BR}
\end{figure}

\begin{figure}[ht]
\centering
\includegraphics[width=8.5cm]{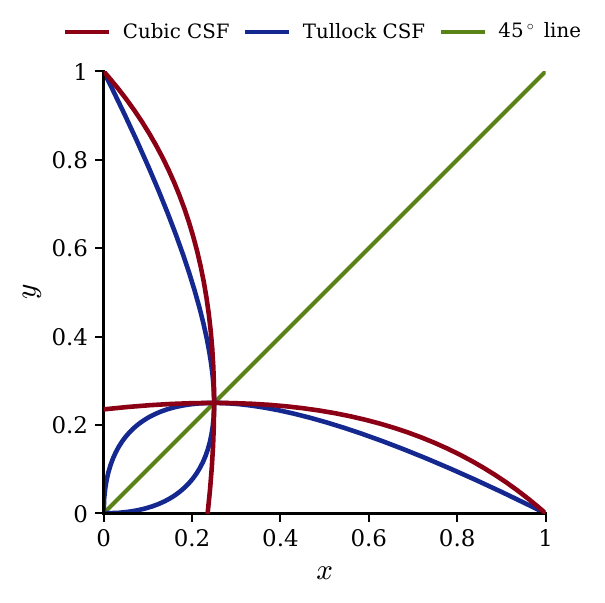}
\caption{Best responses for the cubic CSF and the Tullock CSF with comparable parameters: $a=\nicefrac{1}{2}$, $b=\nicefrac{17}{16}$, $c=1$, $\theta=\nicefrac{1}{2}$. The cubic best response is smooth at zero effort, whereas the Tullock best response has a singularity there.}\label{T-vs-P}
\end{figure}

We now turn to comparative statics. The next result is the centerpiece of the complete-information analysis:  Stronger suppression does not monotonically increase effort. Beyond a threshold, it destabilizes deterministic competition itself.

\begin{theorem}[Single-peaked equilibrium effort]\label{thm:cinfo-effort}
Equilibrium effort, viewed as a function of $a$ with $b$, $c$, and $\theta$ held fixed, is increasing on the pure-equilibrium region $\{a:\zeta(a)\geq 0\}$ and decreasing on the mixed-equilibrium region $\{a:\zeta(a)<0\}$, where on the latter every unrestricted mixed equilibrium satisfies $E[x]=b/a$. Hence equilibrium effort is maximized at the boundary $\zeta=0$, equivalently $b^2=a(c-\theta)$.
\end{theorem}

The mechanism is straightforward. In the pure-equilibrium region, stronger suppression makes it more valuable to cement a lead, so each player invests more aggressively. Once suppression is strong enough to push the benchmark into the mixed-equilibrium region, the relevant equilibrium object is no longer a fixed point but the moment restriction $E[x]=b/a$: further suppression makes dispersion increasingly costly and expected effort falls. The two regions meet at the equilibrium boundary, producing a non-differentiable peak.

The substantive interpretation is that the route to lower rent dissipation through stronger suppression is also the route into mixed play, with the unpredictability that mixing entails. In the pure region, contest behavior is deterministic and rent dissipation is monotone in $a$. Once mixing takes over, players randomize, expected effort declines, but the contest outcome --- and the effort level that produces it --- becomes a matter of which draws happen to be realized. Strong suppression therefore reduces dissipation at the cost of replacing deterministic play with randomized play. This is the symmetric two-player version of the breakdown that Tullock himself flagged in the lottery family and addressed only by departing from the symmetric benchmark, through bias or restrictions on entry.\footnote{The breakdown corresponds to Zones II and III of Tullock's Tables 6.1 and 6.2, where the lottery equilibrium implies that contestants invest more in aggregate than the prize is worth. Tullock characterizes these zones as ``intellectually fascinating'' but practically unhelpful, and redirects attention to Zone I via bias and entry restrictions; see \cite{Tullock1980}, pp.~13--16. The cubic benchmark replaces that breakdown with a tractable mixed-equilibrium characterization inside the symmetric model.}

\begin{corollary}[Rent dissipation]\label{cor:cinfo-max}
Total unrestricted equilibrium effort equals $2x^*$ in the pure-equilibrium region and $2b/a$ in the mixed-equilibrium region. Hence rent dissipation is maximized at $\zeta=0$ rather than at the extreme of suppression: it rises under moderate suppression and falls once mixing takes hold.
\end{corollary}

\begin{center}
\includegraphics[trim=0mm 0mm 0mm 0mm, clip, width=0.82\textwidth]{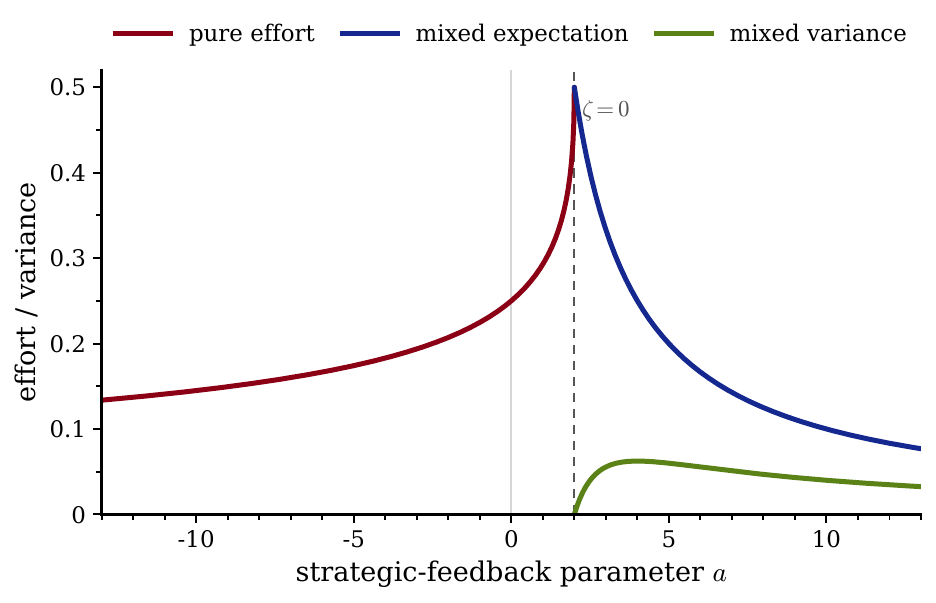}
\captionof{figure}{Equilibrium effort and variance as functions of $a$ for $b=1$, $c=1$, and $\theta=\nicefrac{1}{2}$. The effort curve has a sharp, non-differentiable peak at the pure/mixed equilibrium boundary $\zeta=0$.}\label{EffortFig}
\end{center}

A few comments are in order. First, this sharpens the usual intuition that stronger competitive pressure raises rent dissipation: in the cubic benchmark, the relationship is non-monotone in the strategic-feedback parameter, with the peak at a finite interior degree of suppression rather than at the all-pay limit. Second, we do not push the analysis all the way to extreme suppression because domain restrictions eventually become binding there, as the truncation conditions below make precise; in that limit, the all-pay auction \cite{BayeKovenockDeVries1996} is in any case the more appropriate benchmark. The point of the cubic analysis is that, well before that limit is reached, rent dissipation is already maximized at an interior $a$.

We now turn to the restrictions, following the order announced in the Model section: unrestricted benchmark, participation, local Nash equilibrium, and Nash equilibrium in the restricted game. For pure equilibrium, truncation matters only through participation, in the sense that the prize must be worth competing for given the equilibrium effort cost --- that is, equilibrium expected utility must be nonnegative, so that a player would not strictly prefer deviating to zero effort.

\begin{proposition}[Truncation and participation]\label{prop:cinfo-truncation}
\begin{enumerate}
\item \emph{Pure-equilibrium region.} If $a=0$, or if $a\neq 0$ and $\zeta\geq 0$, the unrestricted pure candidate is the unique on-domain equilibrium of the truncated contest if and only if its expected utility is nonnegative,
\[
\tfrac{1}{2}\geq \theta x^*.
\]
\item \emph{Mixed-equilibrium region.} If $a\neq 0$ and $\zeta<0$, any on-domain mixed equilibrium of the truncated contest must still satisfy the unrestricted moment condition \eqref{mixed} and the participation condition
\[
\tfrac{1}{2}\geq \theta\kappa.
\]
\end{enumerate}
Sufficient conditions for existence in the mixed-equilibrium region are given in Propositions \ref{prop:cinfo-two-point} and \ref{prop:cinfo-mixed-simple}.
\end{proposition}

In the pure-equilibrium region, participation is the whole story for benchmark survival: once it holds, the unrestricted candidate remains an equilibrium of the truncated contest. In the mixed-equilibrium region, the proposition gives a necessary restriction --- the unrestricted moments plus participation --- but turning it into an existence statement requires specifying a representation. The natural one is the smallest support that matches the moments.

\begin{proposition}[Canonical two-point family]\label{prop:cinfo-two-point}
\leavevmode\par\noindent
Assume the mixed-equilibrium region, i.e.\ $a>0$ and $\zeta<0$. Let
\[
s:=\sqrt{-\zeta},
\qquad
z:=\frac{c-\theta}{b},
\qquad
p:=\frac{b^2}{a(c-\theta)}.
\]
The moment profile \eqref{mixed} is generated by the following piecewise two-point family, which is an unrestricted mixed equilibrium by Theorem \ref{thm:cinfo-unrestricted}. If
\[
a\leq \frac{2b^2}{c-\theta},
\]
each player assigns probability $\nicefrac{1}{2}$ to each of $\kappa-s$ and $\kappa+s$. If
\[
a\geq \frac{2b^2}{c-\theta},
\]
each player assigns probability $1-p$ to $0$ and probability $p$ to $z$. At the switch point $a=2b^2/(c-\theta)$, the two branches coincide.
\end{proposition}

Truncation conditions, however, are representation-dependent: different distributions with the same moments can interact differently with the boundary of the admissible region. We now give conservative primitive conditions under which the canonical two-point family is a Nash equilibrium of the restricted game. These conditions are deliberately strong; they give up sharpness in order to deliver clean closed-form expressions, and the Appendix sharpens them substantially.

\begin{proposition}[Restricted-game Nash equilibrium]\label{prop:cinfo-mixed-simple}
Assume the mixed-equilibrium region, and in addition assume that
\[
\theta<c\leq 2\theta,
\qquad
b\geq \frac{c^{2}}{2}.
\]
Then there exists an explicit cutoff $\bar a_{M}\geq 2b^{2}/(c-\theta)$, derived in the Appendix, such that for every $a\leq \bar a_{M}$ within the mixed-equilibrium region, the truncated contest admits an on-domain mixed equilibrium, with the canonical two-point family of Proposition \ref{prop:cinfo-two-point} as one admissible representative. In particular, the on-domain mixed-equilibrium region is nonempty.
\end{proposition}

Economically, the two additional inequalities jointly bound the aggressiveness of the contest technology: $\theta<c\leq 2\theta$ keeps the prize-to-cost ratio from making low-effort competition too aggressive, and $b\geq c^2/2$ ensures that own-effort returns decay fast enough that large investments do not become too potent. Together they carve out a nonempty safe region for suppression in which the unrestricted comparative statics survive in the literal truncated contest. The Appendix gives the tight branch-specific conditions and shows that admissibility genuinely fails once suppression becomes sufficiently strong.

\begin{remark}[Local versus global Nash equilibrium]\label{rem:cinfo-local-mixed}
For the canonical two-point families in Proposition~\ref{prop:cinfo-two-point}, a simple sufficient condition for local Nash equilibrium is that the support lie strictly inside $\mathcal D$, so that truncation is inactive for sufficiently small deviations. This holds on the symmetric branch if and only if $s<1/(4\theta)$ (equivalently $P(\kappa-s,\kappa+s)>0$) and on the endpoint branch if and only if $b>2\theta(c-\theta)$ (equivalently $P(0,z)>0$). These local support checks are weaker than the global deviation checks behind Proposition \ref{prop:cinfo-mixed-simple}; under the strict primitive version $c<2\theta$ and $b>c^2/2$, they are automatic. At boundary equalities the support can touch the truncation boundary, so the strict local-interiority criterion no longer applies.
\end{remark}
\subsection{Incomplete Information}

We now allow each player's cost to be private, with types drawn IID from $F$. The contest technology is unchanged; what changes is that each player chooses effort against a distribution of opponent efforts rather than a known action. As under complete information, the cubic CSF admits a closed-form benchmark.\footnote{Unlike the standard auction-theoretic treatment, our analysis imposes no regularity on $F$ beyond a bounded support in $(0,c)$; in particular, $F$ may have atoms. The closed-form structure of Theorem \ref{thm:bn-unrestricted} survives intact.}

As in the complete-information analysis, we proceed sequentially. In the Bayesian case it is useful to insert one algebraic step before the hierarchy from Section~2: first solve the affine relaxation that ignores $x\geq0$, then impose nonnegativity in the unrestricted polynomial game, and finally check truncation. Let
\[
M_1=\mathrm{E}[\theta],
\qquad
M_2=\mathrm{E}[\theta^2],
\qquad
\Delta:=c-M_1,
\qquad
\sigma_\theta^2:=M_2-M_1^2,
\]
and, for $a\neq 0$, define
\[
\kappa=\frac{b}{a},
\qquad
\omega=\frac{\sigma_\theta^2}{a^2},
\qquad
\zeta=\kappa^2-\frac{\Delta}{a}=\frac{b^2-a\Delta}{a^2}.
\]
Throughout this subsection, assume nondegenerate type uncertainty, so $\sigma_\theta^2>0$. The variable $\zeta$ is the complete-information expression with the realized cost $\theta$ replaced by the mean type $M_1$; $\omega$ is the scaled variance of types and is the channel through which uncertainty enters equilibrium effort.

Given a symmetric opponent strategy $x(\theta)$, if player $X$ of type $\theta_x$ deviates to $\tilde x$, polynomial interim utility is
\[
U(\theta_x,\tilde x)=\int_{\alpha}^{\beta} \left[P(\tilde x,x(\theta_y))-\theta_x\tilde x\right] \, dF(\theta_y),
\]
and the expression for player $Y$ is analogous. On the admissible domain, where $\bar P=P$, this also coincides with the truncated contest studied later.

\begin{theorem}[Affine Bayesian Nash benchmark]\label{thm:bn-unrestricted} 
Assume $a\neq 0$. Consider the affine relaxation, in which actions may be any real number. There is a unique symmetric Bayesian Nash equilibrium. It is affine and strictly decreasing, and expected effort is
\[
\mathrm{E}[x]=\kappa-\frac{\mathrm{sgn}(a)}{\sqrt{2}} \sqrt{ \zeta + \sqrt{\zeta^2+ \omega}}.
\]
\end{theorem}

For later use, write the equilibrium as $x(\theta)=k\theta + d$. Then the slope, intercept, and effort variance are
\[
k=-\frac{1}{\sqrt{2}|a| \sqrt{ \zeta + \sqrt{\zeta^2+ \omega}}},
\qquad
d=\mathrm{E}[x]-kM_1,
\qquad
\mathrm{Var}(x)=\frac{\sqrt{\zeta^{2}+\omega}-\zeta}{2}.
\]

Two features deserve emphasis. First, the equilibrium is in closed form and affine in type for any IID prior --- a tractability that the lottery family does not offer at this generality. Second, expected effort depends on the type distribution only through its mean and variance, exactly as in the complete-information mixed-equilibrium region: only the first two moments matter. The equilibrium is therefore robust to any informational refinement that preserves the mean and variance: changes in higher moments, in the shape of $F$, or in finer distributional features do not affect the equilibrium characterization.

\begin{remark}[Degenerate case $a=0$]\label{rem:bn-zero}
At $a=0$, raw strategic interaction disappears. In the affine relaxation and the unrestricted polynomial game, each type plays the dominant strategy $x(\theta)=(c-\theta)/(2b)>0$, where positivity follows from $\beta<c$. Expected effort is
\[
\mathrm{E}[x]=\frac{c-M_1}{2b},
\]
which depends only on the mean type. Uncertainty therefore has no effect at $a=0$, and the formulas from Theorem \ref{thm:bn-unrestricted} extend continuously to this case, with $k\to -1/(2b)$, $d\to c/(2b)$, and $\mathrm{Var}(x)\to (M_2-M_1^2)/(4b^2)$ as $a\to 0$. In particular, full activity for the nonnegative-effort benchmark holds in a neighborhood of $a=0$. The literal truncated contest is a separate check: at $a=0$, the affine support lies strictly inside $\mathcal D$ if and only if $\beta^2-\alpha^2<2b$, and global restricted-game survival must still rule out deviations created by clipping.
\end{remark}

\begin{remark}[Vanishing uncertainty as equilibrium selection]\label{rem:bn-vanishing}
\leavevmode\par\noindent
Within the affine benchmark, the variance-only dependence noted above has a useful implication in the complete-information limit. As $\sigma_\theta^2\to 0$, the moments of the Bayesian effort distribution converge to the complete-information benchmarks: $\mathrm{E}[x]\to x^*$ and $\mathrm{Var}(x)\to 0$ if $\zeta\geq 0$, while $\mathrm{E}[x]\to \kappa$ and $\mathrm{Var}(x)\to -\zeta$ if $\zeta<0$. Thus, in the complete-information mixed-equilibrium region, vanishing private uncertainty selects the mixed-equilibrium moment profile pinned down by the affine Bayesian limit; for any shrinking prior family whose standardized types converge, the affine push-forward also selects a corresponding distribution with those limiting moments.
\end{remark}

We now turn to the comparative statics. We start with the effect of uncertainty itself, holding the strategic-feedback parameter fixed. Within the fully-active region of the affine benchmark, expected effort depends on uncertainty only through $\omega=\sigma_\theta^2/a^2$, so the sign of the comparative static follows immediately from the sign of $a$.

\begin{theorem}[Uncertainty and effort]\label{thm:uncertainty}
Fix $a$, $b$, $c$, and the mean type $M_1$, and restrict attention to variance changes for which the affine benchmark is fully active, so the affine rule is nonnegative on the whole support. Expected equilibrium effort is decreasing in the variance of types $\sigma_\theta^2$ under suppression ($a>0$) and increasing in $\sigma_\theta^2$ under empowerment ($a<0$). At $a=0$, expected effort is independent of $\sigma_\theta^2$.
\end{theorem}

The mechanism mirrors the complete-information curvature analysis. Under suppression, a random setback is especially costly because falling behind lowers the payoff from further effort, and types react to that prospect by investing less; greater dispersion amplifies the expected cost and depresses average effort. Under empowerment, the same logic runs in reverse: being behind is less damaging because the leading side's effort raises the trailing side's marginal return, so dispersion in types raises expected effort. This is the same asymmetry that produced the curvature flip in Figure \ref{BR}; it now governs how the contest responds to private information.

The result clarifies how our comparative static differs from the existing incomplete-information literature. Hurley and Shogren identify two forces in asymmetric-information lottery contests: effort becomes a risky input, which tends to reduce effort, while perceptions of value asymmetries can raise or lower effort and make efficiency rankings unsystematic; Wasser shows that no information yields lower expected aggregate effort than both private and complete information under equal mean costs, but finds no general effort ranking between complete and private information; and Kirkegaard obtains tractability in mixture contests through summary measures of opponents' behavior, with uncertainty sometimes raising expected action even as performance falls \cite{HurleyShogren1998,Wasser2013,Kirkegaard2025}. Theorem~\ref{thm:uncertainty} gives a cleaner statement within a single specification: holding the rest of the model fixed, the sign of $a$ alone determines whether greater type dispersion discourages expected effort under suppression or encourages it under empowerment.

The variance-only dependence in Theorem~\ref{thm:bn-unrestricted} also has direct consequences for information design. We use the symmetric IPV persuasion environment of \cite{BergemannEtAl2022}. Fix a prior $F$ and let $\mathcal S$ be the set of symmetric Bayes-plausible information structures in which each player observes a private signal about their own type only: the same signal rule is applied to each player, and signals are independent across players conditional on the type profile. For each $\pi\in\mathcal S$, let $\hat\theta_\pi$ denote the symmetric distribution of posterior-mean costs $\mathrm{E}[\theta\mid s_\pi]$ induced by that signal rule, so each player's realized posterior mean is an IID draw from $\hat\theta_\pi$. The affine benchmark induced by $\pi$ is obtained by replacing the original type distribution with $\hat\theta_\pi$; when $\hat\theta_\pi$ is degenerate, the benchmark is interpreted by continuity. Away from neutrality, any signal that changes the variance of posterior-mean costs changes expected effort, and any signal that does not leaves expected effort unchanged. The next theorem turns this observation into an all-or-nothing disclosure result whose sign is again governed by suppression versus empowerment.

\begin{theorem}[Effort-maximizing disclosure in symmetric IPV persuasion]\label{thm:bn-disclosure}
In the affine benchmark induced by the persuasion model above, ex ante expected effort depends on a signal $\pi$ only through $\mathrm{Var}(\hat\theta_\pi)$. Consequently, an effort-maximizing designer chooses
\begin{enumerate}
\item no disclosure under suppression ($a>0$),
\item full disclosure under empowerment ($a<0$), and
\item any feasible signal at $a=0$ (the designer is indifferent).
\end{enumerate}
Equivalently, more Blackwell-informative signal structures weakly lower expected effort under suppression and weakly raise it under empowerment.
\end{theorem}

The contrast with nearby disclosure results is instructive. \cite{BergemannEtAl2022} studies symmetric IPV persuasion in classical auctions; with two bidders and revenue maximizing designer, their result specializes to full pooling --- no disclosure --- and by revenue equivalence the same conclusion holds for the two-player all-pay auction, which corresponds to the case of extreme suppression. Our suppressive endpoint reproduces exactly that no-disclosure benchmark; the empowerment endpoint has no analogue in the suppressive auction or lottery families. In contests, \cite{ZhangZhou2016}, \cite{Serena2022}, and \cite{AntsyginaTeteryatnikova2023} also obtain genuinely interior disclosure policies in richer environments. The all-or-nothing structure here is not a contradiction of those results but a consequence of the cubic benchmark: with a one-dimensional sufficient statistic (the variance of the posterior mean), the designer's problem collapses to a sign question, and the sign is governed by the suppression--empowerment distinction. Conditions under which the same all-or-nothing conclusion holds in the literal truncated contest are given in Proposition~\ref{prop:bn-disclosure-sufficient} below, once the truncation machinery is in place.

Having characterized how the contest responds to type dispersion and to the design of private signals, we now turn from these informational margins to the technological margin: the role of the strategic-feedback parameter $a$ itself, holding the type distribution fixed. Expected effort remains single-peaked in $a$, as in the complete-information benchmark of Theorem~\ref{thm:cinfo-effort}. Under incomplete information, however, the peak is smooth rather than non-differentiable, and its location depends on the dispersion of types. Unpacking the expected-effort expression from Theorem~\ref{thm:bn-unrestricted} gives
\begin{equation}\label{eq:bn-e1-a}
E_1(a)=
\begin{cases}
\displaystyle \frac{b-\sqrt{\frac{b^2-\Delta a+\sqrt{(b^2-\Delta a)^2+\sigma_\theta^2 a^2}}{2}}}{a}, & a\neq 0,\\[2.2ex]
\displaystyle \frac{\Delta}{2b}, & a=0,
\end{cases}
\end{equation}
so the $a$-comparative static reduces to the shape of a scalar function. Define
\[
\rho:=\frac{\sigma_\theta^2}{\Delta^2},
\qquad
\bar a:=\frac{4b^2\Delta}{\sigma_\theta^2}.
\]
\begin{theorem}[Single-peaked expected effort]\label{thm:bn-a-peak}
On the region $E_1(a)>0$, equivalently $a<\bar a$, the expected effort in \eqref{eq:bn-e1-a} has a unique maximizer $a^\dagger$. Moreover,
\[
a^\dagger>0 \iff \sigma_\theta^2<\Delta^2,
\qquad
a^\dagger=0 \iff \sigma_\theta^2=\Delta^2,
\qquad
a^\dagger<0 \iff \sigma_\theta^2>\Delta^2.
\]
Within the fully-active region of the affine benchmark, the effort-maximizing contest is therefore suppressive when type variance is small relative to the baseline return $\Delta^2$, and empowering when type variance is large. Moreover, the proof yields the sharper bound $a^\dagger<2b^2\Delta/(\Delta^2+\sigma_\theta^2)<\bar a$.
\end{theorem}

The mechanism combines the two effects already on the table. When $a$ is small, the strategic gain from pulling ahead is limited, so expected effort is low. Raising $a$ initially strengthens those stakes and pushes effort up. Once $a$ becomes large, however, the suppressive comparative static from Theorem \ref{thm:uncertainty} takes over: heterogeneity makes falling behind increasingly costly, and that cost depresses effort. The peak is the value of $a$ at which these two forces balance, and its location depends on the size of $\sigma_\theta^2$ relative to the baseline marginal return $\Delta$: small uncertainty leaves the peak in suppression, large uncertainty pushes it into empowerment. The knife-edge $\sigma_\theta^2=\Delta^2$ puts the peak exactly at $a=0$. Single-peakedness here is a property of the fully active benchmark; once dropout sets in, the constrained schedule need not be single-peaked, as Remark~\ref{rem:bn-dropout-peak} shows.

This contrasts sharply with the complete-information result in Theorem \ref{thm:cinfo-effort}, where the peak sits at the pure/mixed equilibrium boundary $\zeta=0$ and produces the non-differentiable apex in Figure \ref{EffortFig}. Under incomplete information, the same peak is now a smooth function of the variance ratio $\rho=\sigma_\theta^2/\Delta^2$, and Bayesian heterogeneity bends the schedule before the complete-information pure-to-mixed switch is reached. Figure \ref{BNPeakFig} isolates this.

\begin{remark}[Positive mean versus full activity]\label{rem:bn-positive}
The threshold $\bar a=4b^2\Delta/\sigma_\theta^2$ is an algebraic positivity boundary for the affine-benchmark average effort: the closed-form expression in \eqref{eq:bn-e1-a} has $E_1(a)>0$ exactly when $a<\bar a$, with $E_1(\bar a)=0$. This is not the admissibility boundary for the affine Bayesian strategy. Full activity requires the entire affine rule to be nonnegative on the type support. In the suppressive region this support condition is governed by the prior-dependent dropout threshold $a_D(F)$ of Proposition \ref{prop:bn-dropout} below; once $a>a_D(F)$, high-cost types choose zero and the cutoff-affine equilibrium of Proposition~\ref{prop:bn-nonnegative} becomes the relevant object. Under empowerment, dropout does not arise, but the literal truncated contest still requires the separate support checks discussed below.
\end{remark}

\begin{center}
\begin{minipage}{\textwidth}
\centering
\includegraphics[width=\textwidth]{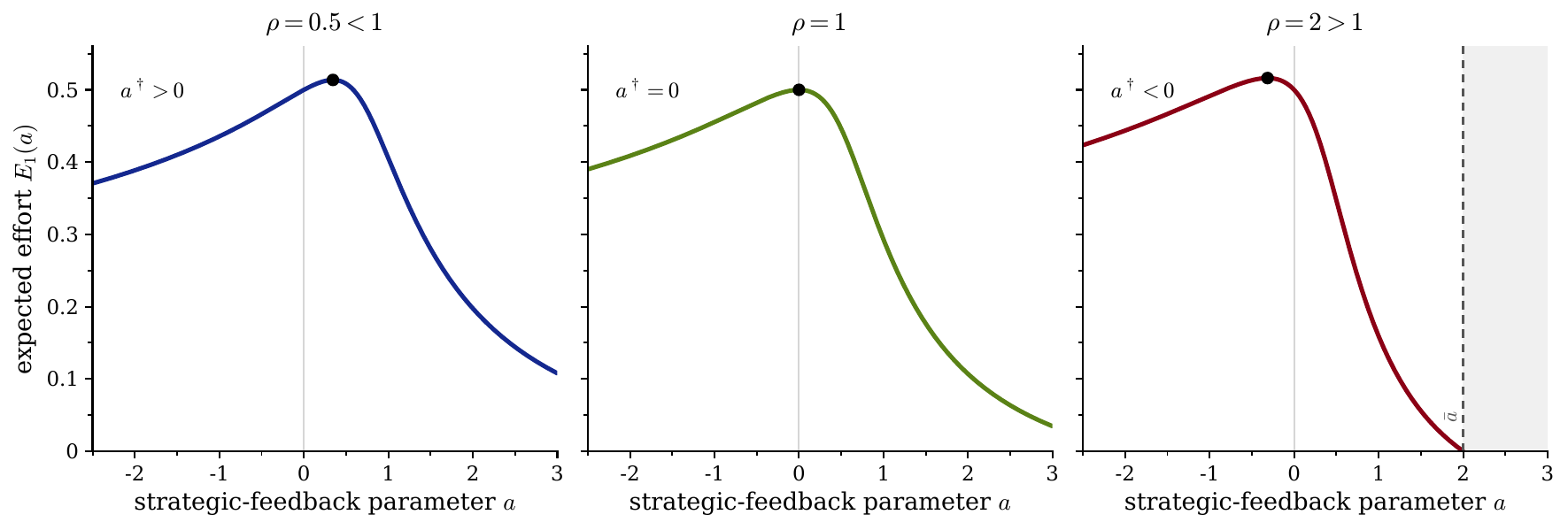}
\captionof{figure}{Expected effort in \eqref{eq:bn-e1-a} for three values of $\rho=\sigma_\theta^2/\Delta^2$, under the normalization $b=\Delta=1$. The left panel has $\rho<1$ and peak at $a^\dagger>0$; the middle has $\rho=1$ and peak at $a^\dagger=0$; the right has $\rho>1$ and peak at $a^\dagger<0$.}\label{BNPeakFig}
\end{minipage}
\end{center}

Figure~\ref{MainFig} summarizes the complete- and incomplete-information regimes as a function of $a$.

\begin{figure}[H]
\centering
\includegraphics[width=\textwidth]{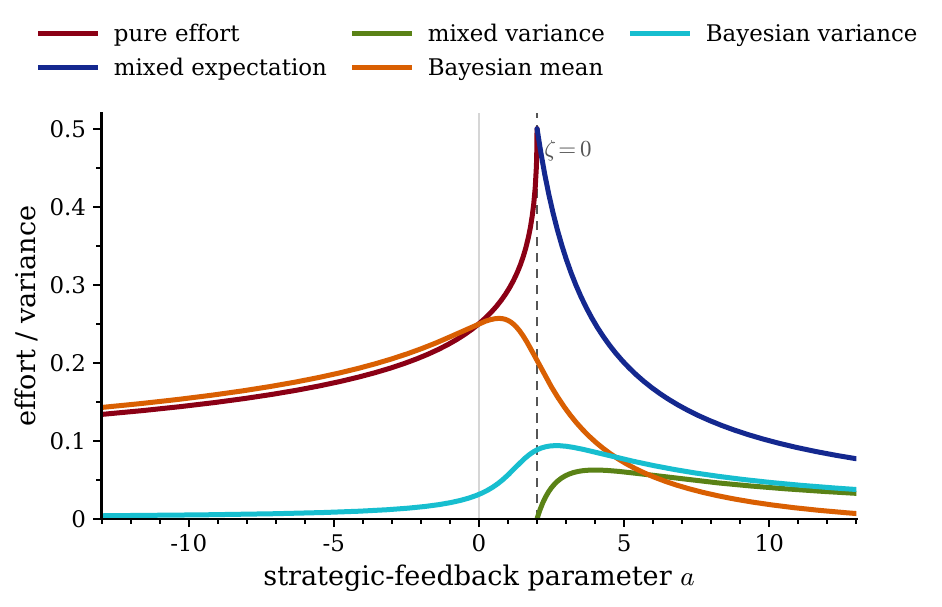}
\caption{Equilibrium as a function of $a$ for $b=1$, $c=1$, $\mathrm{E}\theta=\nicefrac{1}{2}$, and $\mathrm{E}\theta^2=\nicefrac{3}{8}$.}\label{MainFig}
\end{figure}

The next result studies how the peak itself moves with the dispersion of types.

\begin{theorem}[Type variance and the peak-effort parameter]\label{thm:bn-peak-variance}
Holding $b$ and $\Delta=c-M_1$ fixed, there exist unique thresholds
\[
0<\sigma_1^2<\Delta^2<\sigma_2^2
\]
such that $a^\dagger$ is increasing in $\sigma_\theta^2$ on $(0,\sigma_1^2)$, decreasing on $(\sigma_1^2,\sigma_2^2)$, and increasing on $(\sigma_2^2,\infty)$. Moreover,
\[
\lim_{\sigma_\theta^2\downarrow 0}a^\dagger=\frac{b^2}{\Delta},
\qquad
\lim_{\sigma_\theta^2\to\infty}a^\dagger=0^-.
\]
\end{theorem}

Numerically the thresholds are highly asymmetric: $\sigma_1^2\approx 0.0045\Delta^2$ and $\sigma_2^2\approx 8.73\Delta^2$. Moreover, scaled by $\Delta^2$, they are properties of the cubic benchmark itself: they depend on neither $b$, nor $c$, nor any feature of the type distribution beyond its mean. The middle decreasing region is by far the dominant one, so in most plausible applications the observable comparative static is the simple one --- the effort-maximizing parameter moves left as variance rises, first toward neutrality and then into empowerment.

Two economically relevant restrictions sharpen this further. Bounded support bites only from the right: if types are supported on $[\alpha,\beta]$, then feasible variance satisfies $\sigma_\theta^2\leq (M_1-\alpha)(\beta-M_1)$, so the high-variance reversal is absent whenever this upper bound is below $\sigma_2^2$. Full activity bites from the left in a prior-dependent way: for each prior $F$, Proposition \ref{prop:bn-dropout} gives a threshold $a_D(F)$. Along a variance experiment $F_\sigma$ --- a family of priors indexed by $\sigma_\theta^2$ --- the fully active peak is observable only when $a^\dagger(\sigma_\theta^2)\leq a_D(F_\sigma)$. Whether the initial increasing region (where $a^\dagger$ moves up with variance) is fully observable depends on the prior family; Appendix Remark \ref{rem:dropout-calibration} works through the uniform-prior case.

We now turn to the economically natural restriction $x\geq 0$. If the affine rule from Theorem \ref{thm:bn-unrestricted} satisfies $k\beta+d\geq 0$, the nonnegative-effort game coincides with the affine benchmark. If $k\beta+d<0$, sufficiently high-cost types optimally drop out to zero, and the equilibrium becomes cutoff-affine: active types follow an affine schedule up to an endogenous cutoff, while higher-cost types choose zero. This connects to the activity--inactivity problem in private-information contests \cite{SchoonbeekWinkel2006}, but the mechanism here is specific: under empowerment ($a<0$) this complication does not arise, as Proposition \ref{prop:bn-nonnegative} shows; dropout is a suppression phenomenon.

\begin{proposition}[Nonnegative-effort equilibrium]\label{prop:bn-nonnegative}
The unrestricted polynomial game has a symmetric Bayesian Nash equilibrium with $x\geq 0$. Apart from a possible exceptional boundary equilibrium, every symmetric equilibrium is cutoff-affine:
\[
x(\theta)=\lambda (t-\theta)_+
\]
for some $\lambda>0$ and some cutoff $t\geq \alpha$. If $t\geq \beta$, the equilibrium coincides with the affine benchmark from Theorem \ref{thm:bn-unrestricted}. If $t<\beta$, the types $\theta\geq t$ choose zero. An exceptional boundary equilibrium can also arise when $m:=F(\{\alpha\})>0$ and $am(c-\alpha)\geq b^2$: every type $\theta>\alpha$ chooses zero while type $\alpha$ is the only active type and randomizes on $\{0,(c-\alpha)/b\}$, degenerating to the pure action $(c-\alpha)/b$ at equality. The Appendix gives the full characterization, including this exceptional atom-at-$\alpha$ case.
\end{proposition}

When suppression becomes strong enough, high-cost types find that the strategic gain from competing no longer offsets their cost, and they withdraw to zero effort. The equilibrium support contracts endogenously: the contest is still well-defined, but it is now played by a subset of types.

\begin{proposition}[Monotone dropout under an atomless prior]\label{prop:bn-dropout}
Assume the type distribution $F$ is atomless and admits a continuously differentiable density $f$ on $[\alpha,\beta]$. Then there exists a prior-dependent threshold $a_D(F)>0$ such that the equilibrium from Proposition \ref{prop:bn-nonnegative} is fully active for $a\leq a_D(F)$ and partially active for $a>a_D(F)$. For $a>a_D(F)$, the cutoff $t=t(a)\in(\alpha,\beta)$ is uniquely determined and strictly decreasing in $a$. Consequently, the dropout rate
\[
\delta(a):=1-F(t(a))
\]
is strictly increasing in $a$. The mappings $t$ and $\delta$ are continuously differentiable, and $t(a)\downarrow \alpha$ and $\delta(a)\uparrow 1$ as $a\to\infty$.
\end{proposition}

The dependence on $F$ is essential; Appendix Remark~\ref{rem:dropout-calibration} defines $a_D(F)$ through lower partial moments and works it out in closed form for a uniform prior.

Two questions follow: which comparative statics survive once dropout begins, and what dropout does to the shape of the effort schedule. We take them in turn.

\begin{remark}[What survives under dropout]\label{rem:bn-dropout}
Dropout is a suppression-only qualification. Under empowerment ($a<0$) and at neutrality ($a=0$) no type drops out, so Theorems~\ref{thm:uncertainty}, \ref{thm:bn-disclosure}, \ref{thm:bn-a-peak}, and~\ref{thm:bn-peak-variance} hold unchanged, subject in the literal contest only to the support check of Proposition~\ref{prop:bn-truncation}. The same holds under moderate suppression, $a\leq a_D(F)$. For stronger suppression the affine formulas give way to the cutoff-affine equilibrium, where the relevant objects are lower partial moments over active types rather than the raw moments $M_1$ and $M_2$. This region is characterized, not a gap: Proposition~\ref{prop:bn-dropout} gives the monotone cutoff, and Appendix Remark~\ref{rem:dropout-calibration} gives the explicit lower-partial-moment formulas.
\end{remark}

The shape question is different: once the cutoff is active, aggregate effort is governed by endogenous lower partial moments, and the constrained path need not inherit the affine benchmark's single peak.

\begin{remark}[Dropout and multiple effort peaks]\label{rem:bn-dropout-peak}
The single-peakedness of Theorem~\ref{thm:bn-a-peak} is a property of the fully active affine benchmark, not of the constrained equilibrium. Once dropout sets in, two forces act on aggregate effort: stronger suppression removes high-cost types from the top, while the surviving low-cost types compete harder. Either force can dominate locally, so the constrained aggregate-effort path need not be single-peaked. A smooth prior can generate two local peaks. With $b=6$, $\alpha=1$, $\beta=2$, $c=2.00293$, and types $\theta=1+Z$ for the three-component beta mixture specified in the detailed proof, expected effort rises to a first peak at $a\simeq46.43$, falls to a trough at $a\simeq95.86$, and rises again to a second peak at $a\simeq108.97$. The first peak is already on the cutoff-affine branch (dropout starts at $a_D(F)\simeq0.96$), so it differs slightly from the fully active affine maximizer $a^\dagger\simeq46.17$. The active mass falls throughout this range---from $0.996$ at the first peak to $0.439$ at the trough and $0.369$ at the second---so the second rise is genuine: averaging over all types, with dropouts counted as zeros, the survivors' intensified effort more than offsets the additional exit. The scale choice $b=6$ keeps the displayed branch away from probabilistic truncation; on the displayed constrained branch, a numerical support check gives raw probabilities between $0.060$ and $0.940$, and a clipped-payoff deviation check finds no profitable deviation to numerical tolerance.
\end{remark}

\begin{figure}[H]
\centering
\includegraphics[width=\textwidth]{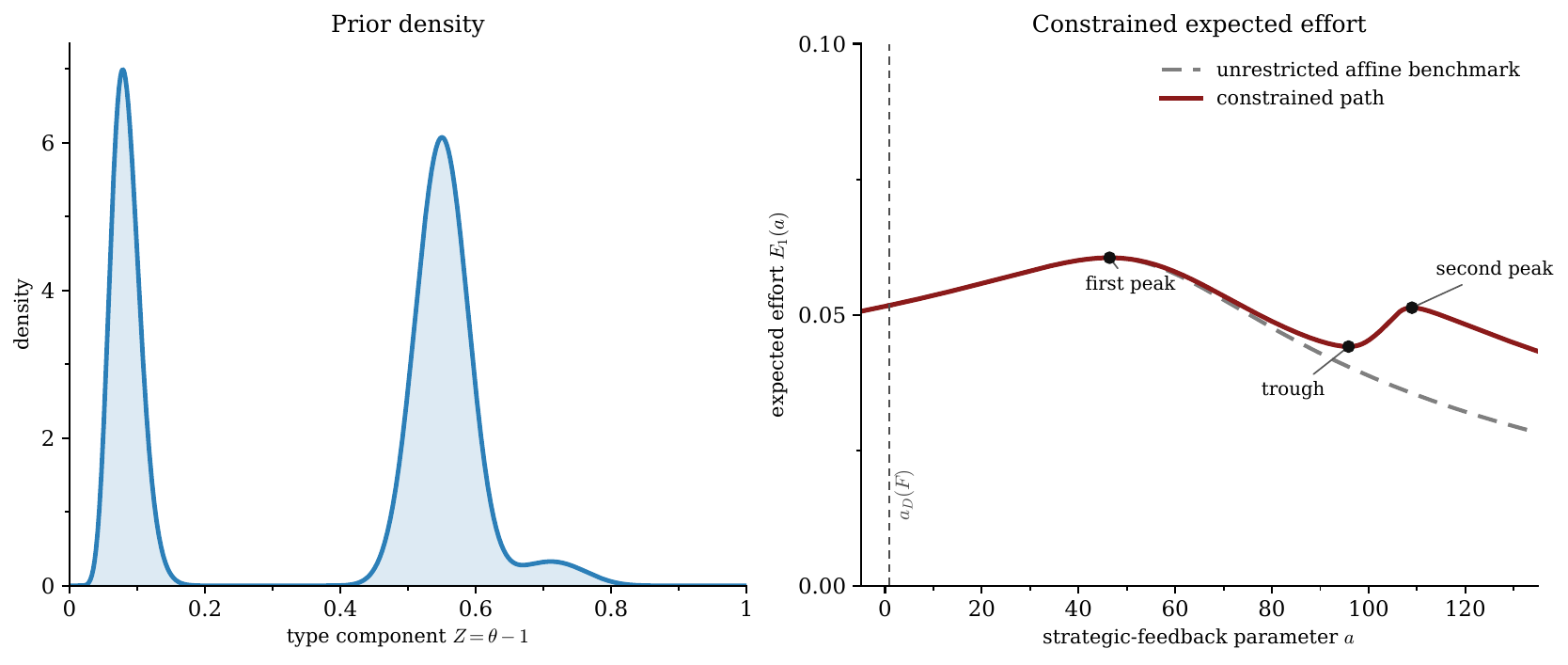}
\caption{Density and expected effort in the two-peak dropout example. The left panel plots the density of $Z=\theta-1$; the right panel plots unconditional expected effort along the constrained path, with inactive types counted as zeros.}
\label{DropoutDensityEffortFig}
\end{figure}

Proposition \ref{prop:bn-nonnegative} settles the action constraint $x\geq 0$, but truncation can still matter in the literal contest based on $\bar P$, even under empowerment. From the viewpoint of a deviating type, the opponent's effort is random because it is generated by the opponent's private type, so the same state-by-state clipping mechanism as in the complete-information mixed region can arise. As before, it is useful to separate a conservative sufficient condition for the main text from a sharper support-based treatment in the Appendix.

\begin{proposition}[Simple sufficient condition for the Bayesian benchmark]\label{prop:bn-truncation}
\leavevmode\par\noindent
Suppose Proposition~\ref{prop:bn-nonnegative} yields a fully active equilibrium, so equilibrium follows the affine benchmark rule
\[
x^{u}(\theta)=k\theta+d,
\]
with
\[
\underline x:=x^{u}(\beta),
\qquad
\overline x:=x^{u}(\alpha).
\]
If
\[
\overline x\leq \frac{1}{\alpha},
\qquad
c^2\leq 2b,
\qquad
a\leq b\alpha,
\qquad
c\alpha^2-2b\alpha+a\geq 0,
\]
then $x^{u}(\theta)$ is a symmetric Bayesian Nash equilibrium of the truncated contest based on $\bar P$.
\end{proposition}

On this parameter region, the affine benchmark, the unrestricted polynomial game, and the truncated Bayesian game all coincide. The inequalities imply $a>0$, so this simple criterion applies only under suppression. The sharper Appendix Proposition \ref{prop:bn-truncation-detailed} replaces the global condition $c^2\leq 2b$ with a support-specific check on $[\underline x,\overline x]$ and can also certify empowering cases; under $a<0$, full activity is automatic, so only support-based truncation remains.

We close the section by extending the disclosure result of Theorem~\ref{thm:bn-disclosure} from the affine benchmark to the literal truncated contest. The conclusion is the same, but it has to be verified signal by signal once nonnegativity and truncation are in play.

\begin{proposition}[Sufficient conditions for Theorem \ref{thm:bn-disclosure} in the literal contest]\label{prop:bn-disclosure-sufficient}
\leavevmode\par\noindent
Maintain the notation of Theorem \ref{thm:bn-disclosure}. Suppose that for every feasible signal $\pi\in\mathcal S$, the induced posterior-mean contest satisfies both of the following conditions:
\begin{enumerate}
\item it is fully active, so Proposition \ref{prop:bn-nonnegative} coincides with the affine benchmark; and
\item truncation is inactive on the induced equilibrium support, for example by Proposition \ref{prop:bn-truncation-detailed}.
\end{enumerate}
Then the conclusion of Theorem \ref{thm:bn-disclosure} also holds in the literal contest with $x\geq 0$ and truncated CSF $\bar P$. Under empowerment ($a<0$), condition 1 is automatic by Proposition \ref{prop:bn-nonnegative}, so only condition 2 must be checked. At $a=0$, full activity is automatic by Remark \ref{rem:bn-zero}, but condition 2 still has to be checked unless a support-based truncation criterion verifies that clipping is harmless. Under suppression ($a>0$), both conditions must be verified signal by signal.
\end{proposition}

Two broader points emerge from the incomplete-information results, both of which sharpen the message of the complete-information analysis. First, the typical comparative statics found in the contest literature --- that dispersion discourages effort and that more information reduces effort --- reflect the suppressive sign of strategic feedback rather than a property of contests in general. The cubic CSF reveals them as one side of a symmetric pair, with empowering environments producing the opposite responses. Second, the location of the effort-maximizing strategic-feedback parameter is itself uncertainty-dependent: over the dominant middle range of type variances identified above, $a^\dagger$ moves left as variance rises, first toward neutrality and then into empowerment, so the design margin that maximizes effort shifts with the informational environment.

Taken together, our results show that the direction of strategic interaction fundamentally shapes how contests respond to uncertainty, disclosure, and heterogeneity. Suppressive and empowering environments therefore generate qualitatively different comparative statics even within the same contest framework.

\section{Discussion: Contests and Institutions}

In contest theory, the success function is the basic building block. Choose it, and its properties travel quietly into every model that follows. The generalized Tullock contest has held that role for forty years, and for good reason: it is tractable, familiar, and well suited to the lottery-like settings it was designed for. Yet the choice carries a feature that is rarely stated out loud. At any asymmetric profile, the leading player's marginal effort lowers the trailing player's marginal return. This is not a parameter the modeler can tune. It comes built into the form. So a model that begins with Tullock begins, by default, in a suppressive world --- even when the application would suggest otherwise.

This is not a complaint about the lottery model where it fits. When effort buys tickets, shares of attention, or probabilistic claims on a prize, the lottery primitive is the right tool, and its built-in structure is part of what makes it useful. The question is what the field gains once that primitive is generalized.

Imagine, for a moment, a version of consumer theory in which only substitutes existed --- not because complements are absent from the world, but because the standard demand systems had been written in a way that ignores them. Antitrust, demand estimation, and platform analysis would all be poorer for it. Suppression has played a similar role in contest theory, and empowerment has been its missing twin. We do not mean that empowerment and complementarity are the same thing; empowerment is rank-dependent --- defined by the current ordering of investments --- whereas complementarity is not; Section 2 develops the distinction. The point is methodological. A single benchmark had quietly absorbed an entire conceptual category, and a whole range of questions become accessible once the missing category has a benchmark of its own.

Once such a benchmark is available, much of what follows comes almost free of charge. The cubic CSF delivers closed-form mixed equilibria under complete information, a unique affine Bayesian Nash equilibrium under incomplete information, and typically single-peaked expected effort in the interaction parameter. The sign of strategic feedback governs how uncertainty affects effort and how an effort-maximizing designer ranks disclosure policies. None of this requires new machinery. It comes from the geometry of a polynomial that is symmetric, smooth, and free of the singularity at zero effort that has long been the source of so much technical difficulty in lottery-style models. The doors that open here were closed for technical reasons, not substantive ones.

Read this way, several familiar findings stop looking like isolated domain-specific complications and begin to line up. In innovation, \cite{AghionEtAl2005} documents an inverted-U relationship between product-market competition and innovation, with competition discouraging innovation by laggard firms but encouraging it by neck-and-neck firms, while \cite{BloomSchankermanVanReenen2013} emphasizes the coexistence of positive technology spillovers and negative business-stealing effects. In politics, \cite{Stromberg2004} shows that media competition and advertising incentives skew coverage toward large groups and groups valuable to advertisers, while \cite{BesleyPrat2006} shows that media capture weakens government accountability and that pluralism and independent ownership can limit it.

These results sit in different literatures and use different language, but they share a common structural feature. Applied work repeatedly encounters institutions that alter how effort translates into effective competition, while the standard contest benchmark builds in only one sign of that feedback. The cubic CSF does not reproduce those environments in full. It does, however, supply a common reduced-form language for discussing them inside contest theory rather than around it.

The parameter $a$ is where institutions meet the contest. Real institutional environments are high-dimensional --- appropriability regimes, the orientation of intermediaries, lock-in, spillovers, scrutiny, disclosure --- and they reach the contest through different channels. The cubic benchmark collapses all of them onto one question: does the institutional feature in front of us make a local lead translate into a smaller or a larger marginal return for the trailing side? Strong patent enforcement and a clientelistic press both suppress, in this sense, however different they look in other respects; weak appropriability and an accountability-driven press both empower. The reduction is deliberate, and it has the usual cost: $a$ does not resolve the channels through which any particular institution reaches the contest. It has the usual gain: it makes very different applications comparable inside one tractable model.

We end where the field began. The second half of \cite{Tullock1980} is a search for institutional levers that lower the social cost of rent seeking. The levers Tullock identifies are all defensive: reduce the number of contestants, make marginal costs steeper, or pre-bias the contest toward a designated winner. He acknowledges that each comes at a price, and closes with an explicit call for further work on practical methods to reduce dissipation.

Inside the fully active cubic benchmark, expected rent-seeking effort is single-peaked in the interaction parameter, so it falls off as the contest moves toward either strong suppression or empowerment. The two edges are not symmetric. Strong suppression is the route Tullock effectively explored, and his own analysis ran into a wall in this region, where pure-strategy equilibrium fails and his proposed remedies work by modifying the symmetric benchmark rather than by working within it. Our framework shows that strong suppression does lower rent seeking, without those external interventions, but it does so at the price of unpredictability, with players mixing their strategies. Empowerment lowers effort by a different mechanism. It does not pre-select the winner or eliminate contestants. It changes how a local lead translates into the trailing side's marginal return: under strong empowerment, falling behind is less damaging because the leading side's effort raises the trailing side's marginal effectiveness, so the strategic value of aggressive investment to stay ahead is muted. The contest remains genuinely contested; what changes is that visible leadership invites response rather than entrenching itself, and the resulting equilibrium effort is lower.

That second route was unavailable in the canonical contest model, and once one starts to look for it in the world rather than the model, it is not hard to find. Independent media, antitrust enforcement, and intellectual-property regimes that balance protection with diffusion all work, in our terms, by making leadership invite response. Whether these arrangements are welcome is a question for an external welfare standard, not for the contest model itself. What the cubic benchmark contributes is a way to see them inside contest theory --- to read institutional reforms as movements in the strategic-feedback parameter, and to ask whether the contest one is studying belongs in the same world that the lottery primitive describes by default. Tullock asked how to limit socially wasteful competition without making the competition itself disappear. The cubic benchmark does not just answer that question. It brings institutions back into the basic contest model.

\pagebreak

\section*{Appendix}
\subsection*{Additional Applications of Empowerment and Suppression}

The two motivating examples in the introduction --- R\&D under different appropriability regimes, and political competition under different media institutions --- are not the only places where empowerment and suppression arise naturally. The basic question is always the same: fix a moment at which one side is ahead, then ask whether a marginal increase in the leading player's investment makes the trailing side's own additional effort more effective or less effective. If less, the environment is suppressive; if more, empowering. This section is an invitation to look for that pattern in other settings. We do not claim to establish that any particular mechanism is the operative one; we only mean to show that the question has a natural answer in a range of familiar contests, and that the answer can fall on either side.

In many contests the feedback runs through a scarce resource that both sides need, and the institutional environment governing that resource sets the sign. Capital and talent are the clearest cases, and within each one finds both suppression and empowerment.

Consider first the competition for capital. Two firms race toward an ambitious technological prize --- a moon hotel, a fusion plant, a brain-computer interface --- and SpaceX and Blue Origin are a recognizable instance. The question is who investors think benefits from the leader's spending. When they read a lead as the strongest signal of who will ultimately win, capital concentrates on the leader: every milestone it reaches tightens its hold on the available money, and the trailing firm's next round is harder to close. A marginal increase in the leader's investment makes the trailing firm's fundraising less effective. That is suppression. But if there is no shortage of investors hunting for the next Google or the next Tesla, some will read the same lead differently --- as the leader paying to clear the path. On this view the leader's spending de-risks the whole endeavour, maps the dead-ends, and proves the approach can work, and a cheaper trailing firm that learns from those mistakes is the better bet precisely because it free-rides on what the leader has established. Then the leader's progress draws money toward the challenger, and the trailing firm raises capital more easily as the leader pulls ahead. That is empowerment. The two regimes track recognizable features of the capital market: concentrated, late-stage money with low risk tolerance tends to chase the leader, while diversified, early-stage money hunting for the next breakout tends to spread out.

The competition for talent has the same two faces, and here the organizing variable is the direction in which experienced people flow relative to rank --- set by whether the costly work of bringing talent into the industry is done at the top or at the bottom.\footnote{The two-player model captures the bilateral version of this story: one firm bears the entry-recruiting cost and the other harvests or skims it. The ``field'' language below is the natural many-player extension, in which the same flow runs between the leading side and the rest of the field.}

When the leader does that work, talent flows downward. A dominant technology firm recruits raw talent from outside the industry and seasons it, after which smaller competitors headhunt that talent away with greater autonomy, a larger equity stake, and the chance to define a product rather than maintain one. The leader keeps doing this because it needs the people now; the leakage is a cost of leading, not a reason to stop. The effect is that the leader's recruiting raises the productivity of the trailing firm's recruiting, because the challenger hires from a pool the leader created; and when the challenger sharpens its offer, more of what the leader brings in walks out the door, so the leader's marginal recruiting becomes less productive. That is empowerment: whoever is on top does the entry recruiting and ends up seeding the field below it.

When the bottom does that work, talent flows upward. In football or modelling, smaller clubs and agencies scout and develop raw talent, and the top clubs and agencies then buy or poach it for the best opportunities. Here the leader's pulling power lowers the productivity of the trailing side's effort, because every prospect a small club develops is liable to be bought out before the club can capture the return; and when the smaller side develops harder, it enriches the pool the leader draws from, so the leader's marginal recruiting becomes more productive. That is suppression: whoever is on top skims the field below it. The contrast is simply where entry recruiting happens --- at the top, talent disperses downward and the leader subsidizes the field; at the bottom, it concentrates upward and the field subsidizes the leader.

The same dichotomy is easy to find beyond capital and talent. Platform and standards races turn on how a lead in installed base translates into the value of further investment. Under lock-in --- proprietary interfaces, exclusive complements, switching costs that grow with the user base --- the larger the lead, the more each remaining user is worth defending and the less a trailing platform gets for the same spending, because a user weighing which platform to join cares how many others have already joined. That is suppression. Under interoperability --- open standards, data portability, mandated compatibility --- a trailing platform can plug into the tooling and infrastructure the leader builds rather than rebuild it, so the leader's investment in deepening the ecosystem raises the productivity of the laggard's own effort. That is empowerment. Telecommunications regulators have chosen between these regimes deliberately for decades, and disputes over dominant digital platforms are, in our terms, arguments about which side of $a=0$ a market sits on.

A simpler everyday version plays out in retail. A leading consumer brand competing for physical shelf space converts its lead into retailer relationships, co-marketing budgets, and prime placement, all of which make a trailing brand's marketing budget less productive per dollar --- suppression. The same two brands competing for online attention face a different geometry: search and recommendation engines surface the leading brand alongside its closest alternatives, so a shopper looking for the leader is shown the challenger too, and the trailing brand's marginal spend benefits from the attention the leader has drawn to the category --- empowerment.

Sports is the setting where this dynamic has been managed most visibly, and for the longest time. The labour-market linkage just discussed is already present in sport --- clubs compete for players exactly as firms compete for talent --- and what follows sits on top of it. A sporting contest is not a clean static instance of the model, since races and seasons are dynamic by nature, but league commissioners, racing regulators, and tournament organizers have spent decades choosing, by hand, what sign the feedback between leader and chaser should take, in vocabulary almost identical to this paper's. At the level of season design, reverse drafts in the NHL, NBA, and NFL hand the trailing teams the next year's best new players, so that a current lead is deliberately made to raise the trailing side's future strength --- empowerment by construction; promotion-and-relegation in European football does the reverse, raising the cost of falling behind and letting strong clubs compound their advantages --- suppression by construction.

 At the level of a single race, the same dichotomy appears in the mechanics of competition. In cycling, the rider in front breaks the air and the rider behind drafts in the slipstream, spending less energy to hold the pace; a lead actively raises the trailing rider's effectiveness, which is why breakaways are hard to sustain and pelotons stay bunched. Motor racing legislates the same effect directly, with drag-reduction systems and aerodynamic rules that tune how much a trailing car benefits from the car ahead, while in other configurations a leading car uses dirty air and blocking lines to make the chaser's effort less effective.

Sport also makes one further feature of the model recognizable. The marginal return to a competitor's own effort is not always positive: athletes over-train, peak too early, and burn out. The cubic specification accommodates this through the local curvature of own-effort returns, which the standard lottery success function does not.

These examples have little in common on the surface, but they all turn on the same simple question: when one side pulls ahead, does that make catching up harder, or does it make the trailing side's next move more effective? The first case is suppression; the second is empowerment. The channel that carries the effect varies from setting to setting --- capital markets, hiring, platform rules, shelf space, league design, aerodynamics --- and so does the institution that fixes its direction. What does not vary is the question itself. In each case a current lead is converted, by some feature of the environment, into either a tighter hold on the contest or an opening for the side behind, and in the cubic benchmark that conversion is captured by the single parameter $a$. The benchmark does not claim that any one of these stories is the true mechanism in any one industry; it claims only that the suppressive and the empowering cases are equally natural.

\subsection*{Proofs}

This subsection gives the print-version proofs. The aim is to keep the logic visible without carrying every algebraic step. The full derivations, including the support checks and branch calculations, are collected in the online appendix below.

\subsubsection*{Model observations}

\begin{proof}{Proof of Observation}{\ref{obs:quadratic}}
The idea is that fairness makes $P-1/2$ antisymmetric. Any antisymmetric polynomial vanishes on the diagonal and therefore contains the factor $x-y$. If the total degree is at most two, the remaining symmetric factor can only be linear in $x+y$, so no $xy$ term can appear. Hence the cross-partial is zero, and cubic degree is the first place where a non-trivial strategic cross-effect can enter.
\end{proof}

\begin{proof}{Proof of Observation}{\ref{obs:tullock-suppressive}}
The proof is just the cross-partial calculation. All terms outside $x^r-y^r$ are positive, so the sign of the cross-partial follows the ranking of efforts. Thus Tullock contests are suppressive by functional form, not by parameter choice.
\end{proof}

\begin{proof}{Proof of Observation}{\ref{obs:domain-geometry}}
The useful coordinates are average effort and the effort gap. In those coordinates the diagonal has raw probability $1/2$, reversing the sign of the gap reflects the probability around $1/2$, and sufficiently large off-diagonal movement eventually pushes the raw polynomial outside $[0,1]$. Thus the admissible set is a symmetric band around the diagonal, unbounded along it but bounded in every other direction. It is necessarily non-convex: it contains every diagonal point and also contains off-diagonal points arbitrarily close to the diagonal, but the midpoint between such an off-diagonal point and a sufficiently large diagonal point keeps a fixed nonzero gap while sending average effort to infinity, where the raw polynomial leaves $[0,1]$.
\end{proof}

\subsubsection*{Complete-information results}

\begin{proof}{Proof of Theorem}{\ref{thm:cinfo-unrestricted}}
The key simplification is that a player's payoff is quadratic in own effort. In pure strategies, the two first-order conditions differ only by the factor
\[
(x-y)\bigl(a(x+y)-2b\bigr),
\]
so any pure candidate is either symmetric or lies on the line $x+y=2\kappa$. The second line cannot satisfy both players' second-order conditions, leaving only the symmetric equation
\[
ax^2-2bx+c-\theta=0.
\]
Thus pure equilibrium exists exactly when this equation has the admissible root selected by concavity,
\[
x^*=\kappa-\operatorname{sgn}(a)\sqrt{\zeta},
\]
which is the region $\zeta\geq0$.

When $\zeta<0$, the symmetric root disappears and mixing becomes necessary. Against a mixed opponent, expected payoff still depends on the opponent only through the first two moments and remains quadratic in the deviating effort:
\[
K+(aE_1-b)x^2+(c-\theta-aE_2)x.
\]
A non-degenerate support can be optimal only if this quadratic is flat. Flatness gives $aE_1=b$ and $aE_2=c-\theta$, hence mean $\kappa$ and variance $-\zeta$. Conversely, these moments make the deviator indifferent across all efforts, so any distributions on $[0,\infty)$ with those moments form an unrestricted mixed equilibrium.
\end{proof}

\begin{proof}{Proof of Remark}{\ref{rem:cinfo-moment-unique}}
The intuition is that the mixed equilibrium is pinned down by indifference of a quadratic payoff. Indifference fixes the coefficients of the quadratic, and therefore fixes only the first two moments of the opponent's distribution. Nothing in the unrestricted payoff sees higher moments, so the equilibrium is unique in moments but not in distributions.
\end{proof}

\begin{proof}{Proof of Remark}{\ref{rem:cinfo-zero}}
At $a=0$, the interaction term drops out of the marginal payoff: each player solves the same strictly concave one-variable problem regardless of the opponent's effort. The dominant strategy is therefore $(c-\theta)/(2b)$, and the formula for the pure-equilibrium root converges to exactly this value as $a\to0$.
\end{proof}

\begin{proof}{Proof of Corollary}{\ref{cor:cinfo-br-curvature}}
The curvature result is the local version of the same quadratic geometry. Rewriting the interior best response as
\[
x^{BR}(y)=\frac{y+\kappa}{2}+\frac{\zeta}{2(y-\kappa)}
\]
gives
\[
\frac{d^2x^{BR}}{dy^2}=\frac{\zeta}{(y-\kappa)^3}=\frac{a^3\zeta}{(ay-b)^3}.
\]
The second-order condition fixes the sign of the denominator. Under empowerment, $\zeta>0$ and the branch is convex. Under suppression, the sign changes with $\zeta$, so best responses switch curvature exactly at the pure/mixed boundary.
\end{proof}

\begin{proof}{Proof of Theorem}{\ref{thm:cinfo-effort}}
The proof treats the pure and mixed regions separately. In the pure region, equilibrium effort is the stable root of
\[
ax^{*2}-2bx^*+c-\theta=0.
\]
Implicit differentiation gives
\[
\frac{dx^*}{da}=\frac{x^{*2}}{2(b-ax^*)}>0,
\]
because the denominator is positive exactly on the second-order-stable branch. Thus effort rises as the contest moves from empowerment toward stronger suppression, up to the pure/mixed boundary. In the mixed region, Theorem \ref{thm:cinfo-unrestricted} replaces the fixed-point equation by the moment restriction $E[x]=b/a$. This region can occur only under suppression, so $a>0$ and $dE[x]/da=-b/a^2<0$. The pure and mixed formulas meet at $\zeta=0$, which is therefore the peak.
\end{proof}

\begin{proof}{Proof of Corollary}{\ref{cor:cinfo-max}}
Rent dissipation is total expected effort in this unit-prize normalization. Theorem \ref{thm:cinfo-unrestricted} gives total effort $2x^*$ on the pure branch and $2E[x]=2b/a$ on the mixed branch. Theorem \ref{thm:cinfo-effort} already shows that the first expression rises up to $\zeta=0$ while the second falls after $\zeta=0$. Multiplying by two preserves the same peak, and at the boundary the two expressions agree because $x^*=\kappa=b/a$.
\end{proof}

\begin{proof}{Proof of Proposition}{\ref{prop:cinfo-truncation}}
The pure part asks when the unrestricted symmetric candidate survives clipping. At $(x^*,x^*)$ the raw probability is $1/2$, so truncation is inactive at the candidate and the payoff is $1/2-\theta x^*$. For any deviation with $0\leq P\leq1$, the unrestricted best-response argument applies. If $P\geq1$, clipping lowers the deviator's winning probability relative to the raw polynomial; if $P\leq0$, the deviator's payoff is just $-\theta x\leq0$. Hence the candidate survives exactly when its own payoff is nonnegative:
\[
\frac12-\theta x^*\geq0.
\]
In the mixed region, on-domain mixing still imposes the unrestricted quadratic indifference conditions, so the moments must be $E[x]=\kappa$ and $\operatorname{Var}(x)=-\zeta$. The resulting common payoff is $1/2-\theta\kappa$. Since zero effort gives a nonnegative truncated payoff, participation requires $1/2\geq\theta\kappa$.
\end{proof}

\begin{proof}{Proof of Proposition}{\ref{prop:cinfo-two-point}}
The mixed region pins down only mean $\kappa$ and variance $-\zeta$, so the proof constructs minimal two-point distributions with those moments. With $s=\sqrt{-\zeta}$, the symmetric support $\{\kappa-s,\kappa+s\}$ is feasible exactly when $s\leq\kappa$, equivalently $a\leq2b^2/(c-\theta)$, and equal probabilities deliver mean $\kappa$ and variance $s^2$. When that lower support point would be negative, use the endpoint support $\{0,z\}$ and solve
\[
pz=\kappa,
\qquad
pz^2=\frac{c-\theta}{a},
\]
which gives $z=(c-\theta)/b$ and $p=b^2/[a(c-\theta)]$. At $a=2b^2/(c-\theta)$, the lower symmetric point is zero, the upper point equals $z$, and $p=1/2$, so the two descriptions match.
\end{proof}

\begin{proposition}[Branch-specific conditions for two-point families]\label{prop:cinfo-mixed-branches}
Suppose $a>0$ and $\zeta<0$. For the canonical two-point families, define
\[
\bar x=\max\left\{0,\kappa-\frac{\theta}{2as}\right\},
\qquad
\bar a_{M}=\frac{\left(b+c^{2}-3 c \theta+2 \theta^{2}+\sqrt{b \left(b-2 c \theta+2 \theta^{2} \right)}\right) b^{2}}{\left(c-\theta \right)^{3}}.
\]
Then the symmetric branch is an on-domain mixed equilibrium whenever
\[
P(0,\kappa-s)\geq 0
\qquad\text{and}\qquad
P(\bar x,\kappa+s)\geq 0.
\]
The endpoint branch is an on-domain mixed equilibrium whenever
\[
a\leq \bar a_{M}.
\]
Wherever the relevant branch condition holds, the truncated contest attains the unrestricted moment profile \eqref{mixed} and the common payoff $\tfrac{1}{2}-\theta\kappa$.
\end{proposition}

\begin{proof}{Proof of Proposition}{\ref{prop:cinfo-mixed-branches}}
Once the two-point distribution has the right moments, the unrestricted payoff is flat, so any profitable deviation must be created by clipping. On the symmetric branch, deviations above the mean are never helped by clipping, while deviations below the mean are safe if the two raw probabilities remain nonnegative at their worst points. On the endpoint branch, the right tail is automatically harmless, and the only remaining danger is that $P(x,z)$ dips below zero before the zero-support probability vanishes; the tangency where this first happens gives the cutoff $\bar a_M$.
\end{proof}

\begin{proof}{Proof of Proposition}{\ref{prop:cinfo-mixed-simple}}
The primitive inequalities make the branch-specific checks automatic on a nonempty part of the mixed region. The bound $c\leq2\theta$ keeps participation from being too demanding and ensures that the symmetric branch's critical point is at the endpoint $x=0$. The bound $b\geq c^2/2$ makes $P(0,y)=1/2-cy+by^2$ nonnegative for every $y$, so the symmetric branch stays away from the lower clipping boundary. On the endpoint branch, Proposition \ref{prop:cinfo-mixed-branches} applies up to $\bar a_M$, and the same primitive inequalities imply $\bar a_M\geq2b^2/(c-\theta)$, so the admissible interval is nonempty.
\end{proof}

\begin{proposition}[Failure of on-domain mixing under extreme suppression]\label{prop:cinfo-mixed-failure}
There exists a finite threshold $\bar a_{F}$ such that for every $a\geq \bar a_{F}$ the truncated contest admits no on-domain mixed equilibrium.
\end{proposition}

\begin{proof}{Proof of Proposition}{\ref{prop:cinfo-mixed-failure}}
The failure result shows that the sufficient restrictions cannot be dispensed with entirely. Under very strong suppression, construct a small deviation $x_a=\lambda/a$ that lies just above the mixed-equilibrium mean. Against moderate opponent efforts, clipping cannot lower this deviation's payoff; against sufficiently large opponent efforts, the raw probability becomes negative and clipping raises it to zero. The mixed moment conditions force every on-domain mixed equilibrium to put positive probability on that large-effort tail, so the clipping gain is strict on a positive-probability set. Hence the deviation is profitable once $a$ is large enough, and on-domain mixed equilibrium cannot survive arbitrarily strong suppression.
\end{proof}

\begin{proof}{Proof of Remark}{\ref{rem:cinfo-local-mixed}}
For local Nash equilibrium, one does not need global deviation checks. If the finite support of the canonical mixed strategy lies strictly inside $\mathcal D$, then sufficiently small deviations never activate truncation, and local best-response conditions coincide with the unrestricted ones. On the symmetric branch, the diagonal support profiles have probability $1/2$, and the only nontrivial check is
\[
P(\kappa-s,\kappa+s)=\frac12-2\theta s>0,
\]
which is equivalent to $s<1/(4\theta)$. On the endpoint branch, the diagonal profiles again have probability $1/2$, while
\[
P(0,z)=\frac12-\frac{\theta(c-\theta)}{b}>0,
\]
which is equivalent to $b>2\theta(c-\theta)$. The opposite off-diagonal profiles then lie in $\mathcal D$ by $P(x,y)=1-P(y,x)$.
\end{proof}

\subsubsection*{Incomplete-information results}

\begin{proof}{Proof of Theorem}{\ref{thm:bn-unrestricted}}
Let the opponent's effort be the random variable $Y$, with moments $E_1$ and $E_2$. In any symmetric equilibrium, a type $\theta$ who deviates to $\tilde x$ obtains interim utility of the form
\[
K(\theta)+(aE_1-b)\tilde x^2+(c-\theta-aE_2)\tilde x.
\]
Finite optimality for all types requires $aE_1<b$, so the best response is unique and affine in type:
\[
x(\theta)=\frac{aE_2+\theta-c}{2(aE_1-b)}=k\theta+d.
\]
Taking moments of this affine rule gives two equations in $E_1$ and $\operatorname{Var}(x)$. With $z:=E_1-\kappa$, they collapse to
\[
z^4-\zeta z^2-\frac{\omega}{4}=0.
\]
The negative quadratic branch is impossible for $z^2$, and the second-order condition selects the sign of $z$ opposite to the sign of $a$. This yields
\[
E_1=\kappa-\frac{\operatorname{sgn}(a)}{\sqrt2}\sqrt{\zeta+\sqrt{\zeta^2+\omega}},
\]
and then $k$, $d$, and $\operatorname{Var}(x)$ follow from the moment identities. Since any symmetric equilibrium must pass through this system, the equilibrium is unique.
\end{proof}

\begin{proof}{Proof of Remark}{\ref{rem:bn-zero}}
At $a=0$, the raw interim payoff is strictly concave with first-order condition $c-\theta-2bx=0$. Hence every type has the dominant strategy $x(\theta)=(c-\theta)/(2b)>0$, and expected effort is $(c-M_1)/(2b)$. For two equilibrium types $\theta$ and $\theta'$, the raw probability is $P(x(\theta),x(\theta'))=1/2+((\theta')^2-\theta^2)/(4b)$, so the affine support is strictly on-domain exactly when $\beta^2-\alpha^2<2b$. The formula from Theorem \ref{thm:bn-unrestricted} has the same limit as $a\to0$, so the neutral case is the continuous boundary of the affine benchmark.
\end{proof}

\begin{proof}{Proof of Remark}{\ref{rem:bn-vanishing}}
Letting $\sigma_\theta^2\to0$ sends $\omega$ to zero while $\kappa$ and $\zeta$ remain fixed. The affine formulas from Theorem \ref{thm:bn-unrestricted} and the display following it then give the stated limits: if $\zeta\geq 0$, effort converges to the complete-information pure action with zero variance; if $\zeta<0$, effort converges in moments to mean $\kappa$ and variance $-\zeta$. Thus vanishing uncertainty selects the complete-information moment profile generated by the affine Bayesian limit, and any shrinking prior family with convergent standardized types selects the corresponding affine push-forward distribution.
\end{proof}

\begin{proof}{Proof of Theorem}{\ref{thm:uncertainty}}
Holding $a$, $b$, $c$, and $M_1$ fixed, the expected-effort formula depends on type dispersion only through
\[
\omega=\frac{\sigma_\theta^2}{a^2}.
\]
Differentiating the formula in Theorem \ref{thm:bn-unrestricted} gives
\[
\frac{\partial E_1}{\partial\omega}
=-\frac{\operatorname{sgn}(a)}{4\sqrt2\sqrt{\zeta^2+\omega}\sqrt{\zeta+\sqrt{\zeta^2+\omega}}}.
\]
The denominator is positive, so the sign is $-\operatorname{sgn}(a)$. Since $\omega$ is increasing in $\sigma_\theta^2$, uncertainty lowers expected effort under suppression and raises it under empowerment; at $a=0$ expected effort depends only on the mean type.
\end{proof}

\begin{proof}{Proof of Theorem}{\ref{thm:bn-disclosure}}
A signal policy replaces the true type by the posterior mean $\hat\theta_\pi=\mathrm{E}[\theta\mid s_\pi]$ in interim expected payoffs. Symmetry and Bayes plausibility keep the mean fixed, and the affine formula from Theorem \ref{thm:bn-unrestricted} makes expected effort depend on the signal only through the variance of this posterior-mean distribution:
\[
0\leq \operatorname{Var}(\hat\theta_\pi)\leq\operatorname{Var}(\theta).
\]
No disclosure attains the lower bound and full disclosure attains the upper bound; the lower-bound case is understood as the continuous limit of the affine formula. The designer therefore minimizes posterior-mean variance under suppression, maximizes it under empowerment, and is indifferent at neutrality. A Blackwell refinement raises the variance of posterior means by the law of total variance, giving the final monotonicity claim.
\end{proof}

\begin{proof}{Proof of Theorem}{\ref{thm:bn-a-peak}}
Normalize expected effort by
\[
\rho=\frac{\sigma_\theta^2}{\Delta^2},
\qquad
\xi=\frac{\Delta a}{b^2},
\qquad
E_1(a)=\frac{\Delta}{b}e(\xi).
\]
Because $\Delta/b>0$ and $b^2/\Delta>0$, maximizing $E_1$ over $a$ is equivalent to maximizing $e$ over $\xi$. The positive-average region $a<\bar a$ is exactly $\xi<4/\rho$.

Write the radical term in \eqref{eq:bn-e1-a} as $y(\xi)$, so $e(\xi)=(1-y(\xi))/\xi$. Inverting $\xi\mapsto y$ gives two branches, $\xi_-(y)$ and $\xi_+(y)$, meeting at $\xi=2/(1+\rho)$. The minus branch is decreasing in $y$, while the plus branch is increasing and already has $e'(\xi)<0$. Thus any interior maximum must come from the minus branch. On that branch the derivative calculation, with the reversed orientation of $\xi_-$ taken into account, gives
\[
\operatorname{sgn} e'(\xi)=\operatorname{sgn}\{g(y)-\rho\},
\qquad
g(y)=\frac{y^3(2-y)}{(y^2-y-1)^2}.
\]
The function $g$ is strictly increasing from $0$ to $\infty$ on $(0,\varphi)$ and satisfies $g(1)=1$. Hence $g(y)=\rho$ has a unique solution $y^\dagger$. A direct endpoint check gives $g(y_0)<\rho$ at the branch point $y_0=\sqrt{\rho/(1+\rho)}$, so $y^\dagger$ is admissible on the minus branch. Since $\xi_-$ is decreasing, $e$ rises before $\xi^\dagger=\xi_-(y^\dagger)$ and falls after it; the plus branch continues falling, so this point is the global maximizer.

Finally, $\xi^\dagger<\xi_-(y_0)=2/(1+\rho)$, which gives $a^\dagger<2b^2\Delta/(\Delta^2+\sigma_\theta^2)$. Also, the sign of $\xi_-(y)$ is the sign of $1-y$, so $a^\dagger$ is positive, zero, or negative according as $y^\dagger$ is below, equal to, or above $1$. Since $g(1)=1$, this is equivalent to $\rho<1$, $\rho=1$, or $\rho>1$.
\end{proof}

\begin{proof}{Proof of Remark}{\ref{rem:bn-positive}}
The claim about $\bar a$ is a sign calculation from \eqref{eq:bn-e1-a}. If $a<0$, the radical term is larger than $b$, so numerator and denominator are both negative and $E_1(a)>0$. If $a=0$, $E_1(0)=\Delta/(2b)>0$. If $a>0$, positivity is equivalent to the radical term being smaller than $b$, which reduces to $\sigma_\theta^2 a<4b^2\Delta$, or $a<\bar a$. Equality gives $E_1=0$. The remaining statements distinguish this average-effort positivity condition from full activity and from literal truncation.
\end{proof}

\begin{proof}{Proof of Theorem}{\ref{thm:bn-peak-variance}}
Use the critical-point representation from Theorem \ref{thm:bn-a-peak}. Set
\[
\rho=\frac{\sigma_\theta^2}{\Delta^2},
\qquad
A(\rho)=\frac{\Delta a^\dagger}{b^2}.
\]
If $y^\dagger$ solves $g(y^\dagger)=\rho$, substitution into the branch formula for $\xi^\dagger$ gives
\[
A(\rho)=\frac{2(1-y^\dagger)(1+y^\dagger-(y^\dagger)^2)}{2-y^\dagger}.
\]
Because $g$ is strictly increasing, $y^\dagger$ is increasing in $\rho$; hence the sign of $A'(\rho)$ is the sign of $dA/dy$. Differentiating the displayed expression gives
\[
\frac{dA}{dy}=-\frac{2q(y)}{(2-y)^2},
\qquad
q(y)=2y^3-8y^2+8y-1.
\]
The denominator is positive. The cubic $q$ has exactly two roots in $(0,\varphi)$: one in $(0,2/3)$ and one in $(1,\varphi)$, since $q'(y)=2(y-2)(3y-2)$, $q(0)<0$, $q(2/3)>0$, $q(1)>0$, and $q(\varphi)<0$. Call them $y_1<y_2$, set $\rho_i=g(y_i)$, and define $\sigma_i^2=\rho_i\Delta^2$. Since $g$ is increasing and $g(1)=1$, we have $0<\sigma_1^2<\Delta^2<\sigma_2^2$.

The sign pattern of $q$ implies that $A$, and therefore $a^\dagger=(b^2/\Delta)A$, is increasing, then decreasing, then increasing as $\rho$ (equivalently $\sigma_\theta^2$) rises. Finally, $y^\dagger\downarrow0$ as $\rho\downarrow0$ and $y^\dagger\uparrow\varphi$ as $\rho\to\infty$; the displayed formula for $A$ then gives $A\to1$ and $A\to0^-$, which yields the two endpoint limits for $a^\dagger$.
\end{proof}

\begin{proof}{Proof of Proposition}{\ref{prop:bn-nonnegative}}
Let $E_1$ and $E_2$ be the opponent's first two effort moments. A type-$\theta$ deviation has interim payoff
\[
U(\theta,\tilde x)=K(\theta)+(aE_1-b)\tilde x^2+(c-\theta-aE_2)\tilde x.
\]
The all-zero profile cannot be an equilibrium because $c-\theta>0$, so $E_1>0$. Equilibrium also requires $aE_1\leq b$; otherwise utility is unbounded above in own effort. If $aE_1<b$, the payoff is strictly concave and the unique best response is cutoff-affine:
\[
x(\theta)=\lambda(t-\theta)_+,
\qquad
\lambda=\frac{1}{2(b-aE_1)},
\qquad
t=c-aE_2.
\]
Writing
\[
A(t)=\mathrm{E}[(t-\theta)_+],
\qquad
B(t)=\mathrm{E}[(t-\theta)_+^2],
\]
and substituting $E_1=\lambda A(t)$ and $E_2=\lambda^2B(t)$ gives the fixed-point system
\[
2aA(t)\lambda^2-2b\lambda+1=0,
\qquad
t=c-aB(t)\lambda^2.
\]
Conversely, any solution with $a\lambda A(t)<b$ makes active types satisfy the first-order condition and inactive types have nonpositive derivative at zero, so it is an equilibrium.

If $aE_1=b$, the payoff is linear. Its slope must be nonpositive for every type and zero for any type who exerts positive effort. Since the slope is strictly decreasing in $\theta$, only the lowest type can be active; this requires an atom $m=F(\{\alpha\})>0$. The moment conditions for such an equilibrium are feasible exactly when $am(c-\alpha)\geq b^2$, and the construction in the statement---type $\alpha$ mixing on $\{0,(c-\alpha)/b\}$, with all higher types choosing zero---delivers it.

When this boundary case is unavailable, existence in the cutoff-affine class follows from the fixed-point system. For $a=0$ the dominant-strategy equilibrium from Remark \ref{rem:bn-zero} applies; for $a<0$ the affine equilibrium is strictly positive because $c-\theta-aE_2>0$ and $b-aE_1>0$. For $a>0$, eliminating $\lambda$ gives a scalar cutoff equation
\[
H(t):=2(c-t)A(t)+B(t)-2b\sqrt{\frac{(c-t)B(t)}{a}}=0.
\]
If the boundary condition fails, then $H(t)<0$ near $\alpha$, while $H(c)=B(c)>0$, so a root exists by continuity. The resulting cutoff-affine equilibrium is fully active when $t\geq\beta$ and has dropout when $t<\beta$.
\end{proof}

\begin{proof}{Proof of Proposition}{\ref{prop:bn-dropout}}
Because the prior is atomless, the exceptional boundary case in Proposition \ref{prop:bn-nonnegative} is unavailable. On a partially active branch with cutoff $t<\beta$, write
\[
A_F(t)=\mathrm{E}_F[(t-\theta)_+],
\qquad
B_F(t)=\mathrm{E}_F[(t-\theta)_+^2],
\qquad
D_F(t)=B_F(t)+2(c-t)A_F(t).
\]
Solving the cutoff-affine fixed-point equations for $\lambda$ and $a$ gives
\[
\lambda_F(t)=\frac{D_F(t)}{2bB_F(t)},
\qquad
a_F(t)=\frac{4b^2(c-t)B_F(t)}{D_F(t)^2}.
\]
Differentiation, together with Cauchy--Schwarz, gives $a_F'(t)<0$, so stronger suppression corresponds to a lower cutoff. Moreover $a_F(t)\to\infty$ as $t\downarrow\alpha$. Therefore $a_D(F):=a_F(\beta)$ is the full-activity threshold, and for $a>a_D(F)$ there is a unique partially active cutoff $t(a)\in(\alpha,\beta)$ with $t'(a)<0$. The dropout rate $\delta(a)=1-F(t(a))$ is then strictly increasing, continuously differentiable, and tends to one as $t(a)\downarrow\alpha$.
\end{proof}

\begin{remark}[Dropout threshold calibration]\label{rem:dropout-calibration}
The dropout threshold $a_D(F)$ is governed by lower partial moments of the prior. For a prior $F$, define
\[
A_F(t):=\mathrm{E}_F[(t-\theta)_+],
\qquad
B_F(t):=\mathrm{E}_F[(t-\theta)_+^2],
\qquad
D_F(t):=B_F(t)+2(c-t)A_F(t).
\]
Then the partially active branch is described by
\[
a_F(t):=\frac{4b^2(c-t)B_F(t)}{D_F(t)^2},
\qquad
\lambda_F(t):=\frac{D_F(t)}{2bB_F(t)},
\]
and
\[
a_D(F)=a_F(\beta).
\]
Thus $a_D(F)$ has no distribution-free closed form; it is a functional of the lower partial moments of $F$, and the cutoff is determined implicitly by $a=a_F(t)$.

For a simple calibration, let $F$ be uniform on $[\alpha,\beta]$. Write $h:=\beta-\alpha$ and $q:=c-\beta>0$, so $\Delta=q+h/2$. Since $A_F(\beta)=h/2$, $B_F(\beta)=h^2/3$, and $D_F(\beta)=h(q+h/3)$,
\[
a_D(F)=\frac{4b^2q}{3(q+h/3)^2},
\qquad
\frac{a_D(F)}{b^2/\Delta}
=\frac{4q(q+h/2)}{3(q+h/3)^2}.
\]
Therefore $a_D(F)\leq b^2/\Delta$ if and only if $q\leq h/\sqrt{3}$. In the mean-preserving uniform family $[M_1-h/2,M_1+h/2]$, this is equivalent to $h\geq \Delta/(1/2+1/\sqrt{3})\approx0.928\Delta$, or $\sigma_\theta^2=h^2/12\gtrsim0.0718\Delta^2$. Thus, for this family, dropout does not bind on the initial increasing region $\sigma_\theta^2<\sigma_1^2$, where $\sigma_1^2\approx0.0045\Delta^2$; only once $\sigma_\theta^2\gtrsim0.072\Delta^2$, well past $\sigma_1^2$, does dropout begin to interact with the peak.
\end{remark}

\begin{proof}{Proof of Remark}{\ref{rem:bn-dropout}}
By Proposition~\ref{prop:bn-nonnegative} the equilibrium is cutoff-affine outside the boundary case. Under empowerment and neutrality the cutoff lies beyond the support, and by Proposition~\ref{prop:bn-dropout} this persists under suppression up to $a_D(F)$; there the affine comparative statics apply directly, subject only to the literal-contest support checks. For $a>a_D(F)$ the cutoff drops below $\beta$, so the equilibrium moments become lower partial moments over active types. The explicit formulas follow by substituting these into the cutoff-affine fixed-point system; Appendix Remark~\ref{rem:dropout-calibration} records them.
\end{proof}

\begin{proof}{Proof of Remark}{\ref{rem:bn-dropout-peak}}
On the partially active branch, write $s=t-1$ and let $A(s)=\mathrm{E}[(s-Z)_+]$, $B(s)=\mathrm{E}[(s-Z)_+^2]$, $C(s)=c-1-s$, $D(s)=B(s)+2C(s)A(s)$. Proposition~\ref{prop:bn-dropout} gives $E_1(s)=A(s)D(s)/(2bB(s))$ with $a(s)=4b^2C(s)B(s)/D(s)^2$ strictly decreasing in $s$. Differentiating, $E_1'(s)/E_1(s)=F(s)/A(s)+2C(s)F(s)/D(s)-2A(s)/B(s)$, where $F(s)=\Pr(Z\le s)$. For the stated mixture this expression changes sign three times on $(0,1)$, giving---since $a'(s)<0$---two local maxima of $E_1$ in $a$ separated by a local minimum. The detailed proof records the incomplete-beta evaluation and the numerical values.
\end{proof}

\begin{proposition}[Support-based truncation criterion]\label{prop:bn-truncation-detailed}
Assume Proposition \ref{prop:bn-nonnegative} is in its fully active region. Then equilibrium follows the affine benchmark rule
\[
x^{u}(\theta)=k\theta+d.
\]
Write
\[
\underline x:=x^{u}(\beta),
\qquad
\overline x:=x^{u}(\alpha).
\]
If
\[
\overline x\leq \frac{1}{\alpha},
\qquad
a\leq b\alpha,
\qquad
c\alpha^2-2b\alpha+a\geq 0,
\]
and
\[
\min_{y\in[\underline x,\overline x]}\left(\frac{1}{2}-cy+by^2\right)\geq 0,
\]
then the same affine rule $x^{u}(\theta)$ is a symmetric Bayesian Nash equilibrium of the truncated contest based on $\bar P$.
\end{proposition}

\begin{proof}{Proof of Proposition}{\ref{prop:bn-truncation-detailed}}
Let $L:=1/\alpha$ and let $Y=x^u(\theta')$ denote the opponent's effort. The candidate support is $[\underline x,\overline x]\subset[0,L]$. For each on-support $y$, the raw probability $P(\tilde x,y)$ is concave in the deviating effort on $[0,L]$, because $ay-b\leq aL-b\leq0$. Hence it is enough to check the two endpoints. The left endpoint is exactly the left-edge criterion in the statement. At the right endpoint, writing $z=L-y\geq0$ gives
\[
P(L,y)=\frac12+z(c-2bL+aL^2)+z^2(b-aL),
\]
which is nonnegative by $a\leq b\alpha$ and $c\alpha^2-2b\alpha+a\geq0$. Thus $P(\tilde x,y)\geq0$ for every $\tilde x\in[0,L]$ and every on-support $y$; by symmetry, the affine candidate itself also has raw probabilities at most one.

Truncation is therefore inactive at the candidate. Deviations above $L$ are dominated by zero because their cost exceeds the prize for every type. For deviations in $[0,L]$, clipping can only weakly lower payoff relative to the untruncated polynomial game, and the affine rule is already the untruncated best response. Hence no deviation is profitable in the truncated contest.
\end{proof}

\begin{proof}{Proof of Proposition}{\ref{prop:bn-truncation}}
This is the preceding support-based criterion with a stronger left-edge condition. The assumption $c^2\leq2b$ makes $P(0,y)=1/2-cy+by^2$ nonnegative for every real $y$, so the left-edge criterion in Proposition \ref{prop:bn-truncation-detailed} is automatic.
\end{proof}

\begin{proof}{Proof of Proposition}{\ref{prop:bn-disclosure-sufficient}}
For each signal rule, if the induced posterior-mean contest is fully active and truncation is inactive on its equilibrium support, then the literal truncated contest and the affine benchmark generate the same equilibrium effort under that signal. Therefore the signal ranking from Theorem \ref{thm:bn-disclosure} carries over signal by signal. Under empowerment full activity is automatic; at $a=0$, full activity is automatic but truncation still requires a support check; under suppression both conditions must be verified signal by signal.
\end{proof}

\pagebreak 

\section*{Online Appendix: Detailed Proofs}

\begin{proof}{Detailed proof of Observation}{\ref{obs:quadratic}}
Let $R(x,y):=P(x,y)-1/2$. The symmetry condition $P(x,y)=1-P(y,x)$ is equivalent to
\[
R(x,y)=-R(y,x),
\]
so $R$ is antisymmetric and therefore vanishes on the diagonal $x=y$. Hence $R$ is divisible by $x-y$, and we can write
\[
R(x,y)=(x-y)Q(x,y)
\]
for some polynomial $Q$. If $P$ has degree at most two, then $Q$ has degree at most one. Antisymmetry of $R$ also implies $Q(x,y)=Q(y,x)$, so a linear $Q$ must have the form
\[
Q(x,y)=q_0+q_1(x+y).
\]
Thus
\[
P(x,y)=\frac12+q_0(x-y)+q_1(x^2-y^2),
\]
and consequently
\[
\frac{\partial^2 P}{\partial x\,\partial y}\equiv 0.
\]
Therefore no polynomial CSF of degree at most two can combine symmetry with a non-trivial cross-effect; cubic is the minimal possible degree.
\end{proof}

\begin{proof}{Detailed proof of Observation}{\ref{obs:tullock-suppressive}}
Let
\[
P(x,y)=\frac{x^r}{x^r+y^r},
\qquad r>0,
\]
on the positive orthant. Differentiating first with respect to $x$ gives
\[
P_x(x,y)=\frac{r x^{r-1}y^r}{(x^r+y^r)^2}.
\]
Differentiating this expression with respect to $y$ yields
\[
P_{xy}(x,y)
=r x^{r-1}\left[
\frac{r y^{r-1}}{(x^r+y^r)^2}
-\frac{2r y^{2r-1}}{(x^r+y^r)^3}
\right].
\]
Putting the terms over a common denominator gives
\[
P_{xy}(x,y)
=\frac{r^2x^{r-1}y^{r-1}\bigl(x^r+y^r-2y^r\bigr)}{(x^r+y^r)^3}
=\frac{r^2x^{r-1}y^{r-1}(x^r-y^r)}{(x^r+y^r)^3}.
\]
For $x,y>0$, the factor
\[
\frac{r^2x^{r-1}y^{r-1}}{(x^r+y^r)^3}
\]
is strictly positive. Hence the sign of $P_{xy}$ is the sign of $x^r-y^r$, which is the sign of $x-y$. Therefore $P_{xy}>0$ exactly when $x>y$ and $P_{xy}<0$ exactly when $x<y$. In the terminology of the paper, the generalized Tullock contest is suppressive throughout its domain.
\end{proof}

\begin{proof}{Detailed proof of Observation}{\ref{obs:domain-geometry}}
Write $z:=(x+y)/2$ and $d:=x-y$, so that $x=z+d/2$ and $y=z-d/2$. The variable $z$ measures movement along the diagonal $x=y$ --- the overall effort level --- whereas $d$ measures the signed gap between the two efforts. The constraint $z\geq |d|/2$ encodes nonnegative efforts. In these coordinates,
\[
P(z,d)=\frac{1}{2}+d\left(c-2bz+az^{2}-\frac{a}{4}d^{2}\right),
\]
so the diagonal becomes the center line $d=0$, on which $P\equiv 1/2$; in particular every symmetric profile lies in $\mathcal D$. Since $P(z,d) = 1 - P(z,-d)$, the boundary curves $P=0$ and $P=1$ are reflections of each other across the diagonal, and $\mathcal D$ is symmetric across it.

Now fix a non-diagonal ray $(x,y)=r(u,v)$ with $u,v\geq 0$, $u\neq v$, and $r\geq 0$. Along this ray,
\[
P(ru,rv)=\frac{1}{2}+r(u-v)\left(c-br(u+v)+ar^2uv\right).
\]
If $a uv\neq 0$, the cubic term $a uv(u-v)r^3$ is nonzero and dominates. If $a uv=0$, the quadratic term $-b(u-v)(u+v)r^2$ is nonzero and dominates. Hence $|P(ru,rv)|\to\infty$ as $r\to\infty$ on every non-diagonal ray, so $\mathcal D$ is bounded along any such direction. Along the diagonal itself, $P\equiv 1/2$, so $\mathcal D$ is unbounded along it.

It remains to show non-convexity. Fix any signed effort gap $d\neq0$ and let $z\to\infty$. From the expression above,
\[
P(z,d)=\frac12+d\left(c-2bz+az^2-\frac{a}{4}d^2\right).
\]
If $a\neq0$, the term $adz^2$ dominates; if $a=0$, the term $-2bdz$ dominates. Hence for every fixed $d\neq0$, $P(z,d)$ eventually leaves $[0,1]$ as $z\to\infty$.

Choose any point $(z_0,d_0)$ with $d_0\neq0$ that lies in $\mathcal D$; such points exist because $P(z_0,d)\to1/2$ as $d\to0$. For any sufficiently large $Z$, the diagonal point $(Z,0)$ also lies in $\mathcal D$. The midpoint in $(z,d)$ coordinates is
\[
\left(\frac{z_0+Z}{2},\frac{d_0}{2}\right).
\]
Its gap $d_0/2$ is fixed and nonzero, while its average effort tends to infinity as $Z\to\infty$. By the preceding paragraph, this midpoint eventually lies outside $\mathcal D$. Thus $\mathcal D$ contains two points whose midpoint is outside $\mathcal D$, so $\mathcal D$ is non-convex.
\end{proof}

\begin{proof}{Detailed proof of Theorem}{\ref{thm:cinfo-unrestricted}}

We fix the opponent's effort at $y$. Player $X$'s expected utility in the unconstrained contest is
\[
U(x,y)=\frac{1}{2}+\left(x-y\right)\left(c-b(x+y)+axy\right)-\theta x,
\]
so
\[
U_x(x,y)=2axy-a y^2-2bx+c-\theta,
\qquad
U_{xx}(x,y)=2ay-2b.
\]
Hence the interior best response is
\[
x^{BR}(y)=\frac{a y^{2}-(c-\theta)}{2(a y-b)}.
\]
For an optimal interior reply, the second-order condition requires
\[
2ay-2b<0.
\]

Now consider any pure-strategy equilibrium $(x,y)$ of the unconstrained contest. The first-order conditions for the two players are
\[
2axy-a y^2-2bx+c-\theta=0,
\qquad
2axy-a x^2-2by+c-\theta=0.
\]
Subtracting gives
\[
(x-y)\bigl(a(x+y)-2b\bigr)=0.
\]
Thus either $x=y$ or $x+y=2\kappa$.

If $x=y$, then the first-order condition reduces to
\[
ax^2-2bx+c-\theta=0.
\]
Real roots exist exactly when $\zeta\geq 0$. The two roots are
\[
x=\frac{b\pm\sqrt{b^2-a(c-\theta)}}{a}.
\]
If $a>0$, the symmetric second-order condition is $ax-b<0$, i.e. $x<\kappa=b/a$, so only the smaller root is admissible. If $a<0$, then $\zeta=\kappa^2+(c-\theta)/|a|>\kappa^2$, so $\kappa-\sqrt{\zeta}<0$ while $\kappa+\sqrt{\zeta}>0$; again only one root is admissible. Therefore
\[
x^*=y^*=\frac{b-\sqrt{b^2-a(c-\theta)}}{a}=\kappa-\mathrm{sgn}(a)\sqrt{\zeta}.
\]
The alternative root either violates the second-order condition or is negative.

If $x\neq y$, then $x+y=2\kappa$. Writing $x=\kappa+d/2$ and $y=\kappa-d/2$ and substituting back into the first-order conditions yields
\[
3a^{2}d^{2}=4\bigl(a(c-\theta)-b^{2}\bigr),
\]
that is,
\[
d^2=-\frac{4\zeta}{3}.
\]
Hence asymmetric pure candidates can arise only if $\zeta<0$, which in turn requires $a>0$. But under suppression the second-order conditions are $y<\kappa$ for player $X$ and $x<\kappa$ for player $Y$. These inequalities cannot both hold when $x+y=2\kappa$. Therefore at least one second-order condition is violated, so no asymmetric pure equilibrium exists. This proves uniqueness of the pure equilibrium whenever $\zeta\geq 0$, and shows that no pure equilibrium exists when $\zeta<0$.

Now consider any mixed-strategy equilibrium of the untruncated contest, and let $\mu_X$ and $\mu_Y$ denote the two players' distributions.
We first rule out the possibility that one player mixes against a degenerate opponent. Suppose $\mu_Y$ is degenerate at a pure action $y$ while $\mu_X$ is nondegenerate. Then player $X$'s payoff from choosing $x$ is
\[
U_X(x;y)=(ay-b)x^2+(c-\theta-ay^2)x+\frac{1}{2}-cy+by^2.
\]
If $\mu_X$ is nondegenerate, every point in its support must be a best response, so this quadratic must be constant. Hence
\[
ay-b=0,
\qquad
c-\theta-ay^2=0.
\]
Thus $y=\kappa$ and $y^2=(c-\theta)/a$, which implies $\zeta=0$. But then any distribution with mean $\kappa$ and second moment $(c-\theta)/a=\kappa^2$ has zero variance, contradicting nondegeneracy of $\mu_X$. Therefore any genuine mixed equilibrium must have both players nondegenerate.

Given nondegeneracy, player $X$'s expected payoff from choosing $x$ against $\mu_Y$ depends only on player $Y$'s first two moments,
\[
E_1^Y=\int y\,d\mu_Y(y),
\qquad
E_2^Y=\int y^2\,d\mu_Y(y),
\]
and is given by the quadratic
\[
U_X(x;E_1^Y,E_2^Y)=\left(-a x +b \right) E_2^Y +\left(a x^{2}-c\right) E_1^Y +\frac{1}{2}-b x^{2}+(c-\theta)x.
\]
If $\mu_X$ is nondegenerate, every point in its support must be a best response to $\mu_Y$. A quadratic on $\mathbb{R}$ has at most one maximizer unless it is constant. Therefore $U_X(\cdot;E_1^Y,E_2^Y)$ must be constant, which implies
\[
aE_1^Y-b=0,
\qquad
-aE_2^Y+c-\theta=0.
\]
Hence
\[
E_1^Y=\frac{b}{a}=\kappa,
\qquad
E_2^Y=\frac{c-\theta}{a}.
\]
The same argument applied to player $Y$ shows that if $\mu_Y$ is nondegenerate, then player $X$'s moments satisfy
\[
E_1^X=\frac{b}{a}=\kappa,
\qquad
E_2^X=\frac{c-\theta}{a}.
\]
Hence every genuine mixed equilibrium must satisfy
\[
\mathrm{E}[x]=\mathrm{E}[y]=\kappa,
\qquad
\mathrm{Var}(x)=\mathrm{Var}(y)=\frac{c-\theta}{a}-\frac{b^2}{a^2}=-\zeta.
\]
If $\zeta>0$, this is impossible for a nondegenerate distribution. If $\zeta=0$, it forces both distributions to be degenerate at $\kappa=x^*$. Thus when $\zeta\geq 0$ the equilibrium found above is the unique equilibrium of the unrestricted contest.

Now suppose $\zeta<0$. If player $Y$'s first two moments satisfy
\[
E_1^Y=\kappa,
\qquad
E_2^Y=\frac{c-\theta}{a},
\]
then
\[
U_X(x;E_1^Y,E_2^Y)=\frac{1}{2}-\theta\kappa
\qquad\text{for all }x\ge 0.
\]
So every action is a best response to $\mu_Y$. The same argument applies symmetrically to player $Y$, so any pair of distributions with the stated moments is an unrestricted mixed equilibrium.
It remains only to show that such distributions exist when $\zeta<0$. Write
\[
s:=\sqrt{-\zeta}>0.
\]
If $s\leq \kappa$, equivalently $-\zeta\leq \kappa^2$, then the two-point distribution on $\{\kappa-s,\kappa+s\}$ with equal mass has mean $\kappa$ and variance $s^2=-\zeta$, and its support is contained in $[0,\infty)$. If instead $s>\kappa$, equivalently $-\zeta>\kappa^2$, then necessarily
\[
a>\frac{2b^2}{c-\theta}.
\]
Define
\[
z:=\frac{c-\theta}{b},
\qquad
p:=\frac{b^2}{a(c-\theta)}.
\]
Since $\zeta<0$ implies $a>b^2/(c-\theta)$, one has $p\in(0,1)$. The two-point distribution that assigns probability $1-p$ to $0$ and probability $p$ to $z$ then satisfies
\[
\mathrm{E}[X]=pz=\frac{b}{a}=\kappa,
\]
and
\[
\mathrm{Var}(X)=pz^2-\kappa^2=
\frac{c-\theta}{a}-\frac{b^2}{a^2}=-\zeta.
\]
Thus in either case there exists a nondegenerate distribution on $[0,\infty)$ with the required moments. Pairing any two such distributions yields an unrestricted mixed equilibrium by the indifference argument above. Therefore the unrestricted contest admits mixed equilibria whenever $\zeta<0$, and every pair of distributions on $[0,\infty)$ with those moments is an unrestricted mixed equilibrium. This completes the proof of the theorem.
\end{proof}

\begin{proof}{Detailed proof of Remark}{\ref{rem:cinfo-moment-unique}}
In the mixed region, the proof of Theorem \ref{thm:cinfo-unrestricted} shows that expected payoff against an opponent's mixed strategy is the quadratic
\[
U_X(x;E_1^Y,E_2^Y)=\left(aE_1^Y-b\right)x^2+\left(c-\theta-aE_2^Y\right)x+K_Y,
\]
where $K_Y$ is independent of $x$. Nondegenerate mixing requires this quadratic to be constant, so it imposes exactly two restrictions on the opponent's distribution:
\[
E_1^Y=\kappa,
\qquad
E_2^Y=\frac{c-\theta}{a}.
\]
Equivalently, the mean is $\kappa$ and the variance is
\[
E_2^Y-(E_1^Y)^2=\frac{c-\theta}{a}-\frac{b^2}{a^2}=-\zeta.
\]
No higher moment or other feature of the distribution enters expected payoff. Thus equilibrium is unique in moments but not in distributions: any two distributions on $[0,\infty)$ with these same moments generate the same indifference condition and hence the same unrestricted mixed-equilibrium payoff.

The convergence statement follows from the Bayesian formulas in Theorem \ref{thm:bn-unrestricted}. As $\sigma_\theta^2\to0$, one has $\omega\to0$. If $\zeta>0$, then
\[
\mathrm{E}[x]\to \kappa-\mathrm{sgn}(a)\sqrt{\zeta}=x^*,
\qquad
\mathrm{Var}(x)\to0.
\]
If $\zeta<0$, then
\[
\mathrm{E}[x]\to\kappa,
\qquad
\mathrm{Var}(x)\to-\zeta.
\]
These limits are exactly the complete-information pure and mixed moment benchmarks.
\end{proof}

\begin{proof}{Detailed proof of Remark}{\ref{rem:cinfo-zero}}
Set $a=0$. Then player $X$'s payoff against any opponent effort $y$ is
\[
U(x,y)=\frac12+(x-y)(c-b(x+y))-\theta x,
\]
so
\[
U_x(x,y)=c-\theta-2bx,
\qquad
U_{xx}(x,y)=-2b<0.
\]
The unique best response is independent of $y$ and equals
\[
x^*=\frac{c-\theta}{2b}.
\]
The same argument applies to player $Y$, so this is a dominant-strategy equilibrium and is unique. It is strictly positive whenever $c>\theta$.

For continuity of the pure-equilibrium formula, rationalize the root:
\[
\frac{b-\sqrt{b^2-a(c-\theta)}}{a}
=
\frac{c-\theta}{b+\sqrt{b^2-a(c-\theta)}}.
\]
Letting $a\to0$ gives $(c-\theta)/(2b)$. Therefore the pure equilibrium extends continuously through $a=0$, and because the pure-equilibrium condition $\zeta\geq0$ is satisfied for all sufficiently small nonzero $a$, the unrestricted equilibrium is pure and unique in a neighborhood of the degenerate case.
\end{proof}

\begin{proof}{Detailed proof of Corollary}{\ref{cor:cinfo-br-curvature}}
From the proof of Theorem \ref{thm:cinfo-unrestricted}, the interior best response is
\[
x^{BR}(y)=\frac{a y^{2}-(c-\theta)}{2(a y-b)}.
\]
Using $\kappa=b/a$ and $\zeta=\kappa^{2}-(c-\theta)/a$, this can be rewritten as
\[
x^{BR}(y)=\frac{y+\kappa}{2}+\frac{\zeta}{2(y-\kappa)}.
\]
Differentiating twice gives
\[
\frac{dx^{BR}(y)}{dy}=\frac{1}{2}-\frac{\zeta}{2(y-\kappa)^{2}},
\qquad
\frac{d^{2}x^{BR}(y)}{dy^{2}}=\frac{\zeta}{(y-\kappa)^{3}}=\frac{a^{3}\zeta}{(ay-b)^{3}}.
\]
On the interior branch, the second-order condition is $ay-b<0$, so the denominator is negative. If $a<0$, then $\zeta>0$ by definition, hence $d^{2}x^{BR}(y)/dy^{2}>0$ and best responses are convex. If $a>0$, then the sign is the opposite of the sign of $\zeta$: best responses are concave when $\zeta>0$ and convex when $\zeta<0$. The change occurs exactly at $\zeta=0$, and the final sentence follows because empowerment always has $\zeta>0$, while under suppression the unrestricted equilibrium is pure precisely when $\zeta\geq 0$.
\end{proof}

\begin{proof}{Detailed proof of Theorem}{\ref{thm:cinfo-effort}}
First consider the interior of the pure-equilibrium region, where $\zeta>0$ and the equilibrium is the symmetric pure action $x^*$. By Theorem \ref{thm:cinfo-unrestricted}, $x^*$ is the second-order-stable root of
\[
ax^{*2}-2bx^*+c-\theta=0.
\]
Equivalently, if $F(a,x):=ax^2-2bx+c-\theta$, then $F(a,x^*(a))=0$. Along the stable branch,
\[
F_x(a,x^*)=2(ax^*-b)<0,
\]
so the implicit function theorem applies away from the boundary $\zeta=0$. Differentiating $F(a,x^*(a))=0$ with respect to $a$ gives
\[
x^{*2}+2(ax^*-b)\frac{dx^*}{da}=0,
\]
and hence
\[
\frac{dx^*}{da}=\frac{x^{*2}}{2(b-ax^*)}>0.
\]
Thus equilibrium effort is strictly increasing on the interior of the pure branch. At $a=0$, the same equation reduces to $x^*=(c-\theta)/(2b)$, and the pure root is continuous through this point by Remark \ref{rem:cinfo-zero}. At the upper boundary of the pure branch, $\zeta=0$, the formula in Theorem \ref{thm:cinfo-unrestricted} gives $x^*=\kappa=b/a$, so the increasing pure branch reaches the boundary value continuously.

Now consider the mixed-equilibrium region. By Theorem \ref{thm:cinfo-unrestricted}, every unrestricted mixed equilibrium has
\[
\mathrm{E}[x]=\kappa=\frac{b}{a}.
\]
The condition $\zeta<0$ can occur only when $a>0$, because for $a<0$ one has $\zeta=\kappa^2-(c-\theta)/a>0$. Therefore, throughout the mixed region,
\[
\frac{d\mathrm{E}[x]}{da}=-\frac{b}{a^2}<0.
\]
Expected effort is therefore strictly decreasing once the contest is mixed.

The two regions meet where $\zeta=0$, equivalently
\[
b^2=a(c-\theta).
\]
At this point the pure value $x^*=\kappa$ and the mixed moment $\mathrm{E}[x]=\kappa$ coincide. Since effort increases up to this boundary from the pure side and decreases after it on the mixed side, unrestricted equilibrium effort is maximized at $\zeta=0$.
\end{proof}

\begin{proof}{Detailed proof of Corollary}{\ref{cor:cinfo-max}}
With the prize normalized to one and effort cost linear, rent dissipation is measured by total expected effort. On the pure-equilibrium branch, Theorem \ref{thm:cinfo-unrestricted} gives the symmetric profile $(x^*,x^*)$, so total effort is
\[
R(a)=2x^*(a).
\]
By Theorem \ref{thm:cinfo-effort}, $x^*(a)$ is increasing on the pure branch. Hence $R(a)$ is increasing there as well.

On the mixed-equilibrium branch, Theorem \ref{thm:cinfo-unrestricted} gives
\[
\mathrm{E}[x]=\mathrm{E}[y]=\kappa=\frac{b}{a},
\]
for every unrestricted mixed equilibrium. Therefore total expected effort is
\[
R(a)=\mathrm{E}[x+y]=\frac{2b}{a}.
\]
Since the mixed branch exists only under suppression, $a>0$, and
\[
R'(a)=-\frac{2b}{a^2}<0.
\]
Thus rent dissipation is decreasing throughout the mixed region.

At the boundary $\zeta=0$, equivalently $b^2=a(c-\theta)$, the pure and mixed formulas coincide: Theorem \ref{thm:cinfo-unrestricted} gives $x^*=\kappa=b/a$, so
\[
2x^*=2\kappa=\frac{2b}{a}.
\]
Therefore total unrestricted equilibrium effort rises on the pure side, falls on the mixed side, and is maximized exactly at $\zeta=0$, rather than at the extreme of suppression.
\end{proof}

\begin{proof}{Detailed proof of Proposition}{\ref{prop:cinfo-truncation}}
We first check when the unrestricted pure candidate survives in the literal truncated contest. If $a\neq0$, let $x^*$ be the pure candidate from Theorem \ref{thm:cinfo-unrestricted}; if $a=0$, let $x^*=(c-\theta)/(2b)$ as in Remark \ref{rem:cinfo-zero}. In either case the candidate is symmetric, so
\[
P(x^*,x^*)=\bar P(x^*,x^*)=\frac12,
\]
and its payoff is $\frac12-\theta x^*$.

Fix any deviation $x\geq0$ against $x^*$. Write the corresponding raw, untruncated payoff as
\[
U(x,x^*)=P(x,x^*)-\theta x.
\]
If $0\leq P(x,x^*)\leq1$, then truncation is inactive and the deviator receives exactly $U(x,x^*)$, which is no larger than $U(x^*,x^*)$ because $x^*$ is the unrestricted best response. If $P(x,x^*)\geq1$, then clipping replaces $P(x,x^*)$ by $1\leq P(x,x^*)$, so the truncated payoff is weakly below the unrestricted payoff and again cannot exceed $U(x^*,x^*)$. Finally, if $P(x,x^*)\leq0$, the truncated payoff from the deviation is $-\theta x\leq0$.

Therefore the pure candidate is a Nash equilibrium of the truncated contest whenever
\[
\frac12-\theta x^*\geq0.
\]
If this inequality fails, the candidate payoff is negative, whereas the zero-effort deviation gives
\[
\bar P(0,x^*)\geq0
\]
and has zero effort cost. Hence the candidate cannot be an equilibrium. Since any on-domain pure equilibrium inherited from the cubic benchmark must satisfy the same local first- and second-order conditions as in the unrestricted polynomial game, Theorem \ref{thm:cinfo-unrestricted} and Remark \ref{rem:cinfo-zero} leave no other on-domain pure candidate. This proves the pure statement.

Now suppose $a\neq0$ and $\zeta<0$. In an on-domain mixed equilibrium, truncation is inactive on the equilibrium support, so the same quadratic indifference argument used in Theorem \ref{thm:cinfo-unrestricted} applies to the support payoffs. Hence any such mixed equilibrium must satisfy the unrestricted moment conditions
\[
\mathrm{E}[x]=\mathrm{E}[y]=\kappa,
\qquad
\operatorname{Var}(x)=\operatorname{Var}(y)=-\zeta.
\]
Substituting these moments into the raw expected payoff gives the common equilibrium payoff
\[
\frac12-\theta\kappa.
\]
A player can always deviate to zero effort and obtain
\[
\mathrm{E}[\bar P(0,Y)]\geq0,
\]
because truncated probabilities are nonnegative. Thus participation is necessary:
\[
\frac12-\theta\kappa=\frac12-\frac{\theta b}{a}\geq0.
\]
The proposition does not claim that these mixed-region conditions are sufficient; the sufficient support-specific checks are supplied by Propositions \ref{prop:cinfo-two-point} and \ref{prop:cinfo-mixed-simple}.
\end{proof}

\begin{proof}{Detailed proof of Proposition}{\ref{prop:cinfo-two-point}}
In the mixed-equilibrium region, $a>0$ and
\[
s^2=-\zeta=\frac{c-\theta}{a}-\frac{b^2}{a^2}>0.
\]
The symmetric two-point construction is feasible precisely when its lower support point is nonnegative:
\[
\kappa-s\geq0
\qquad\Longleftrightarrow\qquad
s^2\leq\kappa^2.
\]
Using $\kappa=b/a$, this condition is
\[
\frac{c-\theta}{a}-\frac{b^2}{a^2}\leq\frac{b^2}{a^2}
\qquad\Longleftrightarrow\qquad
a\leq\frac{2b^2}{c-\theta}.
\]
On this branch, the distribution that assigns probability $\nicefrac12$ to each of $\kappa-s$ and $\kappa+s$ has mean
\[
\frac{\kappa-s+\kappa+s}{2}=\kappa
\]
and variance $s^2=-\zeta$. It therefore matches the mixed-equilibrium moments in \eqref{mixed}.

If $a>2b^2/(c-\theta)$, then $\kappa-s<0$, so the symmetric support is infeasible in $\mathbb R_+$. Use instead the endpoint support $\{0,z\}$, assigning probability $1-p$ to $0$ and probability $p$ to $z$. Its first two moments are
\[
E_1=pz,
\qquad
E_2=pz^2.
\]
Matching the required moments $E_1=\kappa=b/a$ and $E_2=(c-\theta)/a$ gives
\[
z=\frac{E_2}{E_1}=\frac{c-\theta}{b},
\qquad
p=\frac{E_1}{z}=\frac{b^2}{a(c-\theta)}.
\]
Because the mixed region implies $a>b^2/(c-\theta)$, one has $p\in(0,1)$, so this is a valid distribution on $[0,\infty)$. Its variance is
\[
E_2-E_1^2=\frac{c-\theta}{a}-\frac{b^2}{a^2}=-\zeta.
\]

Thus both branches generate exactly the unrestricted mixed moment profile. By Theorem \ref{thm:cinfo-unrestricted}, each such distribution is an unrestricted mixed equilibrium. At the switch point $a=2b^2/(c-\theta)$, one has $s=\kappa$, so the symmetric support is $\{0,2\kappa\}$. Also
\[
z=\frac{c-\theta}{b}=2\kappa,
\qquad
p=\frac12,
\]
so the endpoint description and the symmetric description coincide.
\end{proof}

\begin{proof}{Detailed proof of Proposition}{\ref{prop:cinfo-mixed-branches}}
By Proposition \ref{prop:cinfo-truncation}, any on-domain mixed equilibrium must match the unrestricted mixed moments and satisfy $\tfrac{1}{2}-\theta\kappa\geq 0$. It therefore remains to verify that the stated two-point constructions match these moments and admit no profitable deviation into the truncation region.

\smallskip
\noindent\emph{Step 1: Moment matching for the two-point supports.} On the symmetric branch, $\kappa-s\geq 0$ if and only if $s\leq \kappa$, which, using $s^{2}=-\zeta=(c-\theta)/a-b^{2}/a^{2}$, is equivalent to $a\leq 2b^{2}/(c-\theta)$. The distribution on $\{\kappa-s,\kappa+s\}$ with equal mass has mean $\kappa$ and variance $s^{2}=-\zeta$, matching the unrestricted moments.

On the endpoint branch, solving $zp=\kappa$ and $z^{2}p=(c-\theta)/a$ for $(z,p)$ gives
\[
z=\frac{c-\theta}{b},
\qquad
p=\frac{b^{2}}{a(c-\theta)}.
\]
At the switch $a=2b^{2}/(c-\theta)$, one has $\kappa-s=0$, $\kappa+s=z$, and $p=\tfrac{1}{2}$, so the two branches coincide.

\smallskip
\noindent\emph{Step 2: Truncation analysis on the symmetric branch.} Fix $a\in\bigl(b^{2}/(c-\theta),\,2b^{2}/(c-\theta)\bigr]$ and let
\[
\tau(t)=\min\{1,\max\{0,t\}\}.
\]
The deviator's true expected payoff against the symmetric distribution is
\[
\widetilde U_{\mathrm{sym}}(x)=\tfrac{1}{2}\tau\bigl(P(x,\kappa-s)\bigr)+\tfrac{1}{2}\tau\bigl(P(x,\kappa+s)\bigr)-\theta x.
\]
A direct calculation using $\kappa=b/a$ and $as^{2}=(c-\theta)-b^{2}/a$ gives
\[
P(x,\kappa-s)+P(x,\kappa+s)=1+2\theta(x-\kappa).
\]
The untruncated expected payoff is therefore constant at $\tfrac{1}{2}-\theta\kappa$. Any profitable deviation must come entirely from truncation.

\emph{Case $x\geq \kappa$.} Let $P_{1}=P(x,\kappa-s)$ and $P_{2}=P(x,\kappa+s)$. Then $P_{1}+P_{2}=1+2\theta(x-\kappa)\geq 1$. The truncation adjustment
\[
\Delta(x):=\tau(P_{1})+\tau(P_{2})-P_{1}-P_{2}
\]
is non-positive in every case: it is zero if both terms lie in $[0,1]$; it is weakly negative if exactly one term exceeds $1$; it is strictly negative if both exceed $1$; and if one term is negative then the other must exceed $1$, in which case
\[
\Delta(x)=1-P_{1}-P_{2}=-2\theta(x-\kappa)\leq 0.
\]
Hence $\widetilde U_{\mathrm{sym}}(x)\leq \tfrac{1}{2}-\theta\kappa$ for every $x\geq \kappa$.

\emph{Case $x\in[0,\kappa]$.} Now $P_{1}+P_{2}=1+2\theta(x-\kappa)\leq 1$, so it is enough to show that both terms are non-negative throughout $[0,\kappa]$. Along the lower edge $y=\kappa-s$, the quadratic $P(x,\kappa-s)$ is concave in $x$ with vertex at $\kappa+\theta/(2as)>\kappa$. It is therefore increasing on $[0,\kappa]$, and its minimum is $P(0,\kappa-s)$. Along the upper edge $y=\kappa+s$, the quadratic $P(x,\kappa+s)$ is convex in $x$ with vertex at $\kappa-\theta/(2as)$, so its minimum on $[0,\kappa]$ is attained at
\[
\bar x=\max\left\{0,\kappa-\frac{\theta}{2as}\right\}.
\]
Since $P_{1}+P_{2}\leq 1$, non-negativity of both components also implies $P_{1},P_{2}\leq 1$. The conditions $P(0,\kappa-s)\geq 0$ and $P(\bar x,\kappa+s)\geq 0$ therefore ensure that both components stay in $[0,1]$ on $[0,\kappa]$, as required.

\smallskip
\noindent\emph{Step 3: Truncation analysis on the endpoint branch.} Fix $a\geq 2b^{2}/(c-\theta)$ and consider the distribution $\{0,z\}$ with probabilities $(1-p,p)$. The deviator's true payoff is
\[
\widetilde U_{\mathrm{end}}(x)=(1-p)\tau\bigl(P(x,0)\bigr)+p\tau\bigl(P(x,z)\bigr)-\theta x,
\]
and the corresponding untruncated payoff is again the constant $\tfrac{1}{2}-\theta\kappa$.

Write
\[
x_{0}=\frac{c+\sqrt{c^{2}+2b}}{2b}
\]
for the positive root of $P(x,0)=0$, so $P(x,0)\geq 0$ on $[0,x_{0}]$ and $P(x,0)\leq 0$ for $x\geq x_{0}$. We show that the only truncation issue that can arise is a dip of $P(\cdot,z)$ below zero on $[0,x_{0}]$.

\emph{Right tail ($x\geq x_{0}$).} The quadratic $P(x,z)$ is convex in $x$ because $2(az-b)>0$ in the mixed-equilibrium region, with vertex at $x_{v}=(az^{2}-c)/(2(az-b))$. Direct substitution of $z=(c-\theta)/b$ gives
\[
x_{v}-x_{0}=-\frac{a(c-\theta)\bigl(\theta+\sqrt{c^{2}+2b}\bigr)-b^{2}\sqrt{c^{2}+2b}}{2b\bigl(a(c-\theta)-b^{2}\bigr)}.
\]
On the endpoint branch, $a\geq2b^2/(c-\theta)$, so the denominator is positive. The numerator is already positive at $a=2b^2/(c-\theta)$, where it equals $b^2(2\theta+\sqrt{c^2+2b})$, and is increasing in $a$. Hence $x_v<x_0$ throughout the endpoint branch, and $P(x,z)$ is increasing on $[x_0,\infty)$.

It remains to check the value at $x_0$. Substitution gives
\[
P(x_0,z)-1=\frac{a(c-\theta)\bigl(b+c\theta+\theta\sqrt{c^2+2b}\bigr)-2b^2(b+c\theta-\theta^2)}{2b^3}.
\]
Thus $P(x_0,z)\geq1$ whenever
\[
a\geq a_0:=\frac{2b^{2}(b+c\theta-\theta^{2})}{(c-\theta)\bigl(b+c\theta+\theta\sqrt{c^{2}+2b}\bigr)}.
\]
Moreover,
\[
\frac{a_0}{2b^2/(c-\theta)}=\frac{b+c\theta-\theta^2}{b+c\theta+\theta\sqrt{c^2+2b}}<1,
\]
so the inequality holds throughout the endpoint branch. For $x\geq x_{0}$, therefore, $\tau(P(x,0))=0$, $\tau(P(x,z))=1$, and $\widetilde U_{\mathrm{end}}(x)=p-\theta x$, which is decreasing in $x$. At the left endpoint $x=x_{0}$, truncation on $P(\cdot,z)$ weakly lowers the payoff relative to the untruncated benchmark, so the right-tail region yields no profitable deviation.

\emph{Interior ($x\in[0,x_{0}]$).} Here $P(x,0)\geq 0$, so $\tau(P(x,0))=\min\{1,P(x,0)\}$ weakly lowers payoff relative to the untruncated benchmark. Any strict profit must therefore come from a dip of $P(x,z)$ below zero. The sufficient condition
\[
\min_{x\in[0,x_{0}]}P(x,z)\geq 0
\]
rules this out. Since $P(\cdot,z)$ is convex, the minimum on $[0,x_{0}]$ is attained either at $x=0$ or at the vertex $x_{v}$. The explicit threshold $\bar a_{M}$ comes from the interior tangency case $x_{v}\in(0,x_{0})$ with $P(x_{v},z)=0$. Substituting $y=z=(c-\theta)/b$ into the tangency system
\[
\bigl(2ay-2b\bigr)x-ay^{2}+c=0,
\qquad
P(x,y)=0,
\]
and solving for $a$ yields
\[
a\leq \bar a_{M}=\frac{\bigl(b+c^{2}-3c\theta+2\theta^{2}+\sqrt{b(b-2c\theta+2\theta^{2})}\bigr)b^{2}}{(c-\theta)^{3}}.
\]
If instead $x_{v}\leq 0$, the minimum is attained at $x=0$, and then $P(0,z)=\tfrac{1}{2}-\theta(c-\theta)/b\geq 0$ because the reality of the square-root term in $\bar a_{M}$ implies $b\geq 2\theta(c-\theta)$. Hence the endpoint branch is admissible throughout $[2b^{2}/(c-\theta),\,\bar a_{M}]$, completing the proof.
\end{proof}

\begin{proof}{Detailed proof of Proposition}{\ref{prop:cinfo-mixed-simple}}
The result follows from Proposition \ref{prop:cinfo-mixed-branches} once the primitive inequalities are shown to imply the branch conditions, participation, and nonemptiness of the admissible endpoint interval.

First note that the mixed region implies $a>b^2/(c-\theta)$, hence
\[
\kappa=\frac{b}{a}<\frac{c-\theta}{b}.
\]
The primitive assumptions imply
\[
b\geq \frac{c^2}{2}\geq 2\theta(c-\theta),
\]
because $c^2/2-2\theta(c-\theta)=(c-2\theta)^2/2\geq0$. Therefore
\[
\theta\kappa<\frac{\theta(c-\theta)}{b}\leq \frac12,
\]
so the participation restriction from Proposition \ref{prop:cinfo-truncation} is satisfied throughout the mixed region covered here.

\smallskip
\noindent\emph{Symmetric branch.} On $\left(\frac{b^{2}}{c-\theta},\frac{2b^{2}}{c-\theta}\right]$, the variance
\[
V(a)=s^{2}=\frac{c-\theta}{a}-\frac{b^{2}}{a^{2}}
\]
satisfies
\[
\frac{dV}{da}=\frac{2b^{2}-a(c-\theta)}{a^{3}}\geq 0.
\]
Thus $V$ is increasing on the symmetric branch and attains its maximum at the switch point $a=2b^{2}/(c-\theta)$:
\[
V_{\max}=\frac{(c-\theta)^2}{4b^2}.
\]
Hence
\[
s\leq \frac{c-\theta}{2b},
\qquad
2bs\leq c-\theta\leq \theta,
\]
where the last inequality uses $c\leq2\theta$. Since $2as\kappa=2bs\leq\theta$, one has
\[
\kappa-\frac{\theta}{2as}\leq0,
\]
so $\bar x=0$ in Proposition \ref{prop:cinfo-mixed-branches}. The two symmetric-branch checks therefore reduce to
\[
P(0,\kappa-s)\geq0,
\qquad
P(0,\kappa+s)\geq0.
\]
But
\[
P(0,y)=\frac12-cy+by^2
\]
has discriminant $c^2-2b\leq0$ under $b\geq c^2/2$, and is therefore nonnegative for every $y\in\mathbb R$. Thus the symmetric branch satisfies the branch-specific conditions throughout its range.

\smallskip
\noindent\emph{Endpoint branch.} For $a\geq2b^2/(c-\theta)$, Proposition \ref{prop:cinfo-mixed-branches} gives an on-domain mixed equilibrium whenever $a\leq\bar a_M$. Hence it remains only to verify that the interval
\[
\left[\frac{2b^{2}}{c-\theta},\bar a_{M}\right]
\]
is nonempty. Direct computation gives
\[
\bar a_{M}-\frac{2b^{2}}{c-\theta}=\frac{\left(b-c^{2}+c\theta+\sqrt{b \left(b-2 c \theta+2 \theta^{2} \right)}\right)b^{2}}{(c-\theta)^{3}}.
\]
The square-root term is well-defined because
\[
b-2c\theta+2\theta^{2}\geq \frac{c^{2}}{2}-2c\theta+2\theta^{2}=\frac{1}{2}(c-2\theta)^{2}\geq 0.
\]
Moreover,
\[
b-c^{2}+c\theta\geq \frac{c^{2}}{2}-c^{2}+c\theta=c\left(\theta-\frac{c}{2}\right)\geq 0.
\]
Thus $\bar a_M\geq2b^2/(c-\theta)$, so the endpoint interval is nonempty. Combining the symmetric branch and the endpoint branch proves the proposition.
\end{proof}

\begin{proof}{Detailed proof of Proposition}{\ref{prop:cinfo-mixed-failure}}
Let
\[
z=\frac{c-\theta}{b},
\]
and choose a constant $\lambda>b$ large enough that
\[
\lambda z^{2}+\theta z>\frac{1}{2}.
\]
For each suppression level $a$, consider the deviation
\[
x_{a}=\frac{\lambda}{a}.
\]
The choice $\lambda>b$ ensures that $x_{a}>\kappa=b/a$, so the deviator stays to the right of the equilibrium mean, while $\lambda z^{2}+\theta z>\frac{1}{2}$ guarantees that lower truncation eventually becomes profitable against sufficiently large opponent efforts. The deviation is feasible because $x_{a}>0$ for every $a>0$. Define
\[
\phi_{a}(y)=P(x_{a},y).
\]
Direct substitution gives
\[
\phi_{a}(y)=(b-\lambda)y^{2}+\left(\frac{\lambda^{2}}{a}-c\right)y+\left(\frac{1}{2}+\frac{c\lambda}{a}-\frac{b\lambda^{2}}{a^{2}}\right).
\]
Recall that $\tau(t)=\min\{1,\max\{0,t\}\}$.
Because $b-\lambda<0$, the function $\phi_{a}$ is concave in $y$. Moreover,
\[
\phi_{a}'(0)=\frac{\lambda^{2}}{a}-c\longrightarrow -c<0
\qquad\text{as }a\to\infty,
\]
so for all sufficiently large $a$ the derivative is already negative at $0$, and concavity then implies that $\phi_{a}$ is strictly decreasing on $[0,\infty)$. Also,
\[
\phi_{a}(0)=\frac{1}{2}+\frac{c\lambda}{a}-\frac{b\lambda^{2}}{a^{2}}\longrightarrow \frac{1}{2}<1,
\]
while
\[
\phi_{a}(z)=\frac{1}{2}-\theta z-\lambda z^{2}+\frac{b\lambda z+\lambda^{2}z+\lambda\theta}{a}-\frac{b\lambda^{2}}{a^{2}}
\longrightarrow \frac{1}{2}-\theta z-\lambda z^{2}<0
\]
by the choice of $\lambda$. In particular, the deviation $x_a$ eventually pushes the raw winning probability below zero at $y=z$, and therefore throughout the entire tail $y\ge z$. Hence there exists a finite threshold $\bar a_{F}$ such that for every $a\geq \bar a_{F}$,
\[
\phi_{a}(y)\leq \phi_{a}(0)<1\quad\text{for all }y\in[0,z],
\qquad
\phi_{a}(y)<0\quad\text{for all }y\geq z.
\]

Now suppose, toward contradiction, that an on-domain mixed equilibrium exists for some $a\geq \bar a_{F}$, and let $Y$ denote player $Y$'s effort draw. By Theorem \ref{thm:cinfo-unrestricted},
\[
\mathrm{E}[Y]=\kappa=\frac{b}{a}>0,
\qquad
\mathrm{E}[Y^{2}]=\frac{c-\theta}{a}=\kappa z.
\]
If $\Pr(Y\geq z)=0$, then $Y<z$ almost surely. Since $\mathrm{E}[Y]>0$, this implies $Y>0$ with positive probability, hence $Y^{2}<zY$ with positive probability and therefore $\mathrm{E}[Y^{2}]<z\,\mathrm{E}[Y]$, contradicting the displayed identities. Therefore
\[
\Pr(Y\geq z)>0.
\]

For every $y\in[0,z]$, we have $\phi_{a}(y)<1$, so truncation cannot lower the payoff there; for every $y\geq z$, we have $\phi_{a}(y)<0$, so truncation strictly raises it there. Because $\Pr(Y\ge z)>0$, this lower-tail truncation creates a strict gain on a set of positive probability. Thus
\[
\tau\bigl(\phi_{a}(y)\bigr)\geq \phi_{a}(y)\quad\text{for all }y\geq 0,
\]
with strict inequality on a set of positive probability because $\Pr(Y\geq z)>0$. Consequently,
\[
\mathrm{E}\bigl[\tau(P(x_{a},Y))\bigr]>\mathrm{E}[P(x_{a},Y)].
\]
Subtracting the same cost term $\theta x_{a}$ from both sides and using the unrestricted moment identities yields
\[
\mathrm{E}\bigl[\tau(P(x_{a},Y))\bigr]-\theta x_{a}
>
\mathrm{E}[P(x_{a},Y)]-\theta x_{a}
=
\frac{1}{2}-\theta\kappa.
\]
So $x_{a}$ is a profitable deviation in the truncated contest, contradicting equilibrium. Hence no on-domain mixed equilibrium exists for any $a\geq \bar a_{F}$.
\end{proof}

\begin{proof}{Detailed proof of Theorem}{\ref{thm:bn-unrestricted}}
Recall the notation from the main text:
\[
M_1=\mathrm{E}[\theta],
\qquad
M_2=\mathrm{E}[\theta^2],
\qquad
\Delta=c-M_1,
\qquad
\sigma_\theta^2=M_2-M_1^2,
\]
and, for $a\neq 0$,
\[
\kappa=\frac{b}{a},
\qquad
\omega=\frac{\sigma_\theta^2}{a^2},
\qquad
\zeta=\kappa^2-\frac{\Delta}{a}.
\]
Let $x(\theta)$ be a symmetric equilibrium candidate in the affine relaxation and let $Y=x(\theta')$ denote the opponent's effort. Write
\[
E_1:=\mathrm{E}[Y],
\qquad
E_2:=\mathrm{E}[Y^2].
\]
For a type $\theta$ who deviates to $\tilde x\in\mathbb{R}$, a direct expansion of expected utility gives
\[
U(\theta,\tilde x)=K(\theta)+(aE_1-b)\tilde x^2+(c-\theta-aE_2)\tilde x,
\]
where $K(\theta)$ does not depend on $\tilde x$. Hence any equilibrium best response must satisfy
\[
2(aE_1-b)\tilde x+c-\theta-aE_2=0.
\]
If $aE_1>b$, the objective is strictly convex in $\tilde x$ and cannot have a finite optimum for every type. If $aE_1=b$, expected utility is linear in $\tilde x$. Since the deviation space in this affine relaxation is $\mathbb R$, any type with nonzero slope has no finite best response. The slope $c-\theta-aE_2$, viewed as a function of $\theta$, is affine with coefficient $-1$; by nondegenerate type uncertainty it is nonzero on a set of types of positive measure. Hence some type has no finite best response. Therefore necessarily $aE_1<b$, and the unique best response is affine:
\[
x(\theta)=\frac{aE_2+\theta-c}{2(aE_1-b)}=k\theta+d,
\]
with
\[
k=\frac{1}{2(aE_1-b)}<0,
\qquad
d=\frac{aE_2-c}{2(aE_1-b)}.
\]
Thus any symmetric equilibrium must be affine, and its slope is strictly negative.

Now take moments. From $x(\theta)=k\theta+d$ we obtain
\[
E_1=kM_1+d,
\qquad
\mathrm{Var}(x)=k^2\sigma_\theta^2.
\]
Let $V:=\mathrm{Var}(x)$, so that $E_2=E_1^2+V$. Using the formula for $d$ and the identity $E_1=kM_1+d$, we get
\[
2(aE_1-b)E_1=aE_2+M_1-c,
\]
hence
\[
aE_1^2-2bE_1+\Delta-aV=0.
\]
Since
\[
V=k^2\sigma_\theta^2=\frac{\sigma_\theta^2}{4(aE_1-b)^2},
\]
it is convenient to set $z:=E_1-\kappa$. Then $aE_1-b=az$, and the displayed identity becomes
\[
z^4-\zeta z^2-\frac{\omega}{4}=0.
\]
So $z^2$ solves the quadratic equation
\[
q^2-\zeta q-\frac{\omega}{4}=0,
\]
and therefore
\[
(E_1-\kappa)^2=\frac{\zeta\pm\sqrt{\zeta^2+\omega}}{2}.
\]
Because $\omega>0$, one has $\sqrt{\zeta^2+\omega}>|\zeta|$, so the minus branch is negative and cannot be the square of a real number. Hence every real candidate must satisfy
\[
(E_1-\kappa)^2=\frac{\zeta+\sqrt{\zeta^2+\omega}}{2}.
\]
The second-order condition is $aE_1<b=a\kappa$, i.e.
\[
a(E_1-\kappa)<0.
\]
Thus the sign of $E_1-\kappa$ must be opposite to the sign of $a$, and the only admissible root is
\[
E_1=\kappa-\frac{\mathrm{sgn}(a)}{\sqrt{2}}\sqrt{\zeta+\sqrt{\zeta^2+\omega}}.
\]
This is the expected-effort formula in the statement.

Using $a(E_1-\kappa)=aE_1-b$, we next obtain
\[
k=\frac{1}{2(aE_1-b)}=-\frac{1}{\sqrt{2}|a|\sqrt{\zeta+\sqrt{\zeta^2+\omega}}},
\qquad
d=E_1-kM_1.
\]
Finally,
\[
V=k^2\sigma_\theta^2=\frac{\omega}{2\left(\zeta+\sqrt{\zeta^2+\omega}\right)}
=\frac{\sqrt{\zeta^2+\omega}-\zeta}{2},
\]
which is exactly the variance formula in the theorem.

Conversely, define $x(\theta)=k\theta+d$ from the formulas above. Since $k<0$, the strategy is strictly decreasing. Moreover,
\[
aE_1-b=a(E_1-\kappa)=-\frac{|a|}{\sqrt{2}}\sqrt{\zeta+\sqrt{\zeta^2+\omega}}<0,
\]
so $U(\theta,\tilde x)$ is strictly concave in $\tilde x$. Its first-order condition is solved uniquely by $\tilde x=x(\theta)$, so every type best responds to the candidate strategy. Hence the candidate is a symmetric strictly decreasing Bayesian Nash equilibrium of the affine relaxation.

Uniqueness is now immediate: any symmetric equilibrium must satisfy the quadratic equation for $(E_1-\kappa)^2$ above, and the second-order condition selects exactly one sign. That uniquely determines $E_1$, and then the formulas for $k$, $d$, and $\mathrm{Var}(x)$ follow.
\end{proof}

\begin{proof}{Detailed proof of Remark}{\ref{rem:bn-zero}}
When $a=0$, the raw polynomial expression becomes
\[
P(\tilde x,y)=\frac{1}{2}+c(\tilde x-y)-b(\tilde x^2-y^2).
\]
Hence expected utility of type $\theta$ from choosing $\tilde x$ against any opponent strategy is
\[
U(\theta,\tilde x)=K(\theta)+(c-\theta)\tilde x-b\tilde x^2,
\]
where $K(\theta)$ does not depend on $\tilde x$. Since $b>0$, this objective is strictly concave, and since $\theta<c$, its unique maximizer is
\[
\tilde x=\frac{c-\theta}{2b}>0.
\]
So each type has a dominant strategy, and the affine relaxation and unrestricted polynomial game coincide. Integrating over types gives
\[
\mathrm{E}[x]=\frac{c-M_1}{2b},
\qquad
\mathrm{Var}(x)=\frac{M_2-M_1^2}{4b^2}.
\]
The literal truncated contest is not automatic. On the affine support,
\[
P\bigl(x(\theta),x(\theta')\bigr)=\frac12+\frac{(\theta')^2-\theta^2}{4b},
\]
so all on-support probabilities lie strictly between $0$ and $1$ if and only if $\beta^2-\alpha^2<2b$. If this support check fails, clipping is active even locally; if it passes, global restricted-game survival still requires ruling out deviations created by clipping.

In the fixed-point notation from the proof of Proposition \ref{prop:bn-nonnegative}, one has $t=c-aE_2=c$ and
\[
2aA(t)\lambda^2-2b\lambda+1=0,
\]
which reduces to $\lambda=1/(2b)$, yielding exactly the same strategy.

For continuity, the affine-benchmark formulas from Theorem \ref{thm:bn-unrestricted} can be rewritten as
\[
k=-\frac{1}{\sqrt{2\left(b^2-a\Delta+\sqrt{(b^2-a\Delta)^2+\sigma_\theta^2 a^2}\right)}},
\]
and
\[
\mathrm{Var}(x)=\frac{\sigma_\theta^2}{2\left(\sqrt{(b^2-a\Delta)^2+\sigma_\theta^2 a^2}+b^2-a\Delta\right)}.
\]
Hence
\[
k\to -\frac{1}{2b},
\qquad
\mathrm{Var}(x)\to \frac{\sigma_\theta^2}{4b^2}
\qquad\text{as }a\to 0.
\]
Also, using $E_2=\mathrm{E}[x]^2+\mathrm{Var}(x)$, the affine-benchmark moment identities imply
\[
a\,\mathrm{E}[x]^2-2b\,\mathrm{E}[x]+\Delta-a\,\mathrm{Var}(x)=0.
\]
Taking the relevant root and rationalizing gives
\[
\mathrm{E}[x]=\frac{\Delta-a\,\mathrm{Var}(x)}{b+\sqrt{b^2-a\Delta+a^2\mathrm{Var}(x)}},
\]
so $\mathrm{E}[x]\to \Delta/(2b)=(c-M_1)/(2b)$. Therefore
\[
d=\mathrm{E}[x]-kM_1\to \frac{c}{2b}.
\]
Finally, since the limiting strategy satisfies $x_0(\beta)=(c-\beta)/(2b)>0$, continuity of $k$ and $d$ implies $k\beta+d>0$ for all sufficiently small $|a|$. Thus the nonnegative-effort equilibrium coincides with the affine benchmark in a neighborhood of $a=0$.
\end{proof}

\begin{proof}{Detailed proof of Remark}{\ref{rem:bn-vanishing}}
Fix $a\neq0$, $b$, $c$, and $M_1$, and consider a fully active sequence with $\sigma_\theta^2\downarrow0$. Then $\Delta$, $\kappa$, and $\zeta$ remain fixed, while
\[
\omega=\frac{\sigma_\theta^2}{a^2}\downarrow0.
\]
Theorem \ref{thm:bn-unrestricted} and the displayed moment formulas following it give
\[
\mathrm{E}[x]=\kappa-\frac{\mathrm{sgn}(a)}{\sqrt{2}}\sqrt{\zeta+\sqrt{\zeta^2+\omega}},
\qquad
\mathrm{Var}(x)=\frac{\sqrt{\zeta^2+\omega}-\zeta}{2}.
\]
If $\zeta>0$, then $\sqrt{\zeta^2+\omega}\to\zeta$, so
\[
\mathrm{E}[x]\to\kappa-\mathrm{sgn}(a)\sqrt{\zeta}=x^*,
\qquad
\mathrm{Var}(x)\to0.
\]
At the boundary $\zeta=0$, the same formulas give $\mathrm{E}[x]\to\kappa=x^*$ and $\mathrm{Var}(x)\to0$. These are exactly the complete-information pure-equilibrium moments for the limiting type $M_1$.

If $\zeta<0$, then $\sqrt{\zeta^2+\omega}\to-\zeta$, so
\[
\mathrm{E}[x]\to\kappa,
\qquad
\mathrm{Var}(x)\to-\zeta.
\]
These are exactly the complete-information mixed-equilibrium moments from Theorem \ref{thm:cinfo-unrestricted}, with the realized complete-information cost equal to $M_1$.

The distributional interpretation follows from affinity. Writing
\[
x(\theta)=\mathrm{E}[x]+k(\theta-M_1),
\qquad
k=-\frac{1}{\sqrt{2}|a|\sqrt{\zeta+\sqrt{\zeta^2+\omega}}},
\]
we have, in the mixed region $\zeta<0$,
\[
k\sigma_\theta\to-\sqrt{-\zeta}.
\]
Thus any fixed shrinking family of priors whose standardized costs $(\theta-M_1)/\sigma_\theta$ have a weak limit selects the corresponding affine push-forward effort distribution. Without fixing such a family, the universal selection is the moment profile $\bigl(\kappa,-\zeta\bigr)$.
\end{proof}

\begin{proof}{Detailed proof of Theorem}{\ref{thm:uncertainty}}

Within the fully active region, Theorem \ref{thm:bn-unrestricted} applies. Thus for $a\neq 0$,
\[
E_1=\kappa-\frac{\mathrm{sgn}(a)}{\sqrt{2}} \sqrt{ \zeta + \sqrt{\zeta^2+ \omega}},
\qquad
\omega=\frac{\sigma_\theta^2}{a^2}.
\]
Holding $a$ fixed, both $\kappa$ and $\zeta$ are constant, while $\omega$ is strictly increasing in $\sigma_\theta^2$. Differentiating with respect to $\omega$ gives
\[
\frac{\partial E_1}{\partial \omega}
=
-\frac{\mathrm{sgn}(a)}{4\sqrt{2}\,\sqrt{\zeta^2+\omega}\,\sqrt{\zeta+\sqrt{\zeta^2+\omega}}}.
\]
The denominator is strictly positive, so
\[
\frac{\partial E_1}{\partial \omega}<0 \quad\text{if } a>0,
\qquad
\frac{\partial E_1}{\partial \omega}>0 \quad\text{if } a<0.
\]
Since $\omega=\sigma_\theta^2/a^2$, the same sign pattern holds for the comparative static with respect to the variance of types. Finally, if $a=0$, Remark \ref{rem:bn-zero} gives
\[
E_1=\frac{c-M_1}{2b},
\]
which depends only on the mean type and is therefore independent of $\sigma_\theta^2$.
\end{proof}

\begin{proof}{Detailed proof of Theorem}{\ref{thm:bn-disclosure}}
Fix any $\pi\in\mathcal S$. Because the signal rule is symmetric and types are IID, the posterior-mean costs induced by $\pi$ are IID with common distribution $\hat \theta_{\pi}$. Conditional on signal realization $s$, replacing the unknown cost by its posterior mean $\hat \theta_{\pi}(s)$ leaves expected payoffs unchanged, so the induced affine relaxation is exactly the one studied above with type distribution given by $\hat \theta_{\pi}$. By the law of total variance,
\[
0\leq \mathrm{Var}(\hat \theta_{\pi})=\mathrm{Var}(\mathrm{E}[\theta\mid s_{\pi}])\leq \mathrm{Var}(\theta)=\sigma_\theta^2.
\]
The lower bound is attained under no disclosure, when $\hat \theta_{\pi}=M_1$ almost surely. The upper bound is attained under full disclosure, when $\hat \theta_{\pi}=\theta$ almost surely.

If $a\neq 0$ and $\mathrm{Var}(\hat \theta_{\pi})>0$, Theorem \ref{thm:bn-unrestricted} gives ex ante expected effort
\[
E_1=\kappa-\frac{\mathrm{sgn}(a)}{\sqrt{2}}\sqrt{\zeta+\sqrt{\zeta^2+\omega}},
\qquad
\omega=\frac{\mathrm{Var}(\hat \theta_{\pi})}{a^2},
\]
with $\kappa=b/a$ and $\zeta=\kappa^2-(c-\mathrm{E}[\theta])/a$, both independent of $\pi$. Differentiating with respect to $\omega$ gives
\[
\frac{\partial E_1}{\partial \omega}=-\frac{\mathrm{sgn}(a)}{4\sqrt{2}\,\sqrt{\zeta^2+\omega}\,\sqrt{\zeta+\sqrt{\zeta^2+\omega}}}.
\]
Hence expected effort is decreasing in $\mathrm{Var}(\hat \theta_{\pi})$ when $a>0$ and increasing in $\mathrm{Var}(\hat \theta_{\pi})$ when $a<0$. The same monotonicity extends to the no-disclosure endpoint by continuity as $\mathrm{Var}(\hat \theta_{\pi})\downarrow0$. Therefore no disclosure is optimal under suppression and full disclosure is optimal under empowerment. If $a=0$, Remark \ref{rem:bn-zero} gives $E_1=(c-\mathrm{E}[\theta])/(2b)$, which is independent of $\pi$, so the designer is indifferent over $\mathcal S$.

Finally, suppose signal rule $\pi'$ is a Blackwell garbling of $\pi$. Then
\[
\mathrm{E}[\theta\mid s_{\pi'}]
=\mathrm{E}\bigl[\mathrm{E}[\theta\mid s_{\pi}]\mid s_{\pi'}\bigr],
\]
so the law of total variance gives
\[
\mathrm{Var}(\mathrm{E}[\theta\mid s_{\pi'}])\leq \mathrm{Var}(\mathrm{E}[\theta\mid s_{\pi}]).
\]
Combining this variance ordering with the monotonicity just established gives the Blackwell-order statement.
\end{proof}

\begin{proof}{Detailed proof of Theorem}{\ref{thm:bn-a-peak}}
We normalize the problem, invert the map $\xi\mapsto y(\xi)$ into its two monotone branches, determine the sign of $de/dy$ on each branch, and then translate back to $de/d\xi$ using the orientation of the parametrization.

Introduce normalized variables
\[
\rho=\frac{\sigma_\theta^2}{\Delta^2},
\qquad
\xi=\frac{\Delta a}{b^2},
\qquad
E_1(a)=\frac{\Delta}{b}e(\xi).
\]
Then \eqref{eq:bn-e1-a} becomes
\[
e(\xi)=\frac{1-y(\xi)}{\xi},
\qquad
y(\xi)=\sqrt{\frac{1-\xi+\sqrt{(1-\xi)^2+\rho \xi^2}}{2}},
\]
with the continuous extension $e(0)=\frac{1}{2}$. By Remark \ref{rem:bn-positive}, the relevant region is
\[
\xi<\frac{4}{\rho}.
\]

Now define
\[
u(y)=\sqrt{y^2+\rho(y^2-1)},
\qquad
y_0=\sqrt{\frac{\rho}{1+\rho}}.
\]
Eliminating the outer radical in the definition of $y(\xi)$ shows that, for $y\ge y_0$, the inversion produces two branches,
\[
\xi=\xi_-(y):=\frac{2y(y-u(y))}{\rho}
\qquad\text{or}\qquad
\xi=\xi_+(y):=\frac{2y(y+u(y))}{\rho}.
\]
Moreover,
\[
u(1)=1,
\qquad
u(y_0)=0,
\]
so
\[
\xi_-(1)=0,
\qquad
\xi_-(y_0)=\xi_+(y_0)=\frac{2}{1+\rho},
\qquad
\xi_+(1)=\frac{4}{\rho}.
\]
Differentiating gives
\[
\xi_-'(y)=-\frac{2\bigl((u(y)-y)^2+\rho y^2\bigr)}{\rho u(y)}<0,
\]
and
\[
\xi_+'(y)=\frac{2\bigl((u(y)+y)^2+\rho y^2\bigr)}{\rho u(y)}>0.
\]
Thus the minus branch parameterizes $(-\infty,2/(1+\rho)]$, while the plus branch parameterizes $[2/(1+\rho),4/\rho)$.

We begin with the minus branch. On this branch,
\[
e(\xi_-(y))=\frac{1-y}{\xi_-(y)}=\frac{\rho(1-y)}{2y(y-u(y))}.
\]
Differentiating with respect to $y$, and using
\[
u'(y)=\frac{(1+\rho)y}{u(y)}
\qquad\text{and}\qquad
u(y)^2-y^2=\rho(y^2-1),
\]
gives
\[
\mathrm{sgn}\!\left(\frac{d}{dy}e(\xi_-(y))\right)
=
-\mathrm{sgn}\Bigl((y-1)\bigl[y+u(y)(y^2-y-1)\bigr]\Bigr).
\]
Since $\xi_-'(y)<0$ and
\[
\frac{d}{dy}e(\xi_-(y))=e'(\xi)\xi_-'(y),
\]
the sign flips once more, so
\[
\mathrm{sgn}\,e'(\xi)
=
\mathrm{sgn}\Bigl((y-1)\bigl[y+u(y)(y^2-y-1)\bigr]\Bigr).
\]
If $y\ge \varphi:=\frac{1+\sqrt5}{2}$, then $y^2-y-1\ge 0$, so the second factor is strictly positive and therefore $e'(\xi)>0$.

If $0<y<\varphi$, then $1+y-y^2>0$ and
\[
y+u(y)(y^2-y-1)=(1+y-y^2)\left(\frac{y}{1+y-y^2}-u(y)\right),
\]
so the sign of the bracket equals the sign of the corresponding difference of squares,
\[
\frac{y^2}{(1+y-y^2)^2}-u(y)^2=(1-y^2)\left(\rho-\frac{y^3(2-y)}{(y^2-y-1)^2}\right).
\]
Since $(y-1)(1-y^2)<0$ for $0<y<\varphi$, $y\neq 1$, it follows that
\[
\mathrm{sgn}\,e'(\xi)
=
\mathrm{sgn}\left(\frac{y^3(2-y)}{(y^2-y-1)^2}-\rho\right).
\]
Define
\[
g(y)=\frac{y^3(2-y)}{(y^2-y-1)^2}.
\]
Then
\[
g'(y)=\frac{2y^2(y-3)}{(y^2-y-1)^3}>0
\qquad\text{for }0<y<\varphi,
\]
so $g$ is strictly increasing on $(0,\varphi)$. Moreover,
\[
g(1)=1,
\qquad
\lim_{y\downarrow 0}g(y)=0,
\qquad
\lim_{y\uparrow \varphi}g(y)=\infty.
\]
Therefore there exists a unique $y^\dagger\in(0,\varphi)$ such that
\[
g(y^\dagger)=\rho.
\]
To verify that this critical point lies on the admissible minus branch, note that a direct computation gives
\[
g(y_0)-\rho=-\frac{\rho(1+\rho)^3}{\bigl(\rho^{3/2}+\sqrt{\rho}+\sqrt{1+\rho}\bigr)^2}<0.
\]
Since $g$ is strictly increasing, it follows that $y^\dagger>y_0$. Hence $\xi_-(y^\dagger)$ is well defined, and on the minus branch one has $e'(\xi)>0$ for $y>y^\dagger$ and $e'(\xi)<0$ for $y<y^\dagger$. Since $\xi_-$ is strictly decreasing, this yields a unique critical point
\[
\xi^\dagger=\xi_-(y^\dagger)
\qquad\text{and hence}\qquad
a^\dagger=\frac{b^2}{\Delta}\xi^\dagger.
\]
Moreover, because $y^\dagger>y_0$ and $\xi_-$ is strictly decreasing,
\[
\xi^\dagger<\xi_-(y_0)=\frac{2}{1+\rho},
\]
which is equivalent to the bound $a^\dagger<2b^2\Delta/(\Delta^2+\sigma_\theta^2)$ stated in the theorem.

Now turn to the plus branch. There,
\[
e(\xi_+(y))=\frac{1-y}{\xi_+(y)}
\qquad\text{for }y\in[y_0,1).
\]
Differentiating with respect to $y$ gives
\[
\frac{d}{dy}e(\xi_+(y))
=
-\frac{\xi_+(y)+(1-y)\xi_+'(y)}{\xi_+(y)^2}<0.
\]
Because $\xi_+'(y)>0$, it follows that $e'(\xi)<0$ throughout the plus branch. Hence the critical point found on the minus branch is the unique maximizer on the whole region $\xi<4/\rho$.

Finally, on the minus branch the sign of $\xi_-(y)$ is the sign of $y-u(y)$. Since
\[
u(y)^2-y^2=\rho(y^2-1),
\]
the sign of $y-u(y)$ is the sign of $1-y$. Therefore $\xi^\dagger>0$ if and only if $y^\dagger<1$, $\xi^\dagger=0$ if and only if $y^\dagger=1$, and $\xi^\dagger<0$ if and only if $y^\dagger>1$. Because $g$ is strictly increasing and $g(1)=1$, this is equivalent to
\[
\rho<1,
\qquad
\rho=1,
\qquad
\rho>1,
\]
respectively. Recalling that $\rho=\sigma_\theta^2/\Delta^2$ and $\xi=\Delta a/b^2$, we obtain the sign characterization for $a^\dagger$.
\end{proof}

\begin{proof}{Detailed proof of Remark}{\ref{rem:bn-positive}}
Write
\[
R(a):=\sqrt{\frac{b^2-\Delta a+\sqrt{(b^2-\Delta a)^2+\sigma_\theta^2 a^2}}{2}},
\]
so that, for $a\neq0$,
\[
E_1(a)=\frac{b-R(a)}{a}.
\]
Because $\Delta=c-M_1>0$ and $\sigma_\theta^2>0$, the threshold $\bar a=4b^2\Delta/\sigma_\theta^2$ is positive.

If $a<0$, then $b^2-\Delta a>b^2$ and
\[
\sqrt{(b^2-\Delta a)^2+\sigma_\theta^2 a^2}>b^2-\Delta a.
\]
Therefore
\[
R(a)>\sqrt{b^2-\Delta a}>b.
\]
Thus $b-R(a)<0$ and $a<0$, so $E_1(a)>0$ throughout the empowering region. At $a=0$, the continuous extension in \eqref{eq:bn-e1-a} gives $E_1(0)=\Delta/(2b)>0$.

Now suppose $a>0$. Since the denominator is positive,
\[
E_1(a)>0
\iff
R(a)<b.
\]
Squaring the nonnegative sides gives
\[
R(a)<b
\iff
b^2-\Delta a+\sqrt{(b^2-\Delta a)^2+\sigma_\theta^2 a^2}<2b^2,
\]
or equivalently
\[
\sqrt{(b^2-\Delta a)^2+\sigma_\theta^2 a^2}<b^2+\Delta a.
\]
The right-hand side is strictly positive, so squaring again is legitimate. This yields
\[
(b^2-\Delta a)^2+\sigma_\theta^2 a^2<(b^2+\Delta a)^2,
\]
which simplifies to
\[
\sigma_\theta^2 a^2<4b^2\Delta a.
\]
Because $a>0$, this is equivalent to
\[
a<\frac{4b^2\Delta}{\sigma_\theta^2}=\bar a.
\]
The same chain with equality gives $E_1(\bar a)=0$, and reversing the final inequality gives $E_1(a)<0$ for $a>\bar a$.

This proves that $\bar a$ is only the positivity boundary for the affine average effort. It does not ensure that the affine schedule is nonnegative type by type. Since the affine rule is strictly decreasing in type, full activity on $[\alpha,\beta]$ requires the high-cost endpoint to remain nonnegative. Proposition \ref{prop:bn-dropout} identifies the resulting prior-dependent threshold $a_D(F)$ under suppression; beyond it, the nonnegative-effort equilibrium is cutoff-affine rather than fully affine. The literal truncated contest imposes the further state-by-state support checks treated after Proposition \ref{prop:bn-nonnegative}.
\end{proof}

\begin{proof}{Detailed proof of Theorem}{\ref{thm:bn-peak-variance}}
Write
\[
\rho=\frac{\sigma_\theta^2}{\Delta^2},
\qquad
A(\rho)=\frac{\Delta a^\dagger}{b^2},
\]
so that
\[
a^\dagger=\frac{b^2}{\Delta}A(\rho).
\]
From the proof of Theorem \ref{thm:bn-a-peak}, the critical point is characterized by a unique $y^\dagger\in(0,\varphi)$ such that
\[
g(y^\dagger)=\rho,
\qquad
g(y)=\frac{y^3(2-y)}{(y^2-y-1)^2},
\]
and
\[
A(\rho)=\frac{\Delta a^\dagger}{b^2}=\xi^\dagger=\xi_-(y^\dagger)=\frac{2y^\dagger(y^\dagger-u(y^\dagger))}{\rho}.
\]
Using $\rho=g(y^\dagger)$ and the definition of $u$, we obtain
\[
u(y^\dagger)^2=(y^\dagger)^2+\rho\bigl((y^\dagger)^2-1\bigr)=\frac{(y^\dagger)^2}{\bigl(1+y^\dagger-(y^\dagger)^2\bigr)^2}.
\]
Because $y^\dagger<\varphi$, the denominator is positive, so
\[
u(y^\dagger)=\frac{y^\dagger}{1+y^\dagger-(y^\dagger)^2}.
\]
Substituting into the formula for $A(\rho)=\xi^\dagger$ yields
\[
A(\rho)=\frac{2(1-y^\dagger)\bigl(1+y^\dagger-(y^\dagger)^2\bigr)}{2-y^\dagger}.
\]
Thus $A$ depends only on $\rho$.

Because $g'(y)>0$ on $(0,\varphi)$, the sign of $A'(\rho)$ is the sign of $dA/dy$ evaluated at $y=y^\dagger$. A direct calculation gives
\[
\frac{dA}{dy}=-\frac{2q(y)}{(2-y)^2},
\qquad
q(y):=2y^3-8y^2+8y-1.
\]
Also,
\[
q(0)=-1,
\]
and directly solving $q(y)=0$ gives three real roots; the two in $(0,\varphi)$ are
\[
y_1\approx 0.1453623203,
\qquad
y_2\approx 1.4030317168.
\]
Define
\[
\rho_1:=g(y_1),
\qquad
\rho_2:=g(y_2).
\]
Numerically,
\[
\rho_1\approx 0.00450716,
\qquad
\rho_2\approx 8.73185992.
\]
Because $g$ is strictly increasing and $g(1)=1$, one has $0<\rho_1<1<\rho_2$. Since $q(0)<0$ and the roots are distinct, the sign pattern for $q$ yields
\[
\frac{dA}{dy}>0 \text{ on } (0,y_1),
\qquad
\frac{dA}{dy}<0 \text{ on } (y_1,y_2),
\qquad
\frac{dA}{dy}>0 \text{ on } (y_2,\varphi).
\]
Mapping back through $g(y^\dagger)=\rho$ yields
\[
A'(\rho)>0 \text{ on } (0,\rho_1),
\qquad
A'(\rho)<0 \text{ on } (\rho_1,\rho_2),
\qquad
A'(\rho)>0 \text{ on } (\rho_2,\infty).
\]
Now set
\[
\sigma_1^2:=\rho_1\Delta^2,
\qquad
\sigma_2^2:=\rho_2\Delta^2.
\]
Since $b$ and $\Delta$ are fixed,
\[
\frac{d a^\dagger}{d\sigma_\theta^2}=\frac{b^2}{\Delta^3}A'(\rho),
\]
so $a^\dagger$ has exactly the same monotonicity pattern in $\sigma_\theta^2$.

The sign characterization follows from Theorem \ref{thm:bn-a-peak}: $a^\dagger>0$ if and only if $\sigma_\theta^2<\Delta^2$, $a^\dagger=0$ if and only if $\sigma_\theta^2=\Delta^2$, and $a^\dagger<0$ if and only if $\sigma_\theta^2>\Delta^2$. Finally,
\[
\lim_{y\downarrow 0}g(y)=0,
\qquad
\lim_{y\uparrow\varphi}g(y)=\infty,
\]
so $y^\dagger\to 0$ as $\sigma_\theta^2\downarrow 0$ and $y^\dagger\to\varphi$ as $\sigma_\theta^2\to\infty$. The explicit formula for $A(\rho)$ then gives
\[
\lim_{\sigma_\theta^2\downarrow 0}a^\dagger=\frac{b^2}{\Delta},
\qquad
\lim_{\sigma_\theta^2\to\infty}a^\dagger=0^-.
\]
This proves the claims.
\end{proof}

\begin{proof}{Detailed proof of Proposition}{\ref{prop:bn-nonnegative}}
The proof has three parts. First, any equilibrium with strict concavity in own effort must be cutoff-affine. Second, the only non-concave boundary case is the atom-at-$\alpha$ equilibrium stated in the proposition. Third, outside that boundary case, a cutoff-affine fixed point exists.

Let $X$ denote the opponent's equilibrium effort draw and write
\[
E_1=\mathrm{E}[X],
\qquad
E_2=\mathrm{E}[X^2].
\]
For a type $\theta$ who deviates to some $\tilde x\geq 0$, untruncated interim utility is a quadratic polynomial in $\tilde x$:
\[
U(\theta,\tilde x)=K(\theta)+(aE_1-b)\tilde x^2+(c-\theta-aE_2)\tilde x,
\]
where $K(\theta)$ does not depend on $\tilde x$. Because $\beta<c$, the all-zero profile cannot be an equilibrium: the derivative at zero would equal $c-\theta>0$ for every type. Hence positive effort occurs with positive probability, so $E_1>0$.

If $aE_1>b$, then $U(\theta,\tilde x)\to\infty$ as $\tilde x\to\infty$, which is impossible in equilibrium. Thus
\[
aE_1\leq b.
\]

We first analyze the strict case $aE_1<b$. Then expected utility is strictly concave in own effort, so each type has the unique best response
\[
x(\theta)=\max\left\{0,\frac{c-\theta-aE_2}{2(b-aE_1)}\right\}.
\]
Define
\[
\lambda:=\frac{1}{2(b-aE_1)}>0,
\qquad
t:=c-aE_2.
\]
Then every symmetric equilibrium with $aE_1<b$ must have the cutoff-affine form
\[
x(\theta)=\lambda (t-\theta)_+.
\]

Now set
\[
A(t)=\mathrm{E}[(t-\theta)_+],
\qquad
B(t)=\mathrm{E}[(t-\theta)_+^2].
\]
The cutoff-affine strategy implies
\[
E_1=\lambda A(t),
\qquad
E_2=\lambda^2 B(t).
\]
Substituting these identities into the definitions of $\lambda$ and $t$ yields
\[
2aA(t)\lambda^2-2b\lambda+1=0,
\qquad
t=c-aB(t)\lambda^2.
\]
This proves the cutoff-affine characterization and the fixed-point system in the strict-concavity case.

Conversely, suppose $\lambda>0$, $(\lambda,t)$ satisfies these two equations, and $a\lambda A(t)<b$, and define
\[
x(\theta)=\lambda (t-\theta)_+.
\]
Then
\[
E_1=\lambda A(t),
\qquad
E_2=\lambda^2 B(t),
\qquad
\lambda=\frac{1}{2(b-aE_1)},
\qquad
t=c-aE_2.
\]
For every active type $\theta<t$, we have
\[
2(aE_1-b)x(\theta)+c-\theta-aE_2
=2(aE_1-b)\lambda(t-\theta)+t-\theta
=0,
\]
where the last equality uses $2(b-aE_1)\lambda=1$. So active types satisfy the first-order condition. For every inactive type $\theta\geq t$, the derivative at zero is
\[
c-\theta-aE_2=t-\theta\leq 0,
\]
so the boundary action $x=0$ is optimal. Since $aE_1<b$, expected utility is strictly concave in own effort, hence these pointwise optimality conditions are sufficient. Therefore the strategy is a symmetric Bayesian Nash equilibrium.

We next analyze the boundary case $aE_1=b$. Since $E_1>0$, this case can occur only under suppression, $a>0$. Expected utility is linear in own effort:
\[
U(\theta,\tilde x)=K(\theta)+(c-\theta-aE_2)\tilde x.
\]
For equilibrium, the slope must be weakly negative for every type, otherwise some type would prefer arbitrarily large effort. Since $E_1>0$, some type must nevertheless be willing to put positive probability on positive effort, so the slope must be zero for at least one type. Because the slope is strictly decreasing in $\theta$, this can happen only at the lowest-cost type, and necessarily
\[
c-aE_2=\alpha.
\]
Thus every type $\theta>\alpha$ strictly prefers $0$, while type $\alpha$ is indifferent over all efforts. So the only non-cutoff possibility is a boundary equilibrium supported entirely on the atom at $\alpha$.

Let $m:=F(\{\alpha\})$. In such a boundary equilibrium, if $Z$ denotes the effort randomized by type $\alpha$, then necessarily
\[
m\,\mathrm{E}[Z]=\frac{b}{a},
\qquad
m\,\mathrm{E}[Z^2]=\frac{c-\alpha}{a}.
\]
By Jensen's inequality,
\[
\mathrm{E}[Z^2]\geq \mathrm{E}[Z]^2,
\]
so a necessary condition is
\[
am(c-\alpha)\geq b^2.
\]
Conversely, if this inequality holds, define
\[
x_H:=\frac{c-\alpha}{b},
\qquad
p:=\frac{b^2}{am(c-\alpha)}\leq 1.
\]
Let every type $\theta>\alpha$ choose $0$, and let type $\alpha$ mix between $0$ and $x_H$ with probabilities $1-p$ and $p$. Then
\[
E_1=mpx_H=\frac{b}{a},
\qquad
E_2=mpx_H^2=\frac{c-\alpha}{a}.
\]
Hence type $\alpha$ is indifferent over all efforts, while every type $\theta>\alpha$ has slope
\[
c-\theta-aE_2=\alpha-\theta<0,
\]
so $0$ is uniquely optimal. This yields the symmetric Bayesian Nash equilibrium in the exceptional boundary case stated in the proposition.

At the knife-edge $am(c-\alpha)=b^2$, this construction has $p=1$, so type $\alpha$ plays $x_H=(c-\alpha)/b$ purely. Thus the equality case is already absorbed by the boundary equilibrium and need not be analyzed through the cutoff-affine existence argument.

It remains to prove existence in the cutoff-affine class whenever the boundary case is unavailable. The strict-concavity fixed-point system is the object to solve. For $a=0$, existence is given by Remark \ref{rem:bn-zero}. If $a<0$, Theorem \ref{thm:bn-unrestricted} gives a unique affine equilibrium of the affine relaxation, and at that profile
\[
x(\theta)=\frac{c-\theta-aE_2}{2(b-aE_1)}.
\]
Because $a<0$, $E_2\geq 0$, and $\theta\leq \beta<c$, the numerator is strictly positive for every type, while the denominator is also strictly positive. Hence $x(\theta)>0$ on the whole support, so the affine-relaxation equilibrium is already feasible in the game with $x\geq 0$.

Now suppose $a>0$ and
\[
am(c-\alpha)<b^2,
\]
so the boundary construction above is unavailable. Define
\[
H(t):=2(c-t)A(t)+B(t)-2b\sqrt{\frac{(c-t)B(t)}{a}},
\qquad t\in[\alpha,c].
\]
Since the type support is $[\alpha,\beta]\subset(0,c)$, the maps $A$ and $B$ are continuous on $[\alpha,c]$, with $B(t)>0$ for every $t>\alpha$, so $H$ is continuous on $[\alpha,c]$.

If $m=0$, then by Cauchy--Schwarz,
\[
A(t)^2\leq B(t)\Pr(\theta<t),
\]
and because $F$ has no atom at $\alpha$, we have $\Pr(\theta<t)\to 0$ as $t\downarrow \alpha$. Hence
\[
\frac{A(t)}{\sqrt{B(t)}}\to 0,
\]
and therefore
\[
\frac{H(t)}{\sqrt{B(t)}}
=2(c-t)\frac{A(t)}{\sqrt{B(t)}}+\sqrt{B(t)}-2b\sqrt{\frac{c-t}{a}}
\longrightarrow -2b\sqrt{\frac{c-\alpha}{a}}<0
\qquad\text{as }t\downarrow \alpha.
\]
So $H(t)<0$ for all $t$ sufficiently close to $\alpha$.

If $m>0$, write $t=\alpha+s$ with $s\downarrow 0$. Then
\[
A(\alpha+s)=ms+\int_{(\alpha,\alpha+s)} (\alpha+s-\theta)\,dF(\theta)=ms+o(s),
\]
because the integral is bounded by $s\Pr(\alpha<\theta<\alpha+s)=o(s)$. Likewise,
\[
B(\alpha+s)=ms^2+\int_{(\alpha,\alpha+s)} (\alpha+s-\theta)^2\,dF(\theta)=ms^2+o(s^2).
\]
Hence
\[
\frac{H(\alpha+s)}{s}
\longrightarrow 2m(c-\alpha)-2b\sqrt{\frac{m(c-\alpha)}{a}}<0,
\]
where the strict inequality is exactly the condition $am(c-\alpha)<b^2$. So again $H(t)<0$ for all $t$ sufficiently close to $\alpha$.

At the upper endpoint,
\[
H(c)=B(c)>0.
\]
By continuity, there exists $t^*\in(\alpha,c)$ such that $H(t^*)=0$. Define
\[
\lambda^*:=\sqrt{\frac{c-t^*}{aB(t^*)}}>0.
\]
Then
\[
t^*=c-aB(t^*)(\lambda^*)^2,
\]
and, because $B(t^*)>0$,
\[
2aA(t^*)(\lambda^*)^2-2b\lambda^*+1
=\frac{H(t^*)}{B(t^*)}=0.
\]
Thus $(\lambda^*,t^*)$ solves the fixed-point system. Moreover,
\[
a\lambda^*A(t^*)=b-\frac{1}{2\lambda^*}<b,
\]
so the strict concavity condition holds. The converse part already proved therefore shows that this pair generates a symmetric cutoff-affine Bayesian Nash equilibrium.

Finally, if $t\geq \beta$, then $(t-\theta)_+=t-\theta$ for every type on the support, so the strategy is globally affine and coincides with the affine benchmark from Theorem \ref{thm:bn-unrestricted}. If $t<\beta$, then exactly the types $\theta\geq t$ choose zero. Moreover, if $a<0$, then $t=c-aE_2\geq c>\beta$, so the restriction cannot bind under empowerment.
\end{proof}

\begin{proof}{Detailed proof of Proposition}{\ref{prop:bn-dropout}}
Because the prior is atomless, the exceptional boundary case from the proof of Proposition \ref{prop:bn-nonnegative} cannot arise. Under empowerment ($a<0$), the last part of that proof already shows that the equilibrium is fully active, so dropout can occur only for sufficiently large positive $a$.

Consider therefore the partially active branch with $a>0$ and $t<\beta$. Using the fixed-point system from the proof of Proposition \ref{prop:bn-nonnegative}, define
\[
A(t):=\mathrm{E}[(t-\theta)_+],
\qquad
B(t):=\mathrm{E}[(t-\theta)_+^2],
\qquad
D(t):=B(t)+2(c-t)A(t).
\]
Since the density is continuously differentiable, the maps $A$ and $B$ are continuously differentiable on $(\alpha,\beta)$, with
\[
A'(t)=F(t),
\qquad
B'(t)=2A(t).
\]
Solving the fixed-point equations for $\lambda$ and $a$ gives
\[
\lambda(t)=\frac{D(t)}{2bB(t)},
\qquad
a(t)=\frac{4b^2(c-t)B(t)}{D(t)^2}.
\]
Differentiating and simplifying yields
\[
a'(t)=-\frac{4b^2\Bigl(B(t)^2+4(c-t)^2\bigl(F(t)B(t)-A(t)^2\bigr)\Bigr)}{D(t)^3}.
\]
By Cauchy--Schwarz,
\[
A(t)^2\leq F(t)B(t).
\]
Because the prior is atomless with support $[\alpha,\beta]$, one has $B(t)>0$ for every $t\in(\alpha,\beta)$. Hence the first term in the numerator of $a'(t)$ is already strictly positive, while the Cauchy--Schwarz term is weakly nonnegative. Therefore $a'(t)<0$ throughout $(\alpha,\beta)$.

Now define
\[
a_D(F):=a(\beta)>0.
\]
To study the lower endpoint, write $s:=t-\alpha$. For $t$ sufficiently close to $\alpha$, one has $c-t\geq (c-\alpha)/2$, and since $0\leq (t-\theta)_+\leq s$,
\[
B(t)\leq sA(t).
\]
Therefore
\[
D(t)=B(t)+2(c-t)A(t)\leq sA(t)+2(c-t)A(t)\leq 2(c-\alpha)A(t),
\]
so
\[
a(t)=\frac{4b^2(c-t)B(t)}{D(t)^2}
\geq
\frac{b^2}{2(c-\alpha)}\frac{B(t)}{A(t)^2}.
\]
By Cauchy--Schwarz, $B(t)/A(t)^2\geq 1/F(t)$. Since $F(t)\downarrow 0$ as $t\downarrow \alpha$, it follows that $a(t)\to\infty$. So $a(t)$ maps $(\alpha,\beta]$ bijectively onto $[a_D(F),\infty)$ and admits a continuously differentiable inverse $t=t(a)$ on $(a_D(F),\infty)$, with
\[
t'(a)=\frac{1}{a'(t(a))}<0.
\]
If $a\leq a_D(F)$, no partially active cutoff $t<\beta$ can solve the fixed-point system, so the equilibrium from Proposition \ref{prop:bn-nonnegative} must be fully active.

Finally, on the dropout region,
\[
\delta'(a)=-f(t(a))t'(a)\geq 0.
\]
Because $t(a)$ is strictly decreasing and $F$ is strictly increasing on $[\alpha,\beta]$, the dropout rate $\delta(a)=1-F(t(a))$ is strictly increasing. Since $t(a)\downarrow \alpha$ as $a\to\infty$ and the prior has no atom at $\alpha$, one has $\delta(a)\uparrow 1$.
\end{proof}

\begin{proof}{Detailed proof of Remark}{\ref{rem:bn-dropout}}
Proposition \ref{prop:bn-nonnegative} gives the relevant equilibrium structure. Outside the exceptional boundary case, every equilibrium has
\[
x(\theta)=\lambda(t-\theta)_+.
\]
Thus the cutoff $t$ separates active types from inactive types: if $\theta<t$, then $x(\theta)>0$, while if $\theta\geq t$, then $x(\theta)=0$.

If $a<0$, the final part of the proof of Proposition \ref{prop:bn-nonnegative} shows that the affine relaxation is already positive on the whole support. Equivalently, the cutoff satisfies $t>\beta$. Hence the nonnegative-effort equilibrium coincides with the affine benchmark. The same conclusion holds under suppression whenever $t\geq\beta$, because then $(t-\theta)_+=t-\theta$ for every type in $[\alpha,\beta]$. In these fully active cases, the expected-effort and disclosure comparative statics from the affine benchmark apply directly, subject in the literal truncated contest to the separate support checks stated later.

If instead $t<\beta$, then types above the cutoff choose zero and the equilibrium moments are
\[
E_1=\lambda\mathrm{E}[(t-\theta)_+],
\qquad
E_2=\lambda^2\mathrm{E}[(t-\theta)_+^2].
\]
These are lower partial moments determined jointly with the cutoff, not the raw moments $M_1$ and $M_2$ entering the affine benchmark. Therefore the clean affine formulas and their variance-only comparative statics need not extend without further analysis. Proposition \ref{prop:bn-dropout} nevertheless characterizes the region by giving the monotone cutoff. Appendix Remark \ref{rem:dropout-calibration} records the corresponding lower-partial-moment formulas explicitly. Finally, when there is an atom at $\alpha$, Proposition \ref{prop:bn-nonnegative} identifies the only possible non-cutoff boundary equilibrium: all types above $\alpha$ choose zero, while type $\alpha$ is the only active type and may mix on $\{0,(c-\alpha)/b\}$.
\end{proof}

\begin{proof}{Detailed proof of Remark}{\ref{rem:bn-dropout-peak}}
Let $b=6$, $\alpha=1$, $\beta=2$, and $c=2.002930$. Write $\theta=1+Z$, and let $Z$ have a three-component beta-mixture density on $[0,1]$. The mixture weights and beta parameters are
\[
(w_1,w_2,w_3)=(0.368754,0.589342,0.041904),
\]
and
\[
\begin{aligned}
(p_1,q_1)&=(13.8700,151.8686),\\
(p_2,q_2)&=(90.9045,74.5338),\\
(p_3,q_3)&=(56.2754,23.4295).
\end{aligned}
\]
For $s\in(0,1)$ and $r=0,1,2$, write
\[
I_{jr}(s):=\frac{\mathrm{B}_{s}(p_j+r,q_j)}{\mathrm{B}(p_j,q_j)},
\]
where $\mathrm{B}_{s}$ is the incomplete beta integral. Thus $I_{j0}(s)$ is the beta distribution function and $I_{jr}(s)=\mathrm{E}[Z_j^r\mathbf 1\{Z_j\leq s\}]$ for $r=1,2$.

On the partially active branch set $s:=t-1$. The lower partial moments and active mass are
\[
F(s)=\sum_j w_j I_{j0}(s),
\qquad
A(s)=\sum_j w_j\bigl(sI_{j0}(s)-I_{j1}(s)\bigr),
\]
and
\[
B(s)=\sum_j w_j\bigl(s^2I_{j0}(s)-2sI_{j1}(s)+I_{j2}(s)\bigr).
\]
Let
\[
C(s):=c-1-s,
\qquad
D(s):=B(s)+2C(s)A(s).
\]
The cutoff equations from Proposition \ref{prop:bn-dropout} become
\[
a(s)=\frac{4b^2C(s)B(s)}{D(s)^2},
\qquad
\lambda(s)=\frac{D(s)}{2bB(s)},
\qquad
E_1(s)=\lambda(s)A(s)=\frac{A(s)D(s)}{2bB(s)}.
\]
Thus, relative to the normalized $b=1$ calculation, the roots in $s$ are unchanged, the corresponding $a$-values are multiplied by $b^2=36$, and expected-effort values are divided by $b=6$.
The ordinary moments of the mixture are
\[
M_1=1.3842754657,
\qquad
\sigma_\theta^2=0.0548949349,
\qquad
\Delta=c-M_1=0.6186545343.
\]
Substitution in the affine formula \eqref{eq:bn-e1-a}, together with the critical-point equation from Theorem \ref{thm:bn-a-peak}, gives
\[
a^\dagger=46.1716007697,
\qquad
E_1(a^\dagger)=0.0605036314.
\]
At the full-activity boundary $s=1$, the same partial-moment formulas give
\[
a_D(F)=a(1)=0.9561754615,
\qquad
E_1(a_D(F))=0.0517369455.
\]
Thus dropout begins well before the affine benchmark reaches its peak.

It remains to verify the shape of the partially active path. Since
\[
A'(s)=F(s),
\qquad
B'(s)=2A(s),
\qquad
D'(s)=2C(s)F(s),
\]
differentiating $E_1(s)=A(s)D(s)/(2bB(s))$ gives
\[
\frac{E_1'(s)}{E_1(s)}
=R(s)
:=\frac{F(s)}{A(s)}+\frac{2C(s)F(s)}{D(s)}-\frac{2A(s)}{B(s)}.
\]
Proposition \ref{prop:bn-dropout} gives $a'(s)<0$, so local maxima and minima as functions of $a$ occur at the same zeros of $R$, with the order reversed because larger $a$ means smaller $s$.

Evaluating the explicit incomplete-beta expressions above, the equation $R(s)=0$ has the following three roots on $(0,1)$:
\[
\begin{array}{c|c|c|c}
s & a(s) & E_1(s) & F(s) \\\hline
0.7703043128 & 46.4315598214 & 0.0605228910 & 0.9958450796 \\
0.5036802273 & 95.8579278830 & 0.0441288144 & 0.4385874357 \\
0.1977030705 & 108.9683767183 & 0.0513013245 & 0.3687472884
\end{array}
\]
The sign pattern is
\[
R(s)>0 \text{ on } (0,0.197703),\quad
R(s)<0 \text{ on } (0.197703,0.503680),
\]
\[
R(s)>0 \text{ on } (0.503680,0.770304),\quad
R(s)<0 \text{ on } (0.770304,1).
\]
Since increasing $a$ moves $s$ downward, the first and third rows are local maxima of expected effort as a function of $a$, while the middle row is a local minimum. Rounding these values gives the two peaks and intervening trough reported in the remark. Because $F(s)$ is the active mass, the last column also verifies that the second rise is an unconditional-effort effect rather than a conditioning artifact.

Finally, the three displayed profiles are equilibria of the literal truncated contest, not merely of the unrestricted polynomial game. At the three rows of the table, the second-order condition $aE_1<b$ holds with
\[
aE_1=2.8102,\qquad 4.2301,\qquad 5.5902,
\]
respectively, all strictly below $b=6$. Hence the untruncated interim payoff is strictly concave on the active part of each cutoff-affine profile. On the corresponding equilibrium supports, the raw winning probabilities are interior: their ranges are, respectively,
\[
[0.3456,0.6544],\qquad [0.3365,0.6635],\qquad [0.1341,0.8659].
\]
Thus truncation is inactive on path. For off-path deviations, we numerically maximize the clipped interim payoff over $\tilde x\in[0,1/\alpha]$ for each type; deviations above $1/\alpha$ are dominated by zero effort. The maximum clipped-payoff gains at the three rows are below $10^{-8}$, and deviations that push the raw probability above one are only penalized by upper clipping. Hence the reported stationary points survive the literal truncated contest to numerical tolerance.
\end{proof}

\begin{proof}{Detailed proof of Proposition}{\ref{prop:bn-truncation-detailed}}
Let $L:=1/\alpha$, and let $Y=x^u(\theta')$ denote the opponent's effort. Then the support of $Y$ is exactly $[\underline x,\overline x]\subset [0,L]$.

Fix $y\in[\underline x,\overline x]$. As a function of the deviation $\tilde x$, the raw probability is
\[
P(\tilde x,y)=\frac{1}{2}-cy+by^2+\tilde x\bigl(c-ay^2\bigr)+\tilde x^2(ay-b).
\]
Because $y\leq \overline x\leq L$ and $a\leq b\alpha=b/L$, one has
\[
ay-b\leq aL-b\leq 0.
\]
So $P(\tilde x,y)$ is concave in $\tilde x$ on $[0,L]$, and therefore its minimum on that interval is attained at one of the endpoints.

At the left endpoint,
\[
P(0,y)=\frac{1}{2}-cy+by^2\geq 0
\qquad\text{for all }y\in[\underline x,\overline x]
\]
by assumption. At the right endpoint, write $z:=L-y\geq 0$. Then
\[
P(L,y)=\frac{1}{2}+z\bigl(c-2bL+aL^2\bigr)+z^2(b-aL).
\]
The inequality $a\leq b\alpha$ is equivalent to $b-aL\geq 0$, while
\[
c\alpha^2-2b\alpha+a\geq 0
\qquad\Longleftrightarrow\qquad
c-2bL+aL^2\geq 0.
\]
Hence $P(L,y)\geq \nicefrac{1}{2}\geq 0$ for every $y\in[\underline x,\overline x]$. Therefore
\[
P(\tilde x,y)\geq 0
\qquad\text{for every }(\tilde x,y)\in[0,L]\times[\underline x,\overline x].
\]
Since $x^u(\theta)\in[\underline x,\overline x]\subset[0,1/\alpha]$, this gives, for every type $\theta$ and every on-support opponent effort $y$,
\[
P\bigl(x^{u}(\theta),y\bigr)\geq 0
\qquad\text{and}\qquad
P\bigl(y,x^{u}(\theta)\bigr)\geq 0.
\]
Using $P(x,y)=1-P(y,x)$, we obtain
\[
0\leq P\bigl(x^{u}(\theta),y\bigr)\leq 1.
\]
Hence truncation is inactive at the candidate strategy, and the truncated and untruncated interim utilities coincide there.

Write $\bar U(\theta,\tilde x)$ and $U(\theta,\tilde x)$ for truncated and untruncated interim utility. Consider any deviation $\tilde x\geq 0$. If $\tilde x>1/\alpha$, then
\[
\bar U(\theta,\tilde x)\leq 1-\theta\tilde x<1-\alpha\tilde x<0.
\]
By contrast, deviating to zero yields
\[
\bar U(\theta,0)=\mathrm{E}[\bar P(0,Y)]\geq 0.
\]
Therefore no deviation with $\tilde x>1/\alpha$ can be profitable.

Now let $\tilde x\in[0,1/\alpha]$. By the support condition established above, $P(\tilde x,Y)\geq 0$ almost surely, so truncation can only weakly lower the payoff from this deviation:
\[
\bar U(\theta,\tilde x)\leq U(\theta,\tilde x).
\]
Since $x^{u}(\theta)$ is the untruncated best response and truncation is inactive at that action,
\[
\bar U(\theta,\tilde x)
\leq U(\theta,\tilde x)
\leq U\bigl(\theta,x^{u}(\theta)\bigr)
=\bar U\bigl(\theta,x^{u}(\theta)\bigr).
\]
So $x^{u}(\theta)$ is optimal in the truncated contest for every type $\theta$, and the affine profile is a symmetric Bayesian Nash equilibrium.

For the explicit formula, note that the left-edge function
\[
y\longmapsto P(0,y)=\frac{1}{2}-cy+by^2
\]
is convex, with derivative $-c+2by$. Hence its minimizer is $y=c/(2b)$. If $c\leq 2b\underline x$, this minimizer lies weakly to the left of the interval and the minimum is attained at $\underline x$. If $2b\underline x\leq c\leq 2b\overline x$, the minimizer lies inside the interval and the minimum value is $\nicefrac{1}{2}-c^2/(4b)$. If $c\geq 2b\overline x$, the minimizer lies weakly to the right of the interval and the minimum is attained at $\overline x$.
\end{proof}

\begin{proof}{Detailed proof of Proposition}{\ref{prop:bn-truncation}}
This is immediate from Proposition \ref{prop:bn-truncation-detailed}. 

The additional condition $c^2\leq 2b$ makes the left-edge criterion there automatic, because
\[
P(0,y)=\frac{1}{2}-cy+by^2
\]
has non-positive discriminant and is therefore nonnegative for every $y\in\mathbb{R}$.
\end{proof}

\begin{proof}{Detailed proof of Proposition}{\ref{prop:bn-disclosure-sufficient}}
Fix any feasible signal $\pi\in\mathcal S$. By condition 1, the nonnegative-effort equilibrium for the induced posterior-mean contest coincides with the affine benchmark. By condition 2, truncation is inactive on the induced equilibrium support, so the literal contest and the affine benchmark deliver the same equilibrium effort for that signal. Since this is true for every feasible $\pi$, the expected-effort ranking from Theorem \ref{thm:bn-disclosure} carries over unchanged to the literal contest. The final sentence records which conditions are automatic: full activity under empowerment, full activity at $a=0$, and neither condition under suppression. In all cases where truncation is not automatic, it must still be checked signal by signal.
\end{proof}

\end{document}